\documentclass
[11pt,letterpaper]
{article}
\usepackage{macros}

\usetikzlibrary{shapes.geometric, quotes}
\tikzset{
    squarenode/.style = {draw, regular polygon,regular polygon sides=4, inner sep=0, minimum size=24pt, scale=0.8},
    circlenode/.style = {draw, circle, minimum size=20pt, scale=0.8},
    leftvertices/.style = {draw, color=blue},
    rightvertices/.style = {draw, color=red},
    every edge quotes/.append style={font=\tiny},
    separator/.style = {fill=green!30},
}

\newcommand{\Hyperedge}[6]{
    \path (#1) edge (#4,#5);
    \path (#2) edge (#4,#5);
    \path (#4,#5) edge (#3);
    \node at (#4,#5+0.2) [scale=0.5] {$#6$};
}
\newcommand{\ParallelEdge}[5]{
    \path (#1) edge [bend right] (#3,#4);
    \path (#1) edge [bend left] (#3,#4);
    \path (#3,#4) edge (#2);
    \node at (#3,#4+0.2) [scale=0.5] {$#5$};
}

\newcommand{\LeftVertices}[4]{
    \draw
        node at (#1,#2) [leftvertices, minimum width=24, minimum height=#3*24] (left) {}
        node at (#1-0.8,#2) [color=blue, scale=0.8] {#4}
    ;
}
\newcommand{\LeftVerticesOnRight}[4]{
    \draw
        node at (#1,#2) [leftvertices, minimum width=24, minimum height=#3*24] (left) {}
        node at (#1+0.8,#2) [color=blue, scale=0.8] {#4}
    ;
}
\newcommand{\RightVertices}[4]{
    \draw
        node at (#1,#2) [rightvertices, minimum width=24, minimum height=#3*24] (right) {}
        node at (#1-0.8,#2) [color=red, scale=0.8] {#4}
    ;
}
\newcommand{\RightVerticesOnRight}[4]{
    \draw
        node at (#1,#2) [rightvertices, minimum width=24, minimum height=#3*24] (right) {}
        node at (#1+0.8,#2) [color=red, scale=0.8] {#4}
    ;
}
\newcommand{\LeftRightIntersection}[4]{
    \draw
        node at (#1,#2) [leftvertices, minimum width=24, minimum height=#3*24] (left) {}
        node at (#1,#2) [rightvertices, minimum width=28, minimum height=#3*24+4] (right) {}
        node at (#1-1.2,#2) [color=Purple, scale=0.8] {#4}
    ;
}
\newcommand{\VerticalSep}{-1}
\newcommand{\HorizontalSep}{1.5}

\newcommand{\FigTrivialShapes}{
    \begin{figure}[ht]
        \centering
        \begin{subfigure}[b]{0.4\textwidth}
            \centering
            \begin{tikzpicture}
            \draw
                node at (0,0) [squarenode,separator] (i1) {$i_1$}
            ;
            \LeftRightIntersection{0}{0}{1}{$U_a\cap V_a$}
            \end{tikzpicture}
            \caption{$|\lambda_a| = n^{-1}$, $\|M_a\| = 1$.}
        \end{subfigure}
        \quad
        \begin{subfigure}[b]{0.4\textwidth}
            \centering
            \begin{tikzpicture}
            \draw
                node at (0,0) [squarenode,separator] (i1) {$i_1$}
                node at (0,\VerticalSep) [squarenode,separator] (i2) {$i_2$}
            ;
            \LeftRightIntersection{0}{\VerticalSep/2}{2}{$U_a\cap V_a$}
            \end{tikzpicture}
            \caption{$|\lambda_a| = n^{-2}$, $\|M_a\| = 1$.}
        \end{subfigure}
        \caption{Trivial shapes: $U_a = V_a$, $E(a) = \varnothing$, and $M_a = \Id$.}
        \label{fig:trivial_shapes}
    \end{figure}
}

\newcommand{\FigNegligibleShapes}{
    \begin{figure}[ht]
        \centering
        \begin{subfigure}[b]{0.45\textwidth}
            \centering
            \begin{tikzpicture}
                \draw
                    node at (0,0) [squarenode] (i) {$i$}
                    node at (0,\VerticalSep) [squarenode,separator] (j) {$j$}
                    node at (0,2*\VerticalSep) [squarenode] (k) {$k$}
                    node at (\HorizontalSep,0.5*\VerticalSep) [circlenode] (s1) {$\s_1$}
                    node at (\HorizontalSep,1.5*\VerticalSep) [circlenode] (s2) {$\s_1$}
                    node at (2*\HorizontalSep,0.5*\VerticalSep) [squarenode] (l1) {$l_1$}
                    node at (2*\HorizontalSep,1.5*\VerticalSep) [squarenode] (l2) {$l_2$}
                ;

                \Hyperedge{i}{j}{s1}{\HorizontalSep/2}{\VerticalSep/2}{}
                \Hyperedge{j}{k}{s2}{\HorizontalSep/2}{\VerticalSep*1.5}{}
                \ParallelEdge{l1}{s1}{1.5*\HorizontalSep}{0.5*\VerticalSep}{}
                \ParallelEdge{l2}{s2}{1.5*\HorizontalSep}{1.5*\VerticalSep}{}
                \LeftVertices{0}{0}{1}{$U_a$}
                \RightVertices{0}{2*\VerticalSep}{1}{$V_a$}
            \end{tikzpicture}
            \caption{$\varphi(a) = \eps = \frac{\eps}{4}|E(a)|$.}
        \end{subfigure}
        \quad
        \begin{subfigure}[b]{0.45\textwidth}
            \centering
            \begin{tikzpicture}
                \draw
                    node at (0,0) [squarenode] (i1) {$i_1$}
                    node at (0,\VerticalSep) [squarenode] (k) {$k$}
                    node at (0,2*\VerticalSep) [squarenode] (i2) {$i_2$}
                    node at (\HorizontalSep,\VerticalSep) [circlenode,separator] (s) {$\s$}
                    node at (2*\HorizontalSep,0) [squarenode] (j1) {$j_1$}
                    node at (2*\HorizontalSep,\VerticalSep) [squarenode] (l) {$l$}
                    node at (2*\HorizontalSep,2*\VerticalSep) [squarenode] (j2) {$j_2$}
                ;

                \Hyperedge{i1}{k}{s}{\HorizontalSep/2}{\VerticalSep/2}{}
                \Hyperedge{i2}{k}{s}{\HorizontalSep/2}{\VerticalSep*1.5}{}
                \Hyperedge{j1}{l}{s}{1.5*\HorizontalSep}{0.5*\VerticalSep}{}
                \Hyperedge{j2}{l}{s}{1.5*\HorizontalSep}{1.5*\VerticalSep}{}
                \LeftVertices{0}{0}{1}{$U_a$}
                \LeftVertices{0}{2*\VerticalSep}{1}{$U_a$}
                \RightVerticesOnRight{2*\HorizontalSep}{0}{1}{$V_a$}
                \RightVerticesOnRight{2*\HorizontalSep}{2*\VerticalSep}{1}{$V_a$}
            \end{tikzpicture}
            \caption{$\varphi(a) = 1$.}
        \end{subfigure}
        \caption{Negligible shapes.}
        \label{fig:negligible_shapes}
    \end{figure}
}

\newcommand{\FigDisconnectedShapes}{
    \begin{figure}[ht]
        \centering
        \scalebox{0.75}{
        \begin{subfigure}[b]{0.95\textwidth}
            \centering
            \begin{tikzpicture}
                \draw
                    node at (0,0) [squarenode] (i1) {$i_1$}
                    node at (0,\VerticalSep) [squarenode] (k) {$k$}
                    node at (0,2*\VerticalSep) [squarenode] (i2) {$i_2$}
                    node at (\HorizontalSep,\VerticalSep) [circlenode] (s1) {$\s_1$}

                    node at (1.5*\HorizontalSep, \VerticalSep) () {$\times$}
                    node at (2*\HorizontalSep,\VerticalSep) [circlenode] (s2) {$\s_2$}
                    node at (3*\HorizontalSep,0) [squarenode] (j1) {$j_1$}
                    node at (3*\HorizontalSep,\VerticalSep) [squarenode] (j2) {$j_2$}
                    node at (3*\HorizontalSep,2*\VerticalSep) [squarenode] (l) {$\ell$}
                    node at (4*\HorizontalSep, \VerticalSep) () {$=$}
                ;
                \Hyperedge{i1}{k}{s1}{0.5*\HorizontalSep}{0.5*\VerticalSep}{}
                \Hyperedge{i2}{k}{s1}{0.5*\HorizontalSep}{1.5*\VerticalSep}{}
                \Hyperedge{j1}{j2}{s2}{2.5*\HorizontalSep}{0.5*\VerticalSep}{}
                \ParallelEdge{l}{s2}{2.5*\HorizontalSep}{1.5*\VerticalSep}{}
                \LeftVertices{0}{0}{1}{$U_{a_1}$}
                \LeftVertices{0}{2*\VerticalSep}{1}{$U_{a_1}$}
                \RightVerticesOnRight{3*\HorizontalSep}{0.5*\VerticalSep}{2}{$V_{a_2}$}
            \end{tikzpicture}
        \end{subfigure}
        }

        \vspace{1.5em}
        \begin{subfigure}[b]{\textwidth}
            \centering
            \resizebox{0.9\textwidth}{!}{
            \begin{tikzpicture}
                \draw
                    node at (0,0) [squarenode] (i1) {$i_1$}
                    node at (0,\VerticalSep) [squarenode] (k) {$k$}
                    node at (0,2*\VerticalSep) [squarenode] (i2) {$i_2$}
                    node at (\HorizontalSep,\VerticalSep) [circlenode] (s1) {$\s_1$}

                    node at (2*\HorizontalSep,\VerticalSep) [circlenode] (s2) {$\s_2$}
                    node at (3*\HorizontalSep,0) [squarenode] (j1) {$j_1$}
                    node at (3*\HorizontalSep,\VerticalSep) [squarenode] (j2) {$j_2$}
                    node at (3*\HorizontalSep,2*\VerticalSep) [squarenode] (l) {$\ell$}
                    node at (4*\HorizontalSep, \VerticalSep) () {$+$}
                ;
                \Hyperedge{i1}{k}{s1}{\HorizontalSep/2}{0.5*\VerticalSep}{}
                \Hyperedge{i2}{k}{s1}{\HorizontalSep/2}{1.5*\VerticalSep}{}
                \Hyperedge{j1}{j2}{s2}{2.5*\HorizontalSep}{0.5*\VerticalSep}{}
                \ParallelEdge{l}{s2}{2.5*\HorizontalSep}{1.5*\VerticalSep}{}
                \LeftVertices{0}{0}{1}{$U_a$}
                \LeftVertices{0}{2*\VerticalSep}{1}{$U_a$}
                \RightVerticesOnRight{3*\HorizontalSep}{0.5*\VerticalSep}{2}{$V_a$}

                \draw
                    node at (5*\HorizontalSep,0) [squarenode]  (i21) {$i_1$}
                    node at (5*\HorizontalSep,\VerticalSep) [squarenode] (k2) {$k$}
                    node at (5*\HorizontalSep,2*\VerticalSep) [squarenode] (i22) {$i_2$}
                    node at (6*\HorizontalSep,\VerticalSep) [circlenode,separator] (s21) {$\s$}
                    node at (7*\HorizontalSep,0) [squarenode]  (j21) {$j_1$}
                    node at (7*\HorizontalSep,\VerticalSep) [squarenode] (j22) {$j_2$}
                    node at (7*\HorizontalSep,2*\VerticalSep) [squarenode] (l2) {$\ell$}
                    node at (8.5*\HorizontalSep, \VerticalSep) () {$+\ \cdots\ +$}
                ;
                \Hyperedge{i21}{k2}{s21}{5.5*\HorizontalSep}{0.5*\VerticalSep}{}
                \Hyperedge{i22}{k2}{s21}{5.5*\HorizontalSep}{1.5*\VerticalSep}{}
                \LeftVertices{5*\HorizontalSep}{0}{1}{$U_a$}
                \LeftVertices{5*\HorizontalSep}{2*\VerticalSep}{1}{$U_a$}
                \Hyperedge{j21}{j22}{s21}{6.5*\HorizontalSep}{0.5*\VerticalSep}{}
                \ParallelEdge{l2}{s21}{6.5*\HorizontalSep}{1.5*\VerticalSep}{}
                \RightVerticesOnRight{7*\HorizontalSep}{0.5*\VerticalSep}{2}{$V_a$}

                \draw
                    node at (11*\HorizontalSep,0) [squarenode,separator] (i3) {$i$}
                    node at (10*\HorizontalSep,\VerticalSep) [squarenode] (k3) {$k$}
                    node at (10*\HorizontalSep,2*\VerticalSep) [squarenode,separator] (i32) {$i_2$}
                    node at (11*\HorizontalSep,1.5*\VerticalSep) [circlenode] (s31) {$\s$}
                    node at (12*\HorizontalSep,\VerticalSep) [squarenode] (j32) {$j_2$}
                    node at (12*\HorizontalSep,2*\VerticalSep) [squarenode] (l3) {$\ell$}
                    node at (13.5*\HorizontalSep, \VerticalSep) () {$+\ \cdots$}
                ;
                \Hyperedge{i3}{k3}{s31}{10.7*\HorizontalSep}{\VerticalSep}{}
                \Hyperedge{i32}{k3}{s31}{10.5*\HorizontalSep}{1.5*\VerticalSep}{}
                \Hyperedge{i3}{j32}{s31}{11.3*\HorizontalSep}{\VerticalSep}{}
                \Hyperedge{j32}{l3}{s31}{11.5*\HorizontalSep}{1.5*\VerticalSep}{}
                \LeftRightIntersection{11*\HorizontalSep}{0}{1}{$U_a \cap V_a$}
                \LeftVertices{10*\HorizontalSep}{2*\VerticalSep}{1}{$U_a$}
                \RightVerticesOnRight{12*\HorizontalSep}{\VerticalSep}{1}{$V_a$}
            \end{tikzpicture}
            }
        \end{subfigure}
        \caption{Shapes introduced by multiplying $M_{a_1} M_{a_2}^\top$.  The first is a disconnected shape; the second is by collapsing $\s_1, \s_2$ (setting $\s_1=\s_2$);
        the third is by collapsing $i_1, j_1$ and $\s_1,\s_2$ (setting $\s_1=\s_2$ and $i_1 = j_1$).
        }
        \label{fig:disconnected_shapes}
    \end{figure}
}

\newcommand{\FigSpider}{
    \begin{figure}[ht]
        \centering
        \resizebox{0.4\textwidth}{!}{
        \begin{tikzpicture}
            \draw
                node at (0,0) [squarenode] (i1) {$i_1$}
                node at (0,\VerticalSep) [squarenode] (i2) {$i_2$}
                node at (\HorizontalSep,\VerticalSep/2) [circlenode, separator] (s) {$\s$}
                node at (2*\HorizontalSep,0) [squarenode] (j1) {$j_1$}
                node at (2*\HorizontalSep,\VerticalSep) [squarenode] (j2) {$j_2$}
            ;
            \Hyperedge{i1}{i2}{s}{\HorizontalSep/2}{\VerticalSep/2}{}
            \Hyperedge{j1}{j2}{s}{\HorizontalSep*1.5}{\VerticalSep/2}{}
            \LeftVertices{0}{\VerticalSep/2}{2}{$U_a$}
            \RightVerticesOnRight{2*\HorizontalSep}{\VerticalSep/2}{2}{$V_a$}
            \end{tikzpicture}
            }
        \caption{Spider $\spider$.}
        \label{fig:spider}
    \end{figure}
}

\newcommand{\FigLTL}{
    \begin{figure}[ht]
        \centering
        \scalebox{0.8}{
        \begin{subfigure}[b]{0.95\textwidth}
            \centering
            \begin{tikzpicture}
                \draw
                    node at (0,0) [squarenode] (i1) {$i_1$}
                    node at (0,\VerticalSep) [squarenode] (i2) {$i_2$}
                    node at (\HorizontalSep,0.5*\VerticalSep) [circlenode] (s1) {$\s$}

                    node at (2*\HorizontalSep, 0.5*\VerticalSep) () {$\times$}
                    node at (3*\HorizontalSep,0.5*\VerticalSep) [circlenode] (s2) {$\s$}
                    node at (4*\HorizontalSep,0) [squarenode] (j1) {$j_1$}
                    node at (4*\HorizontalSep,\VerticalSep) [squarenode] (j2) {$j_2$}
                    node at (5*\HorizontalSep, 0.5*\VerticalSep) () {$=$}
                ;
                \Hyperedge{i1}{i2}{s1}{0.5*\HorizontalSep}{0.5*\VerticalSep}{}
                \Hyperedge{j1}{j2}{s2}{3.5*\HorizontalSep}{0.5*\VerticalSep}{}
                \LeftVertices{0}{0.5*\VerticalSep}{2}{$U_{a_1^\top}$}
                \RightVerticesOnRight{\HorizontalSep}{0.5*\VerticalSep}{1}{$V_{a_1^\top}$}
                \LeftVertices{3*\HorizontalSep}{0.5*\VerticalSep}{1}{$U_{a_1}$}
                \RightVerticesOnRight{4*\HorizontalSep}{0.5*\VerticalSep}{2}{$V_{a_1}$}
            \end{tikzpicture}
        \end{subfigure}
        }

        \vspace{1.5em}
        \begin{subfigure}[b]{\textwidth}
            \centering
            \resizebox{0.95\textwidth}{!}{
            \begin{tikzpicture}
                \draw
                    node at (0,0) [squarenode] (i1) {$i_1$}
                    node at (0,\VerticalSep) [squarenode] (i2) {$i_2$}
                    node at (\HorizontalSep,0.5*\VerticalSep) [circlenode,separator] (s1) {$\s$}
                    node at (2*\HorizontalSep,0) [squarenode] (j1) {$j_1$}
                    node at (2*\HorizontalSep,\VerticalSep) [squarenode] (j2) {$j_2$}
                    node at (3*\HorizontalSep, 0.5*\VerticalSep) () {$+$}
                ;
                \Hyperedge{i1}{i2}{s1}{\HorizontalSep/2}{0.5*\VerticalSep}{}
                \Hyperedge{j1}{j2}{s1}{1.5*\HorizontalSep}{0.5*\VerticalSep}{}
                \LeftVertices{0}{0.5*\VerticalSep}{2}{$U_a$}
                \RightVerticesOnRight{2*\HorizontalSep}{0.5*\VerticalSep}{2}{$V_a$}

                \def \shift {4.5*\HorizontalSep}
                \draw
                    node at (\shift,0) [squarenode,separator]  (i21) {$i_1$}
                    node at (\shift,\VerticalSep) [squarenode,separator] (i22) {$i_2$}
                    node at (\shift,2*\VerticalSep) [squarenode] (j22) {$j_2$}
                    node at (\shift+\HorizontalSep,\VerticalSep) [circlenode] (s21) {$\s$}
                    node at (\shift+1.5*\HorizontalSep, 0.5*\VerticalSep) () {$+$}
                ;
                \Hyperedge{i21}{i22}{s21}{\shift+0.5*\HorizontalSep}{0.5*\VerticalSep}{}
                \Hyperedge{i22}{j22}{s21}{\shift+0.5*\HorizontalSep}{1.5*\VerticalSep}{}
                \LeftVertices{\shift}{0}{1}{$U_a$}
                \RightVertices{\shift}{2*\VerticalSep}{1}{$V_a$}
                \LeftRightIntersection{\shift}{\VerticalSep}{1}{$U_a\cap V_a$}

                \def \shift {7.5*\HorizontalSep}
                \draw
                    node at (\shift,0) [squarenode,separator] (i31) {$i_1$}
                    node at (\shift,\VerticalSep) [squarenode,separator] (i32) {$i_2$}
                    node at (\shift+\HorizontalSep,0.5*\VerticalSep) [circlenode] (s31) {$\s$}
                    node at (\shift+1.5*\HorizontalSep, 0.5*\VerticalSep) () {$+$}
                ;
                \Hyperedge{i31}{i32}{s31}{\shift+0.5*\HorizontalSep}{0.5*\VerticalSep}{2}
                \LeftRightIntersection{\shift}{0.5*\VerticalSep}{2}{$U_a \cap V_a$}

                \def \shift {10.5*\HorizontalSep}
                \draw
                    node at (\shift,0) [squarenode,separator] (i31) {$i_1$}
                    node at (\shift,\VerticalSep) [squarenode,separator] (i32) {$i_2$}
                    node at (\shift+\HorizontalSep,0.5*\VerticalSep) [circlenode] (s31) {$\s$}
                ;
                \LeftRightIntersection{\shift}{0.5*\VerticalSep}{2}{$U_a \cap V_a$}
            \end{tikzpicture}
            }
        \end{subfigure}
        \caption{Expansion of $M_{a_1}^\top M_{a_1}$.}
        \label{fig:LTL}
    \end{figure}
}

\newcommand{\FigDegTwoL}{
    \begin{figure}[ht]
        \centering
        \begin{subfigure}[b]{\textwidth}
            \centering
            \resizebox{0.5\textwidth}{!}{
            \begin{tikzpicture}
                \draw
                    node at (0,\VerticalSep) [circlenode,separator] (s) {$\s$}
                    node at (\HorizontalSep,0.5*\VerticalSep) [squarenode] (i1) {$i_1$}
                    node at (\HorizontalSep,1.5*\VerticalSep) [squarenode] (i2) {$i_2$}
                    node at (2*\HorizontalSep, \VerticalSep) () {$+$}
                ;
                \Hyperedge{i1}{i2}{s}{0.5*\HorizontalSep}{\VerticalSep}{}
                \LeftVertices{0}{\VerticalSep}{1}{$U_{a_1}$}
                \RightVerticesOnRight{\HorizontalSep}{\VerticalSep}{2}{$V_{a_1}$}

                \draw
                    node at (2.5*\HorizontalSep, \VerticalSep) () {$\frac{1}{n}\ \cdot$}
                    node at (3.5*\HorizontalSep,\VerticalSep) [circlenode] (s2) {$\s$}
                    node at (4.5*\HorizontalSep,\VerticalSep) [squarenode]  (i3) {$i$}
                ;
                \ParallelEdge{i3}{s2}{4*\HorizontalSep}{\VerticalSep}{}
                \LeftVertices{3.5*\HorizontalSep}{\VerticalSep}{1}{$U_{a_2}$}
            \end{tikzpicture}
            }
        \end{subfigure}
        \caption{$L_2 = M_{a_1} + \frac{1}{n}M_{a_2}$.}
        \label{fig:L2_shapes}
    \end{figure}
}

\newcommand{\FigDegFourL}{
    \begin{figure}[ht]
        \centering
        \begin{subfigure}[b]{\textwidth}
            \centering
            \resizebox{0.8\textwidth}{!}{
            \begin{tikzpicture}
                \draw
                    node at (0,\VerticalSep) [circlenode,separator] (s) {$\s$}
                    node at (0,2*\VerticalSep) [squarenode,separator] () {$k$}
                    node at (0,3*\VerticalSep) [squarenode,separator] () {$\ell$}
                    node at (\HorizontalSep,0.5*\VerticalSep) [squarenode] (i) {$i$}
                    node at (\HorizontalSep,1.5*\VerticalSep) [squarenode] (j) {$j$}
                    node at (2*\HorizontalSep, \VerticalSep) () {$+$}
                ;
                \Hyperedge{i}{j}{s}{0.5*\HorizontalSep}{\VerticalSep}{}
                \LeftVertices{0}{\VerticalSep}{1}{$U_a$}
                \RightVerticesOnRight{\HorizontalSep}{\VerticalSep}{2}{$V_a$}
                \LeftRightIntersection{0}{2.5*\VerticalSep}{2}{$U_a \cap V_a$}

                \draw
                    node at (2.5*\HorizontalSep, \VerticalSep) () {$\frac{1}{n}\ \cdot$}
                    node at (3.5*\HorizontalSep,\VerticalSep) [circlenode] (s2) {$\s$}
                    node at (3.5*\HorizontalSep,2*\VerticalSep) [squarenode,separator]  () {$\ell$}
                    node at (4.5*\HorizontalSep,0.5*\VerticalSep) [squarenode]  (i2) {$k$}
                    node at (4.5*\HorizontalSep,1.5*\VerticalSep) [squarenode,separator]  (j2) {$j$}
                    node at (5.5*\HorizontalSep, \VerticalSep) () {$+$}
                ;
                \Hyperedge{i2}{j2}{s2}{4*\HorizontalSep}{\VerticalSep}{}
                \LeftVertices{3.5*\HorizontalSep}{\VerticalSep}{1}{$U_a$}
                \LeftVerticesOnRight{4.5*\HorizontalSep}{0.5*\VerticalSep}{1}{$U_a$}
                \LeftRightIntersection{3.5*\HorizontalSep}{2*\VerticalSep}{1}{$U_a\cap V_a$}
                \RightVerticesOnRight{4.5*\HorizontalSep}{1.5*\VerticalSep}{1}{$V_a$}

                \draw
                    node at (6*\HorizontalSep, \VerticalSep) () {$\frac{1}{n}\ \cdot$}
                    node at (7*\HorizontalSep,\VerticalSep) [circlenode] (s3) {$\s$}
                    node at (7*\HorizontalSep,2*\VerticalSep) [squarenode,separator] () {$k$}
                    node at (7*\HorizontalSep,3*\VerticalSep) [squarenode,separator] () {$\ell$}
                    node at (8*\HorizontalSep,\VerticalSep) [squarenode] (i3) {$i$}
                    node at (9*\HorizontalSep, \VerticalSep) () {$+\ \cdots$}
                ;
                \ParallelEdge{i3}{s3}{7.5*\HorizontalSep}{\VerticalSep}{}
                \LeftVertices{7*\HorizontalSep}{\VerticalSep}{1}{$U_a$}
                \LeftRightIntersection{7*\HorizontalSep}{2.5*\VerticalSep}{2}{$U_a \cap V_a$}
            \end{tikzpicture}
            }
        \end{subfigure}
        \caption{Shapes in $L_4$.}
        \label{fig:L4_shapes}
    \end{figure}
}

\newcommand{\FigSpecialShape}{
    \begin{figure}[ht]
        \centering
        \begin{subfigure}[b]{\textwidth}
            \centering
            \resizebox{0.22\textwidth}{!}{
            \begin{tikzpicture}
                \draw
                    node at (0,0) [circlenode,separator] (s1) {$\s_1$}
                    node at (0,\VerticalSep) [squarenode] (i1) {$i_1$}
                    node at (0,2*\VerticalSep) [squarenode] (j1) {$j_1$}
                    node at (\HorizontalSep,0) [squarenode] (i2) {$i_2$}
                    node at (\HorizontalSep,1*\VerticalSep) [squarenode] (j2) {$j_2$}
                    node at (\HorizontalSep,2*\VerticalSep) [circlenode,separator] (s2) {$\s_2$}
                    node at (0,\VerticalSep) [leftvertices, minimum width=24, minimum height=80] (right) {}
                    node at (-0.8,\VerticalSep) [color=blue, scale=0.8] {$U_a$}
                    node at (\HorizontalSep,\VerticalSep) [rightvertices, minimum width=24, minimum height=80] (left) {}
                    node at (\HorizontalSep+0.8,\VerticalSep) [color=red, scale=0.8] {$V_a$}
                ;
                \Hyperedge{i2}{j2}{s1}{0.5*\HorizontalSep}{0.5*\VerticalSep}{}
                \Hyperedge{i1}{j1}{s2}{0.5*\HorizontalSep}{1.5*\VerticalSep}{}
            \end{tikzpicture}
            }
        \end{subfigure}
        \caption{Special shape $a^*$. $\|M_{a^*}\| = \wt{O}(n^2)$.}
        \label{fig:special_shape}
    \end{figure}
}

\title{Algorithmic Thresholds for Refuting Random Polynomial Systems}
\author{Jun-Ting Hsieh\thanks{Carnegie Mellon University, Supported by NSF CAREER Award \#2047933.}
\and Pravesh K. Kothari\thanksmark{1}
}
\date{\today}

\begin{document}
\maketitle

\begin{abstract}
Consider a system of $m$ polynomial equations $\{p_i(x) = b_i\}_{i \leq m}$ of degree $D\geq 2$ in $n$-dimensional variable $x \in \mathbb{R}^n$ such that each coefficient of every $p_i$ and $b_i$s are chosen at random and independently from some continuous distribution. We study the basic question of determining the smallest $m$ -- the \emph{algorithmic threshold} -- for which efficient algorithms can find \emph{refutations} (i.e.\ certificates of unsatisfiability) for such systems. This setting generalizes problems such as refuting random SAT instances, low-rank matrix sensing and certifying pseudo-randomness of Goldreich's candidate generators and generalizations. 


We show that for every $d \in \N$, the $(n+m)^{O(d)}$-time canonical sum-of-squares (SoS) relaxation refutes such a system with high probability whenever $m \geq O(n) \cdot (\frac{n}{d})^{D-1}$. We prove a lower bound in the restricted \emph{low-degree polynomial} model of computation which suggests that this trade-off between SoS degree and the number of equations is nearly tight for all $d$. We also confirm the predictions of this lower bound in a limited setting by showing a lower bound on the canonical degree-$4$ sum-of-squares relaxation for refuting random quadratic polynomials. Together, our results provide evidence for an algorithmic threshold for the problem at $m \gtrsim \widetilde{O}(n) \cdot n^{(1-\delta)(D-1)}$ for $2^{n^{\delta}}$-time algorithms for all $\delta$.

Our upper-bound relies on establishing a sharp bound on the smallest integer $d$ such that degree $d-D$ polynomial combinations of the input $p_i$s generate all degree-$d$ polynomials in the ideal generated by the $p_i$s. Our lower bound actually holds for the easier problem of distinguishing random polynomial systems as above from a distribution on polynomial systems with a ``planted'' solution. Our choice of planted distribution is slightly (and necessarily) subtle: it turns out that the natural and well-studied planted distribution for quadratic systems (studied as the \emph{matrix sensing} problem in machine learning) is easily distinguishable whenever $m\geq \widetilde{O}(n)$ -- a factor $n$ smaller than the threshold in our upper bound above. Thus, our setting provides an example where refutation is harder than search in the natural planted model.

\end{abstract}

\thispagestyle{empty}
\setcounter{page}{0}
\newpage
\tableofcontents
\thispagestyle{empty}
\setcounter{page}{0}
\newpage

\section{Introduction}
Suppose you are given a system of polynomial equations $\{p_i(x)=b_i\}_{i \leq m}$ where each $p_i$ is a homogeneous polynomial of degree $D$ and each $b_i$ and each coefficient of $p_i$ are independent standard Gaussians. When $m \geq n+1$, an elementary argument\footnote{The classical Bezout's theorem says that the number of common complex zeros of $n$ ``generic'' polynomials -- this condition holds with probability 1 whenever the coefficients are independently drawn from a continuous distribution --  of degree $\leq d$ in $n$ variables is at most $d^n$. Apply this to the first $n$ equations and observe that the chance that the $n+1$-th polynomial has a zero at any of the $\leq d^n$ common zeros of the rest is $0$. Via more sophisticated arguments (e.g.~\cite{MR3106890}), such a result can be extended to random polynomial systems with coefficients chosen from discrete distributions.} shows that the system has no real (or complex!) solution with probability $1$. In this work, we study the problem of finding \emph{refutations} -- certificates of unsatisfiability -- for such random polynomial systems with $m \geq n+1$. 

When $D=1$, the classical Gauss-Jordan elimination for solving linear systems efficiently produces a refutation whenever $m \geq n+1$. When $D\geq 2$, the problem is NP-hard in the worst-case (it encodes Max-Cut) and the setting above is a (and perhaps, ``the'') natural average-case formulation. As $m$ increases, finding refutations gets easier and indeed when $m \geq N_D = \sum_{i \leq D} \binom{n+i-1}{i} = \Omega(n^D)$ (we call this the \emph{linearization threshold}), one can simply ``linearize'' the polynomials and apply Gauss-Jordan elimination to linear functions on $N_D$ variables to obtain a refutation algorithm. 

Can efficient algorithms refute random polynomial systems below the linearization threshold? More generally, what's the \emph{algorithmic threshold} -- i.e. the smallest $m=m(n,D)$ -- such that efficient algorithms can (with high probability) come up with efficiently verifiable certificates of unsatisfiability of random polynomial systems with $m$ equations?

Let's cut to the chase: in this paper, we design algorithms and prove lower bounds that suggest that polynomial time algorithms can \emph{non-trivially but not appreciably} beat the ``linearization'' threshold above. Specifically, for any $d \in \N$, and degree $D\geq 2$ polynomials, we give an $n^{O(d)}$ time algorithm that succeeds in refuting random polynomial systems with $m \gtrsim O(n) \cdot \Paren{\frac{n}{d}}^{D-1}$ equations and our lower bounds (in restricted models of computation) suggest that this trade-off is nearly tight. This threshold is smaller (but only by a constant factor) than the linearization threshold for any $d\geq 2$. This may come as a surprise since for related problems such as \emph{maximizing} low-degree polynomials, the case of degree $D=2$ polynomials (in contrast to degree $D\geq 3$) is often ``easy" and exhibits no information-computation gap. For $2^{n^{\delta}}$ time algorithms more generally, our results suggest that the algorithmic threshold beats the linearization threshold by a multiplicative factor of $\sim n^{\delta (D-1)}$. Taken together, our results suggest that the algorithmic threshold ``smoothly'' drops from $\sim n^d$ to the information-theoretic threshold of $\sim n$ as the running time budget for the refutation algorithm grows from $\poly(n)$ to $2^n$.

Before presenting our results, we discuss how the problem above is the refutation counterpart of natural algorithmic questions arising in diverse areas.  

\subsection{Random polynomial systems generalize well-studied problems}

\emph{In algebraic geometry,} the study of random polynomials and their zeros began with the 1932 paper of Bloch and Pólya~\cite{MR1576817} leading to the seminal work of Kac~\cite{MR1575245} on average number of real zeros of random univariate polynomials. More recently, beginning with work of Shub and Smale as part of their ``Bezout series'' ~\cite{MR1175980,MR1230872,MR1213484,MR1377247,MR1294430,MR2496558,MR2496559}, an influential sequence of works has focused on estimating the distribution of the number of common zeros of $n$ Gaussian random $n$-variate polynomials of degree $D$. For example, it is known that the expected number of complex common zeros grows as $\sim D^{m/2}$ -- a quadratic improvement over the the ``worst-case'' bound obtained via Bezout's theorem. Extending this to counting real common zeros requires constraining the combinatorial structure of the monomials with non-zero coefficients and more sophisticated ideas (see Kostlan's work~\cite{MR2021981} for an overview). As his 17th problem for the 21st century, Smale~\cite{smale17} asked if there is a deterministic polynomial time algorithm for finding one such common zero. A sequence of breakthroughs due to Beltrán and Pardo~\cite{MR2403529}, Burgisser and Cucker~\cite{MR2846491} and Lairez~\cite{MR3709332} resolved Smale's question and found a deterministic polynomial time algorithm based on \emph{numerical homotopy methods}. 

The problem we study in this work is a natural extension: when the number of equations is $m \leq n$, the pertinent algorithmic problem is that of counting and finding common zeros. When $m > n$, the relevant question is of finding refutations (i.e.\ certificates of unsatisfiability). 

\emph{In combinatorial optimization,} refuting random polynomial systems generalizes foundational problems such as certifying bounds on combinatorial quantities like the clique number and chromatic number of random graphs. One well-studied special case is that of refuting random constraint satisfaction problems (CSPs), which is equivalent to refuting \emph{sparse} polynomial equations (one for each ``clause'') with random coefficients over the hypercube. A long line of work~\cite{MR2121179-Feige02,MR2286509,4389511,AOW15,DBLP:conf/colt/BarakM16,DBLP:conf/stoc/RaghavendraRS17,10.1145/2746539.2746625,KMOW17,10.5555/3458064.3458092} have led to a complete understanding of algorithmic thresholds for refuting random constraint satisfaction problems in terms of a basic combinatorial property of the underlying predicate. The problem we study in this work is a natural \emph{dense} (i.e.\ all monomials appear with non-zero coefficients) counterpart to the sparse random polynomial systems arising in the study of random CSPs.

\emph{In statistical learning theory,} random polynomial systems arise naturally and in fact are closely related to the well-studied matrix and tensor sensing problems. For example, in the matrix sensing problem with ``random Gaussian measurements'', there is an unknown rank-$r$ matrix $X$ such that one is required to reconstruct $X$ from equations of the form $\langle G_i,X \rangle = b_i$ where $G_i$ are random matrices with Gaussian entries. When the rank $r = 1$ (and $X=xx^{\top}$ is symmetric) this corresponds to the problem of solving random quadratic equations (i.e.\ $D=2$) where the right hand sides correspond to evaluation of the polynomials at some unknown vector $x$. The tensor sensing problem is a generalization where instead of matrices, $G_i$s are random Gaussian tensors of order $D$. A long line of work beginning with~\cite{DBLP:conf/soda/Candes10,DBLP:journals/jmlr/Recht11,5714248} has led to essentially optimal algorithms based on semidefinite programming for solving the matrix sensing (and variants such as matrix completion) problem and more recently, much progress~\cite{DBLP:conf/colt/BarakM16,DBLP:conf/colt/PotechinS17,9317965} has been made even on the tensor variants. This work can be seen as studying the refutation variant of the matrix/tensor sensing problems for rank 1 matrices/tensors. 

\emph{In cryptography,} random polynomial systems over the reals arise naturally in a recent sequence of works that use conjectures about the \emph{hardness} of solving random polynomial systems. The work of~\cite{10.5555/3081770.3081772} led to a program~\cite{10.1007/978-3-319-63688-7_20,10.1007/978-3-319-56620-7_6,cryptoeprint:2017:250} for building \emph{indistinguishability obfuscation} (iO) based, among other components, on a certain variant of Goldreich's~\cite{DBLP:journals/iacr/Goldreich00a} candidate pseudo-random generator. An offshoot of this program recently culminated~\cite{JLS21} in the discovery of the first construction of indistinguishability obfuscation based on standard assumptions. 

At a high-level such works consider maps $f:\Z^n \rightarrow \Z^m$ where each of the $m$ outputs is computed as a low-degree $D$ polynomial $p_i$ of the $n$ inputs. The interest is in finding maps $f$ such that $D$ is small (say $2$) but for $m \gg n$ (say, $m \sim n^{1.1}$ for concreteness), the $m$-dimensional output is computationally indistinguishable from some distribution where each output is independent. Several candidate constructions of such pseudo-random generators were shown to be insecure by describing efficient algorithms (based on the sum-of-squares hierarchy of semidefinite programs studied in this work) that \emph{invert} the map $f$ -- i.e.\ compute a solution to the system of polynomial equations defined by the map~\cite{DBLP:conf/tcc/LombardiV17,BBKK18,BHJK19}. One candidate construction (see~\cite{BHJK19}) was in fact based on choosing each of the $m$ polynomials to be random quadratic polynomials as in the model studied in this paper. This work provides strong evidence for the algorithmic threshold for the refutation version of this problem.

\subsection{Our results}

\paragraph{Algorithms.} Our main algorithmic result uses the sum-of-squares hierarchy to non-trivially improve on the linearization trick for refuting random polynomial systems. We note that for the various special cases (such as random constraint satisfaction problems, the matrix/tensor sensing problems and generalizations of Goldreich's pseudo-random generator), semidefinite programs from the sum-of-squares hierarchy provide the state-of-the-art algorithms for solving/refuting polynomial systems. 

\begin{theorem}[Refutation Algorithm, Informal, See~\pref{thm:nullsatz_upper_bound} for a formal version] \label{thm:algo-main-intro}
For every $D \in \N$ and every $d \geq D$, there is a $n^{O(d)}$ time algorithm -- namely the canonical degree-$d$ sum-of-squares relaxation -- that takes input a system of $m$ polynomial equations and either correctly outputs ``infeasible'' or returns ``don't know''. When each coefficient of each input polynomial is drawn from an independent \underline{nice} distribution and $m \geq O(n^{D}/d^{D-1})$, the algorithm outputs ``infeasible'' with probability $\geq 1-1/n$.
\end{theorem}
We note a few important comments about some implications to settings studied in average-case complexity, cryptography and proof complexity. 

\noindent \emph{Computational Model:} The algorithm works in the standard word RAM model of computation. We assume that the coefficients of all our polynomials are rational numbers. The bit-complexity of our algorithm (see~\pref{thm:nullsatz_upper_bound} for details) is polynomial in the input size (including the size of the coefficients of the input polynomials).

\bigskip
\noindent \emph{Nice Distributions:} Our algorithm works for any system of random polynomial equations as long as the coefficients are chosen from independent (possibly different for each coefficient) distributions as long as they satisfy two niceness conditions (see~\pref{def:nice-distributions}). The first asks that the distributions be supported on rational numbers with some upper bound $B$ on the bit complexity. The running time of our algorithm grows polynomially in $B$. Such a condition is essentially\footnote{In principle, there could be specialized algorithms that work with non-standard representations of real numbers. We do not know of any such algorithm. } necessary for any algorithmic result. The second condition forces a certain weak \emph{anti-concentration} property and posits that no rational number should have a probability larger than $n^{-O(d)}$. We note that $n^{O(d)}$-bit rational truncations of any natural continuous distribution such as uniform distribution on $[-1,1]$ or the standard Gaussian distribution $\cN(0,1)$ satisfy such properties.

\bigskip
\noindent \emph{Time vs Signal Strength Trade-off:} The algorithm provides a certificate of unsatisfiability of the input polynomial system since whenever the algorithm outputs ``infeasible'' it is correct. The guarantees of the algorithm provide a trade-off between running time budget (parameterized by $d$) and the number of equations (a measure of ``signal strength'' in this setting) required for the algorithm to succeed in refuting with high probability. For any $d$, the smallest $m$ for which the algorithm succeeds improves on the requirement of the linearization trick by a factor of roughly $d^{\Omega(D)}$. On the other hand, to refute at the information-theoretically minimum required $m$, the algorithm runs in time exponential in $n$. In general, by setting $d=n^{\delta}$, we obtain a $2^{O(n^{\delta})}$ time algorithm that succeeds in refuting random degree-$D$ polynomial systems with $O(n) \cdot n^{(1-\delta)(D-1)}$ equations. 

One can view this result as a generalization of the work of Bhattiprolu, Guruswami and Lee~\cite{DBLP:conf/approx/BhattiproluGL17}, and Raghavendra, Rao and Schramm~\cite{DBLP:conf/stoc/RaghavendraRS17} who proved a sum-of-squares degree vs signal-strength trade-off for certifying maxima of random tensors and refuting random CSPs, respectively. In particular, the result in~\cite{DBLP:conf/stoc/RaghavendraRS17} can be seen as a degree vs number of equations trade-off that is qualitatively similar to ours above for random \emph{$1$-sparse} polynomials (i.e.\ monomials) over the hypercube.

\bigskip
\noindent \emph{The Importance of Solutions with Typical vs Atypical Norm:} There appears to be a key and perhaps surprising difference in the setting of random polynomial system refutation when compared to random CSP refutation (and more generally, related problems such as certifying maxima of random low-degree polynomials) that we wish to highlight. For random polynomial systems arising in the context of refuting CSPs (as in~\cite{AOW15,DBLP:conf/stoc/RaghavendraRS17}), the case of degree $D=2$ polynomial systems is ``easy" and appears to exhibit essentially no information-computation gap. In contrast, in our setting, our upper bound above requires $m =\Omega(n^2)$ for polynomial time algorithms to succeed in refutation. Further, our lower bounds suggest that our algorithm above is in fact essentially optimal in the running time vs number of equations trade-off. 

This apparently paradoxical difference is related to the issue of having a good upper bound on the $\ell_2$ norm in the solution space. In the context of CSP refutation, the goal is to find certificates of unsatisfiability over the $n$-dimensional hypercube -- in particular, the solution vectors have a fixed $\ell_2$ norm of $\sqrt{n}$. Indeed, the spectral (and thus, SoS-based) refutation algorithms developed in that context continue to work even for refuting random (sparse) polynomial systems over the space of all vectors with ``typical" (with respect to the scale of the coefficients of the input polynomials $p_i$s and the right-hand sides $b_i$s) $\ell_2$ norm. 

On the other hand, in our setting, the goal is to refute the given random polynomial system over a solution space of vectors of arbitrarily large norms. This crucial difference appears to make our setting harder even for the usually tame case of  quadratic polynomials. Indeed, this becomes even more apparent when we construct our low-degree polynomial hardness described below where it's crucial to make a subtle choice of planted distribution where the solution vector needs to be of $\ell_2$-norm about $n^{1/2}$-factor larger than ``typical". 

Further, this uncertainty in $\ell_2$-norm of the solution space in fact occurs in random polynomial systems that arise in applications. For e.g., in the cryptographic applications discussed above, the ``planted" solution vectors have integer coordinates with variance $\poly(n)$ and thus, the $\ell_2$-norm is known only up to some (large) $\poly(n)$-factor. That is precisely the setting where our results apply and appear to suggest a difference from the refutation settings studied in prior works. In particular, it suggests that speculating the hardness of solving/refuting random polynomial systems based on CSP hardness results may lead to incorrect conclusions. We believe that it's an interesting goal to chalk out a full trade-off between sparsity, length of the solution vector and algorithmic thresholds. Such an endeavor is likely to yield interesting insights into the phase transitions between the qualitatively different behaviors exhibited by random polynomial systems.

\bigskip
\noindent \emph{\nullsatz vs Sum-of-Squares Refutation:} Our proof of the theorem above works by showing that there is a degree-$d$ sum-of-squares ``refutation'' (i.e.\ a proof of unsatisfiability of the polynomial system that can be written in the restricted degree-$d$ \emph{sum-of-squares proof system}) of the input polynomial equations (see~\pref{sec:proof_systems}). Thus, from a proof complexity perspective, our result shows that there are degree-$d$ sum-of-squares refutations for systems of random polynomials over the reals whenever $m \geq O(n^{D}/d^{D-1})$. Our proof in fact establishes a stronger result: we show that our certificate of unsatisfiability can be written in the (formally) weaker \nullsatz proof system. As we discuss next, this shows that for the problem of refuting random polynomial equations, the degree required for \nullsatz and sum-of-squares proof systems can only be different by some fixed constant factor. 

\paragraph{Sharp thresholds at degree 2.} For the special case of degree $d=2$ and $D=2$ (i.e.\ quadratic polynomials and degree-$2$ sum-of-squares algorithm) and standard Gaussian coefficients, we can obtain sharp constants in the threshold $m$. 
\begin{theorem}[Sharp Thresholds for Degree-2 SoS] \label{thm:deg_2_main_thm}
Let $\cG = \{g_i(x) = \B{i}\}_{i \in [m]}$ be a system of $m$ polynomial equations where each coefficient of each $g_i$ is chosen from the standard Gaussian distribution $\cN(0,1)$. Then, 
    \begin{itemize}
        \item if $m \geq \frac{n^2}{4} + \wt{O}(n)$, there is a degree-$2$ sum-of-squares refutation of the system $\cG$ with probability $0.49$.

        \item if $m \leq \frac{n^2}{4} - \wt{O}(n)$, there is no degree-$2$ sum-of-squares refutation of the system $\cG$ with probability $1-1/n$. 
    \end{itemize}
\end{theorem}

Our proof of \pref{thm:deg_2_main_thm} is short and is based on direct application of results from~\cite{ALMT14} that build methods based on conic integral geometry to analyze the feasibility of convex programs with random inputs. Our present analysis does not give a bit-complexity bound on the degree-$2$ sum-of-squares proofs obtained via this technology. As a result, they do not immediately imply algorithmic results. However, they do strongly suggest that the threshold value of $m$ at $d=2$ should be $\sim n^2/4$. 

It is not hard to prove that the threshold $m$ for $d=2$-degree \nullsatz refutation is $\sim n^2/2$. Thus, this result implies a factor $2$ multiplicative gap between the thresholds for \nullsatz and SoS refutations to succeed at degree $d=2$.

\paragraph{Lower bounds in the low-degree polynomial model.} Our algorithms beat the linearization trick non-trivially at all degrees $d$. However, polynomial-time methods from our schema still require $\Omega(n^{D})$ equations for the refutation to succeed. This is a factor $n^{D-1}$ larger than the information-theoretic threshold of $n+1$. Thus, it is natural to ask if this \emph{information vs computation} gap is ``real'' and in particular, if our algorithmic results are suboptimal. We provide lower bounds that suggest that our algorithmic results are tight up to absolute constant factors. 

Our lower bounds actually hold for the formally \emph{easier}\footnote{Any refutation algorithm provides a distinguishing algorithm that succeeds with high probability in distinguishing between an instance of a random polynomial system from an instance chosen at random from any planted distribution. This algorithm runs the refutation algorithm and simply returns ``not planted'' if the algorithm outputs ``infeasible''. } algorithmic task of distinguishing random systems of polynomial equations from an appropriately designed \emph{planted} distribution on polynomial systems that always admit a solution. 

In this work, we prove the following lower bound for the distinguishing variant above in the \emph{low-degree polynomial} model of computation. 

\begin{theorem}[Low-Degree Hardness, Informal, See~\pref{thm:low_degree_hardness} for a formal version] \label{thm:low-degree-hardness-intro}
Fix $D \in \N$. For every $d \leq \frac{2n}{D}$, whenever $m \leq O\Paren{\frac{n^D}{d^{D-1}}}$, there exists a probability distribution $\nu_P$ on systems of $m$ polynomial equations that admit a solution with probability $1$ such that degree-$d$ polynomials fail to distinguish between $\nu_P$ and the distribution of random polynomial systems with $m$ equations. 
\end{theorem}

The trade-off between $d$ and $m$ achieved by \pref{thm:low-degree-hardness-intro} matches that of our algorithm in \pref{thm:algo-main-intro} up to absolute constant factors. This suggests, in particular, that the algorithmic threshold of polynomial-time algorithms might be $\Omega(n^D)$. 

The \emph{low-degree polynomial method} (see~\cite{KWB19} for a great exposition) allows distinguishers that compute thresholds of bounded-degree polynomials of input data. While low-degree polynomials might appear restricted, they capture several algorithms including power iteration, approximate message passing, and local algorithms on graphs (cf.\ \cite{DMM09,GJW20}).  Moreover, it turns out that they are enough to capture the best known spectral algorithms for several canonical problems such as planted clique, community detection, and sparse/tensor principal component analysis~\cite{BHK19,HS17,DKWB19,HKP+17}. This model arose naturally from work on constructing sum-of-squares lower bound for the planted clique problem~\cite{BHK19}. It was formalized in~\cite{HKP+17} with a concrete quantitative conjecture (called the \emph{pseudo-calibration conjecture}) which informally says that for average-case refutation problems satisfying some mild niceness conditions, degree-$(d\log n)$ lower bounds for the low-degree polynomial method for distinguishing a random draw from a random draw of some planted distribution imply lower bounds on the canonical sum-of-squares relaxation of degree-$d$ for the refutation problem. Subsequently, starting with~\cite{HS17,Hop18}, researchers have used the low-degree polynomial method as a technique to demarcate average-case algorithmic thresholds for a number of average-case algorithmic problems including densest $k$-subgraph, sparse/tensor principal component analysis, finding independent sets in random graphs among others~\cite{HKP+17,GJW20,SW20,Wei20}.

\paragraph{Sum-of-Squares lower bound at degree $4$.} We provide further evidence in favor of the thresholds suggested by both our algorithms and hardness results by proving a lower bound on the degree-$4$ sum-of-squares relaxation for refuting random \emph{quadratic} polynomial systems. Our proof is based on constructing a dual witness via pseudo-calibration -- this has become a standard technique for constructing dual witnesses for sum-of-squares lower bounds~\cite{BHK19,HKP+17,MRX20,GJJ20}. We believe that it is possible to extend our lower bounds (with the same construction of the dual witness) to both higher-degree random polynomials and higher-degree SoS relaxations. But this will likely require challenging technical work in the analysis. 

\begin{theorem}[Sum-of-Squares Lower Bound at Degree $4$, see \pref{thm:main-sos-lower-bound} for a formal version]
Let $g_1, g_2, \ldots g_m$ be homogeneous degree-$2$ polynomials in $x_1, x_2, \ldots, x_n$ with independent Gaussian coefficients. Then, whenever $m \leq n^2/\poly(\log n)$, the canonical degree-$4$ sum-of-squares relaxation fails to refute $\{g_i(x) = 0\}_{i \leq m}$ with probability at least $1-o(1)$ over the draw of $g_i$s.
\end{theorem}

\begin{remark}[Hardness of Refutation vs Hardness of Natural Planted Variants] \label{rem:planted-vs-ref}
Random polynomial systems arising in applications are often studied in two closely related variants: the refutation version for random polynomial system (the \emph{null} model) studied in this work and a related \emph{planted} variant. In the planted setting, the resulting polynomial system has a solution with probability $1$. There are three natural problems that are studied in this context: 1) efficiently distinguish between a polynomial system chosen from either random or the planted distribution, 2) efficiently find the planted solution, and 3) efficiently refute the existence of a solution.  

For average-case variants of several well-studied problems, the complexities of the three problems for natural planted and null distributions are often conjectured to be related. Indeed, researchers often prove lower bounds for the refutation problem (this turns out to be especially natural in the context of hardness for convex programs) and interpret it as a lower bound for the associated planted variant \footnote{For e.g., the works~\cite{pmlr-v40-Deshpande15,DBLP:conf/stoc/MekaPW15,DBLP:conf/soda/HopkinsKPRS16,BHK19} proving sum-of-squares lower bounds for refuting clique number of random graphs has \emph{planted or hidden clique} appearing in the title.}.

The natural, well-studied planted variant in the context of random polynomial systems happens to be the following model: a) choose polynomials $p_1, p_2, \ldots,p_m$ randomly, say, with independent Gaussian coefficients, b) choose a $z \sim \cN(0,1)^n$, and c) output $\{p_i(x) = p_i(z)\}_{i \leq m}$ with the planted solution $z$. This planted variant captures both the rank-1 case of the matrix/tensor sensing problems in machine learning and the low-degree pseudo-random generators arising in recent works~\cite{DBLP:conf/tcc/LombardiV17,BBKK18,BHJK19} on constructing indistinguishability obfuscation in cryptography. 

The planted distribution that we use in proving both low-degree hardness and our sum-of-squares lower bound is actually different from this natural variant and it turns out that this is necessary! Indeed, for $D=2$, for e.g, the natural planted distribution on quadratic systems above turns out to be solvable at nearly the information-theoretic threshold of $m=\wt{O}(n)$ via the nuclear norm minimization semidefinite program ($\sim$ degree-2 SoS). In contrast, our results suggest that refuting random degree-2 polynomial systems in polynomial time likely requires $\Omega(n^2)$ equations. On the other hand, our slightly subtle variant (see \pref{def:planted_distribution}) of the planted distribution that appears hard even with $\Omega(n^2)$ equations for all the three problems. 

Beyond the application to polynomial systems, this suggests that care must be taken in speculating hardness of natural planted variants of average-case problems based on the hardness of the refutation variants of the problem. 
\end{remark}

\subsection{Overview of our techniques}
In this section, we give a brief overview of our techniques.  

\paragraph{Algorithm via completeness of generated ideals.}
Let $\{p_i(x)-b_i=0\}_{i \leq m}$ be the input polynomial equations of degree $D$ given to the algorithm. If $x$ satisfies this system, observe that it must hold that $p_i(x)x^{\alpha} - b_i x^{\alpha} = 0$ for any monomial $x^\alpha$. Further, if $|\alpha| =d-D$ is the degree of the monomial $x^{\alpha}$, then this reasoning is ``captured'' by the degree-$d$ sum-of-squares proof system (in fact, by simply the degree-$d$ \nullsatz proof system). Thus, starting from the original polynomial system, we can ``derive'' a collection of degree-$d$ polynomial equations that must all be true if the original system is: $\{p_i(x)x^{\alpha} - b_i x^{\alpha}=0\}_{i\leq m}$ -- we call this the \emph{generated ideal at degree $d$}. 

Here's the main idea in our algorithm: suppose that the generated ideal at degree $d$ happens to be \emph{complete} -- that is, for every \emph{homogeneous} polynomial $f$ of degree $d$, there are polynomials $a_1, a_2, \ldots, a_m$ of degree $d-D$ such that $\sum_{i \leq m} a_i(x) (p_i(x)-b_i) = f(x)$. We claim that it is easy to find a refutation in this case. To see why, suppose that for some $i\leq m$, $b_i \neq 0$ (such an $i$ exists whp). Then, note that we can derive ${p_i(x)}^{d/D} = b_i^{d/D}$ from the input equations in degree $d$ (assuming $d$ is a multiple of $D$). On the other hand, since ${p_i(x)}^{d/D}$ is a homogeneous polynomial of degree $d$ and the generated ideal at degree $d$ is complete, we must also have that ${p_i(x)}^{d/D} = \sum_i a_i(x) (p_i(x)-b_i)$ for some polynomials $a_i$ of degree $d-D$. Thus, together, we can infer that $b_i^{d/D} = {p_i(x)}^{d/D}=\sum_j a_j(x) (p_j(x)-b_j)$ or:
\[
1 = b_i^{-d/D} \sum_j a_j(x) (p_j(x)-b_j)\mper
\]
This is a (degree-$d$ \nullsatz and thus, sum-of-squares) refutation since the LHS is the constant $1$ while at any $x$ that satisfies the input system, the RHS must be $0$. Finally, we can argue (see \pref{lem:low-bit-complexity-rep-arbitrary-gen-ideals}) that whenever such a polynomial identity as above exists, the $a_i$s can be guaranteed to exist with coefficients of bit-complexity polynomial in $n^d$ and the bit-complexity of the coefficients of the inputs $p_i$s. This immediately implies (via \pref{fact:finding-pseudo-distributions}) that the $n^{O(d)}$ time algorithm (see \pref{algo:refute-degree-2-polynomials}) for approximately solving the canonical degree-$d$ SoS relaxation of the polynomial system above succeeds in refuting the input polynomial system. 

Thus, our task reduces to establishing that when $m \gg n^D/d^{D-1}$, the generated ideals at degree $d$ of random $(p_i-b_i)$s are complete. Such a condition naturally yields a system of linear equations so our task reduces to proving that this system admits a solution -- that happens if and only if the coefficient matrix of the equations has full row rank. One might be tempted to prove such a claim by showing that when the $p_i$s are random, then for each $i$ and each $\alpha$, $p_i(x)x^{\alpha}-b_ix^{\alpha}$ are linearly independent when viewed as their coefficient vectors. This is false -- there are several linear dependencies between such vectors. 

Indeed, in general, such an argument requires some care as the entries of the matrix defining the linear equations are heavily correlated -- this is not surprising since the there are roughly $mn^D \leq n^{2D}$ independent random variables in the input while the matrix is of dimension roughly $n^{d}$ (and $d \gg D$). We analyze this matrix by a careful decomposition (see \pref{def:row_rank_decomposition}) that exploits the structure of the matrix to argue that whenever $m \gg n^{D}/d^{D-1}$, the resulting matrix is indeed full row rank with probability $1-n^{-O(d)}$ over the draw of the coefficients of $p_i$s. 

\paragraph{Low-degree hardness and the hard-to-distinguish planted distribution.}
In order to construct our lower bound, we need to come up with a planted distribution on polynomial systems with $m \geq n+1$ equations such that 1) every system in the support always admits a solution but at the same time, 2) a draw from the planted distribution is indistinguishable from random polynomial systems that do not have a solution with probability $1$ by any low-degree polynomial in the coefficients of the input polynomials. Notice that for e.g., this must mean that low-degree polynomials in the coefficients of the input polynomials cannot approximate the fraction of polynomial constraints satisfied by any $x$. Operationally, this means that we must pick a distribution on polynomial equations that is ``as close as possible'' to random polynomial systems (the null model) while being satisfiable. 

As we discussed in \pref{rem:planted-vs-ref}, the design of the planted distribution is slightly subtle. One might be tempted to use the natural (and well-studied) variant where we pick each $p_i$ randomly just as in the null model and then choose $b_i$s to equal $p_i(x^*)$ for each $i$ for some random $x^{*}$. Notice that this must introduce correlations in $b_i$s and it in fact turns out that these correlations are strong enough that the resulting planted model can be easily distinguished from the null model by just a degree-$4$ polynomial\footnote{The degree-$4$ polynomial $(\sum_i b_i^2)^2$ is a distinguisher whp between the null and planted models.} in the coefficients of the input polynomials whenever $m \gg n$ -- the information theoretic threshold. This is not surprising as there is in fact an algorithm (the so called \emph{nuclear norm minimization} semidefinite program) that recovers the planted $x^*$ when given input a random polynomial system chosen from the planted model above. 

Instead, our construction of the planted distribution encodes subtle correlations in the polynomials $p_i$s themselves. Specifically, our planted distribution first picks $b_i$s to be independent draws from the standard Gaussian distribution, chooses a random $x^*$ \emph{of sufficiently large length that $\rightarrow \infty$ as $m \rightarrow \infty$}, and then chooses $p_i$s to be polynomials with standard Gaussian coefficients conditioned on $p_i(x^*) =b_i$. In this version, notice that $b_i$s are clearly independent but unlike the natural planted variant, the coefficients of $p_i$s are mildly correlated. We show that such correlations are subtle enough that no low-degree polynomial can ``notice'' them. The argument crucially requires that the planted solution $x^*$ have sufficiently large norm -- in \pref{rem:importance-of-scaling} we show that there's a simple distinguisher if the planted solution has bounded or slowly growing norm. It turns out that when the planted solution $x^*$ has sufficiently fast-growing norm, there's a \emph{sharp phase transition} for distinguishability by degree-$d$ polynomials at a threshold $m = O_D(n^D/d^{D-1})$ from the planted model -- a threshold that precisely matches the bound at which our algorithm works! This gives a pleasingly tight algorithmic threshold of $\Theta_D(n^D/d^{D-1})$ for distinguishing random polynomial systems from the above planted ones and thus also for refuting them. 

Our analysis of the performance of low-degree polynomial distinguishers for the above pair relies on expressing the coefficients of the ``likelihood ratio'' (ratio of the probability density functions of the planted and the null distributions) in the Hermite basis -- this is a standard strategy~\cite{KWB19,SW20}  employed in proving such results. The performance of the low-degree polynomial distinguishers is related (again via standard ideas from prior works) to the truncated low-degree likelihood ratio -- a natural quantity that depends on the density functions of the planted and the null models above. Our analysis then proceeds by combinatorial characterization and estimates for the Hermite coefficients of the planted density function. 

\paragraph{Sum-of-Squares lower bounds.}
Our sum-of-squares lower bound shows that for a system of $m$ random homogeneous quadratic equations with RHS all set to $0$, there is no degree-$4$ sum-of-squares refutation as long as $m \leq n^2/\poly(\log n)$. As is standard, we show such a statement by exhibiting a dual witness --a \emph{pseudo-distribution} of degree $4$ (see \pref{def:pseudoexpectation}) -- that is consistent with the input system of polynomial equations. 

More specifically, we view the equations as $\{x^\top G_i x = 0\}_{i \leq m}$ where each $G_i$ is a matrix with independent standard Gaussian entries. A pseudo-distribution of degree $4$ satisfying such constraints is a linear map  $\pE_{\mu}$ that assigns a real number to every degree $\leq 4$ polynomial and satisfies 1) \textbf{Normalization:} $\pE_{\mu}[1] =1$, 2) \textbf{Positivity:} $\pE_{\mu}[q^2] \geq 0$ for every degree $\leq 2$ polynomial $q$, and 3) \textbf{Constraints:} $\pE_{\mu}[(x^{\top} G_i x) q]=0$ for every degree-$2$ polynomial $q$ and every $i$.

Our construction of such a map uses \emph{pseudo-calibration} -- a general technique for constructing candidate dual witnesses discovered in~\cite{BHK19}. Informally speaking, this technique gives a ``mechanical'' method of constructing a candidate pseudo-distribution for an average-case refutation problem on some null distribution given a \emph{planted} distribution that is indistinguishable from the null by low-degree polynomials. Our construction (see \pref{def:candidate-pd}) is based on the planted distribution described above in the context of our proof of lower bounds for the low-degree method. 

While pseudo-calibration makes the job of coming up with candidate pseudo-distributions easy, the analysis of the resulting construction still essentially needs to be done via techniques specific to a given setting (we note that the pseudo-calibration conjecture of~\cite{HKP+17} hypothesizes the existence of a more mechanical translation). Thus, the bulk of our technical work goes into analyzing the construction above. 

Our idea (as is standard in such settings) is to use the Hermite polynomial basis to explicitly write down expressions for the pseudo-distribution. Analyzing the pseudo-distribution requires analyzing the spectrum of the \emph{moment matrix} associated with the pseudo-distribution. The moment matrix $\MomentMatrix$ is indexed by indices of monomials $I,J$ of degree $\leq 2$ on rows and columns and has its $(I,J)$ entry given by $\pE_{\mu}[x^{I} x^{J}]$. In our case, note that this is a random matrix with heavily correlated entries. The positivity property of $\pE_{\mu}$ is equivalent to the positive semidefiniteness of the matrix $\MomentMatrix$. 

Our analysis works by decomposing the $\MomentMatrix$ into a linear combination of \emph{graph matrices} (analogs of the bases used in prior works e.g.~\cite{HKP+17,BHK19,GJJ20}) that form a good basis for analyzing the spectra of such correlated random matrices. Thankfully, some of the technology for understanding the spectra of such graph matrices -- whose spectra can be directly related to combinatorial properties of associated graphs called \emph{shapes} -- was developed in the context of proving $n^{O(1)}$-degree sum-of-squares lower bounds for the Sherrington-Kirkpatrick Hamiltonian by~\cite{GJJ20,AMP20}. 

The high-level outline of our analysis resembles the strategy adopted by~\cite{GJJ20} though the details differ because of the difference in the structure of the pseudo-distribution. First, the construction above does not quite exactly satisfy the constraints but we show that an appropriately small perturbation of it does. To analyze this construction,  we study the decomposition of the $\MomentMatrix$ and identify the shapes that are \emph{negligible} (i.e.\ contribute sufficiently small singular values), \emph{trivial} (these contribute a large positive semidefinite mass) and \emph{spiders} -- these can be ``killed'' -- that is, one can show that the contributions of the corresponding terms adds up to $0$. 

While our analysis follows a similar high-level plan to~\cite{GJJ20} so far, the combinatorial characterization of shapes that fall into each of the three types above differs from that ~\cite{GJJ20} and requires an analysis specialized to our setting. This is because ~\cite{GJJ20} work with a special form of ``rank 1'' polynomial constraints relevant to their setting $\{\brac{x,g_i}^2 = 1\}_{i\leq m}$ where the $g_i$s are random vectors (the ``affine planes'' problem). As a result, the resulting construction of pseudo-distribution leads to a moment matrix with a different set of shapes playing a prominent role -- 2 uniform graphs as opposed to 3-uniform hypergraphs in our case. 

With the above techniques, it is possible to obtain a lower bound that works whenever $m \ll n^{1.5}$ as in the work of~\cite{GJJ20}. But just as in their setting, this bound is off from the optimal bound of $\sim n^2$. This is entirely due to the difficulties in the analysis of the construction above. 

In our setting, we are able to obtain (only at degree $4$) an analysis of the construction that does work all the way to the threshold of $n^2$ up to polylogarithmic factors in $n$. One crucial ingredient is a further sub-classification of non-negligible shapes appearing in the decomposition into \emph{disconnected} and the rest. We give a different ``charging scheme'' for the disconnected shapes in order to show that they cannot contribute negative eigenvalues to the spectrum of $\MomentMatrix$.

\section{Preliminaries}

\subsection{Notations}
We use the standard conventions $\N = \{0,1,2,\dots\}$ and $[N] = \{1,2,\dots, N\}$.
Consider vectors $x\in \R^N$ and $\alpha \in \N^N$.
We use the notation $|\alpha| = \sum_{i=1}^N \alpha_i$ and $\alpha! = \prod_{i=1}^N (\alpha_i!)$, and further denote $x^\alpha \seteq \prod_{i=1}^N x_i^{\alpha_i}$.
Moreover, we say $\alpha$ is \textit{simple} if $\alpha \in \{0,1\}^N$, i.e.\ the monomial $x^{\alpha}$ is multilinear.
With slight abuse of notation, we will often treat $\alpha\in \N^N$ as a multiset of $[N]$.

In this paper, we will encounter the case where $N = m \times n \times n$.
The same notations apply: for $\alpha \seteq (\Alpha{1},\dots, \Alpha{m}) \in \N^{m \times n \times n}$, $|\alpha| = \sum_{\s\in[m]} \sum_{i,j\in[n]} \Alphasij$ and $\alpha! = \prod_{\s\in[m]} \prod_{i,j\in[n]} (\Alphasij)!$.
In this case, we may view $\alpha$ as a \textit{labeled directed multigraph} (with self-loops allowed) on vertex set $[n]$, where each edge has a label in $[m]$.
Thus, $|\alpha|$ is the total number of edges, and $|\Alphas|$ is the number of edges labeled $\s$.

In this work, we will deal with algorithms that operate on numerical inputs. In all such cases, we will rely on the standard word RAM model of computation and assume that all the numbers are rational represented as a pair of integers describing the numerator and the denominator. In order to measure the running time of our algorithms, we will need to account for the length of the numbers that arise during the run of the algorithm. The following definition captures the size of the representations of rational numbers:

\begin{definition}[Bit Complexity]
The bit complexity of an integer $p \in \Z$ is $1+ \lceil \log_2 p \rceil$. The bit complexity of a rational number $p/q$ where $p,q \in \Z$ is the sum of the bit complexities of $p$ and $q$. 
\end{definition}

\subsection{Hermite polynomials}

In this section, we introduce the \textit{Hermite polynomials}, which are orthogonal polynomials with respect to the Gaussian measure (see \cite{Sze39} for a standard reference).
The \textit{univariate Hermite polynomials} $\{h_k\}_{k\in \N}$ are defined by the following recurrence:
\begin{equation*}
    h_0(x) = 1, \quad h_1(x) = x, \quad h_{k+1}(x) = x h_k(x) - k h_{k-1}(x) \mper
\end{equation*}
Next, we define the \emph{multivariate} Hermite polynomials.
For an index $\alpha \in \N^{N}$ and vector $x\in \R^N$, $h_{\alpha}(x) \coloneqq \prod_{i=1}^N h_{\alpha_i}(x_i)$.
The Hermite polynomials form an orthogonal basis with respect to the Gaussian measure:
for $\alpha_1,\alpha_2 \in \N^N$,
\begin{equation*}
    \expover{x\sim \calN(0,\Id)}{h_{\alpha_1}(x) h_{\alpha_2}(x)} = 
    \begin{cases}
        \alpha_1! & \textnormal{ if $\alpha_1 = \alpha_2$,} \\
        0 & \textnormal{ otherwise.}
    \end{cases}
\end{equation*}

We will need the following facts about Hermite polynomials:

\begin{fact} \label{fact:hermite_zero}
    For an even integer $k \geq 2$, $h_k(0) = (-1)^{k/2} (k-1)!!$. For an odd $k$, $h_k(0) = 0$.
\end{fact}

\begin{fact} \label{fact:double_factorial_bound}
    For an even integer $k \geq 2$, $(k-1)!! \leq (\frac{k}{2})^{k/2}$.
\end{fact}

\begin{fact} \label{fact:hermite_generating_function}
    For any $x\in\R$, the generating function of Hermite polynomials is the following,
    \begin{equation*}
        e^{xt - \frac{t^2}{2}} = \sum_{k=0}^\infty h_k(x) \frac{t^k}{k!} \mper
    \end{equation*}
\end{fact}

\subsection{Sum-of-Squares and \nullsatz proofs vs algorithms}
\label{sec:proof_systems}

Sum-of-squares proof system is a restricted reasoning system for certifying unsatisfiability of a system of polynomial equality and inequality constraints over the reals. We refer the reader to the monograph~\cite{TCS-086} for a detailed exposition. 

\begin{definition}[Sum-of-Squares Refutations]
Let $p_1, p_2, \ldots, p_k$ be polynomials in variables $x_1, \ldots, x_n$ with coefficients over the reals.
Given a system of constraints $\{p_i \geq 0\}_{i \leq k}$, a sum-of-squares refutation of the system is a polynomial identity of the following form:
    \begin{equation}\label{eq:sos-proof-def}
        -1 = \sum_{T \subseteq [k]}  S_T \prod_{i \in T} p_i\mcom
    \end{equation}
    where $S_0, S_1, \ldots, S_T$ are sum-of-squares polynomials. 
The \emph{degree} of the sum-of-squares proof is the minimum positive integer $\ell$ such that for every $T \subseteq [k]$ such that $S_T \neq 0$, $\sum_{i \in T}\deg(p_i) + \deg(S_T) \leq \ell$. 
\end{definition}

Observe that if an identity of the form \pref{eq:sos-proof-def} exists, then it immediately proves that the associated constraint system is unsatisfiable. This is because any $x$ that satisfies the constraint system must make the right hand side evaluate to a non-negative real number while the left hand side is the negative real $-1$. Under mild conditions on the polynomials $p_1, p_2, \ldots, p_k$, a converse holds. Such results are called \emph{positivstellensatz}. We state a general one due to Krivine and Stengle~\cite{Kri64,Ste74}.

\begin{fact}[Krivine/Stengle's Positivstellensatz~\cite{Kri64,Ste74}, see Theorem 3.73 in~\cite{TCS-086}]
Let $p_1, p_2, \ldots,p_k$ be $n$-variate real-coefficient polynomials in $x_1, x_2,\ldots, x_n$. If there does not exist $x \in \R^n$ such that $p_i(x) \geq 0$ for every $i \leq k$, then, there are sum-of-squares polynomials $\{S_T\}_{T \subseteq [k]}$ such that the following polynomial identity holds:
\[
-1 = \sum_{T \subseteq [k]} S_T \prod_{i \in T} p_i \mper 
\]
\end{fact}

While positivstellensatz implies that there's always a refutation for all polynomial systems, it provides no bound on the degree of the resulting proof.

It is instructive to compare it with the strictly weaker \nullsatz proof system that we will also encounter in this work. 
\begin{definition}[\nullsatz Refutation] \label{def:nullsatz-refutation}
Let $p_1, p_2, \ldots, p_k$ be polynomials in variables $x_1, x_2, \ldots, x_n$ with coefficients over the reals. Given a system of constraints $\{p_i = 0\}_{i \leq k}$, a \nullsatz refutation of the system is a polynomial identity of the following form:
    \begin{equation}
        1 = \sum_{i \leq k}  a_i p_i\mcom
    \end{equation}
    where $a_1,\ldots,a_k$ are arbitrary polynomials. 
The \emph{degree} of the \nullsatz proof is the minimum positive integer $\ell$ such that for every $i$, $\deg(p_i) + \deg(a_i) \leq \ell$. 
\end{definition}
Unlike the sum-of-squares proof system, the \nullsatz proof systems only deals with polynomial \emph{equality} constraints. Analogously to positivstellensatz, the completeness of the \nullsatz proof systems is implied by Hilbert's \nullsatz. 

\begin{fact}[Corollary of Hilbert's Nullstellensatz, see for e.g.,~\cite{MR1451386}]
Suppose $p_1, p_2, \ldots, p_k$ are real-coefficient polynomials in $x_1, x_2, \ldots, x_n$ such that there is no $x$ satisfying $p_i(x)=0$ for every $i \leq k$. Then, there are polynomials $a_1, a_2, \ldots, a_k$ with real coefficients such that the following polynomial identity holds:
\[
1=\sum_{i\leq k} a_i p_i \mper
\]
\end{fact}
Informally speaking, the key difference between the sum-of-squares and the \nullsatz proof system is the ability to reason about the non-negativity of square polynomials. This seemingly minor change results in a huge difference in the power of the proof systems. For example, the pigeonhole principle requires $\Omega(n)$ degree for \nullsatz to refute but has a degree-4 SoS refutation (see Claim 3.59 on Page 125 of~\cite{TCS-086} for a short proof).





\subsection{Pseudo-distributions}

Pseudo-distributions are generalizations of probability distributions and form dual objects to sum-of-squares proofs in a precise sense that we will describe below. 

\begin{definition}[Pseudo-distribution, Pseudo-expectations, Pseudo-moments] \label{def:pseudoexpectation}
A \emph{degree-$\ell$ pseudo-distribution} is a finitely-supported function $\mu:\R^n \rightarrow \R$ such that $\sum_{x} \mu(x) = 1$ and $\sum_{x} \mu(x) f(x)^2 \geq 0$ for every polynomial $f$ of degree at most $\ell/2$. (Here, the summations are over the support of $\mu$.) 

The \emph{pseudo-expectation} of a function $f$ on $\R^d$ with respect to a pseudo-distribution $\mu$, denoted $\pE_{\mu(x)} f(x)$, as
\begin{equation}
  \pE_{\mu(x)} f(x) = \sum_{x} \mu(x) f(x) \,\mper
\end{equation}

The degree-$\ell$ moment tensor of a pseudo-distribution $\mu$ is the tensor $\E_{\mu(x)} (1,x_1, x_2,\ldots, x_n)^{\otimes \ell}$.
In particular, the moment tensor has an entry corresponding to the pseudo-expectation of every monomial of degree at most $\ell$ in $x$.
\end{definition}

Observe that if a pseudo-distribution $\mu$ satisfies, in addition, that $\mu(x) \geq 0$ for every $x$, then it is a mass function of some probability distribution. Further, a straightforward polynomial-interpolation argument shows that every degree-$\infty$ pseudo-distribution satisfies $\mu \ge 0$ and is thus an actual probability distribution. The set of all degree-$\ell$ moment tensors of probability distribution is a convex set. Similarly, the set of all degree-$\ell$ moment tensors of degree-$d$ pseudo-distributions is also convex.

\begin{definition}[Constrained pseudo-distributions]
  Let $\mu$ be a degree-$\ell$ pseudo-distribution over $\R^n$.
  Let $\cA = \{p_1\ge 0, p_2\ge 0, \ldots, p_m\ge 0\}$ be a system of $m$ real-coefficient polynomial inequality constraints.
  We say that \emph{$\mu$ satisfies the system of constraints $\cA$ at degree $\ell$} if for every sum-of-squares polynomial $h$ and any $T \subseteq [m]$ such that $\deg(h) + \sum_{i\in T} \deg(p_i) \leq \ell$, $\pE_{\mu} [h \cdot \prod _{i\in T} p_i] \ge 0$.

\end{definition}

The following fact describes the precise sense in which pseudo-distributions are duals to sum-of-squares proofs. 

\begin{fact}[Strong Duality,~\cite{MR3441448-Josz16}, see Theorem 3.70 in~\cite{TCS-086} for an exposition]
Let $p_1, p_2, \ldots, p_k$ be real-coefficient polynomials in $x_1, x_2, \ldots, x_n$. 
Suppose there is a degree-$d$ sum-of-squares refutation of the system $\{p_i(x) \geq 0\}_{i \leq k}$.
Then, there is no pseudo-distribution $\mu$ of degree $\geq d$ satisfying $\{p_i(x) \geq 0\}_{i \leq k}$. 
On the other hand, suppose that there is a pseudo-distribution $\mu$ of degree $d$ consistent with $\{p_i(x) \geq 0\}_{i \leq k}$. Suppose further that the set $\{p_1, p_2, \ldots, p_k\}$ contains the quadratic polynomial $R-\sum_i x_i^2$ for some $R > 0$. Then, there is no degree-$d$ sum-of-squares refutation of the system $\{p_i(x) \geq 0\}_{i \leq k}$.
\end{fact}


\subsection{Algorithms and numerical accuracy} 

The sum-of-squares proof system is automatizable via semidefinite programming in an appropriate sense that we describe next. Informally, this means that degree-bounded sum-of-squares proofs and low-degree pseudo-distributions satisfying a system of constraints can be found via efficient algorithms. Such algorithms deal with numerical inputs and thus, in the context of algorithms, we only allow our input polynomial systems to have rational coefficients.



The following fact follows by using the ellipsoid algorithm for semidefinite programming. The resulting algorithm to compute  pseudo-distributions approximately satisfying a given set of polynomial constraints is called the \emph{sum-of-squares algorithm}. 

\begin{fact}[Computing pseudo-distributions consistent with a set of constraints \cite{MR939596-Shor87,parrilo2000structured,MR1748764-Nesterov00,MR1846160-Lasserre01}] \label{fact:finding-pseudo-distributions}
There is an algorithm with the following properties: The algorithm takes input $B \in \N$, $\tau >0$, and polynomials $p_1, p_2, \ldots, p_k$ of degree $\ell$ with rational coefficients of bit complexity $B$. If there is a pseudo-distribution of degree $d$ consistent with the constraints $\{p_i(x)=0\}_{i \leq k}$, the algorithm in time $\poly(B,\frac{1}{\tau}) \cdot n^{O(d)}$ outputs a pseudo-distribution $\mu$ of degree $d$ satisfying $|\pE_{\mu} p_i(x) x^{\alpha}|\leq \tau$ if it exists and otherwise outputs ``infeasible''.
\end{fact}

\subsection{Background on the low-degree polynomial method}
\label{sec:low_degree_method}


The low-degree polynomial method is a restricted class of computationally bounded algorithms for hypothesis testing problems arising in statistics.

In order to describe this method, let $\nu_N$ (for ``null'') and $\nu_P$ (for ``planted'' distribution; often called the ``alternative'' distribution in statistics) be a pair of probability distributions on $\R^K$. 
Informally, we will set $\nu_N$ to be a distribution on instances of some optimization problem that admit no solutions with high probability (such as random polynomial systems in our case) while $\nu_P$ will be the distribution on random polynomial systems that always admit a solution. 

In the hypothesis testing problem, the algorithm is given a sample $z$ with the promise that it is generated by the mixture $0.5 \nu_N + 0.5 \nu_P$. The goal is to determine correctly with high probability if $z$ is generated from $\nu_N$ or $\nu_P$. 
Often $\nu_N$ and $\nu_P$ are parameterized family of distributions (for e.g, the degree $D$, the number of variables $n$ or equations $m$ in our setting).

The key question is to determine the parameter regimes under which the hypothesis testing problem is solvable with high (say $1-o_K(1)$) probability. Any such ``testing'' algorithm can be seen as computing some function $T_K:\R^K \rightarrow \R$ on the input sample $z$ and outputting ``null'' if $T_K(z)$ exceeds some threshold $\tau$. Observe that a family of tests $\{T_K\}_{K}$ succeeds with probability $1-o_K(1)$ as $K \rightarrow \infty$ if $\E_{\nu_P} T_K - \E_{\nu_N} T_K \rightarrow \infty$ as $K \rightarrow \infty$.

Information-theoretically speaking, the classical Neyman-Pearson lemma identifies an optimal (in the sense of achieving optimal trade-off between false positives and false negatives) statistical test -- the \emph{likelihood ratio} -- that distinguishes the given pair of distributions. 

\paragraph{Restricting to low-degree polynomial tests.} While the likelihood ratio test is statistically optimal, it is often hard to compute and thus does not yield an efficiently computable distinguisher. The \emph{low-degree polynomial method} restricts the algorithm to a smaller class of statistical tests so as to gain computational efficiency. 

Specifically, such tests $T$ are restricted to 1) evaluating some degree-$d$ polynomial $f$ on the input sample $z$ and 2) ``accepting'' if $f(z)$ exceeds some chosen threshold $\tau$. Such a test is clearly computable in $K^{O(d)}$ time by explicitly evaluating each monomial of $f$.  

While such tests may appear restricted, recent works showed that $O(\log n)$-degree polynomial tests in fact can \emph{simulate} algorithms such as power iteration (and thus computing spectral norms), approximate message passing, and local algorithms applied to $z$ and more generally matrices/tensors with entries set to constant-degree polynomials of $z$. This allows the method to capture the strongest known algorithms for fundamental distinguishing tasks including planted clique and spiked Wigner models, and more generally, random optimization problems such as clique/independent set and densest $k$-subgraph in random graphs. In what can be construed to be an even more evidence of the power of the method, recent work~\cite{BBH+20} shows that under appropriate restrictions, algorithms in the $O(\log n)$-degree polynomial model are as powerful as polynomial time algorithms in the statistical query model studied arising in learning theory and recently applied~\cite{FGR+17} to prove lower bounds for average-case variants of several foundational combinatorial and statistical learning problems. The low-degree likelihood ratio and the low-degree polynomial tests were introduced in the context of establishing sum-of-squares lower bounds implicitly in \cite{BHK19} and formalized explicitly in~\cite{HKP+17}. In particular, for average-case distinguishing problems satisfying some mild ``niceness'' conditions, \cite{HKP+17} conjecture (this is called the \emph{\pc conjecture}) that indistinguishability by degree-$d$ polynomials implies lower bounds for a canonical $\wt{O}(d)$-degree SoS relaxation for the associated refutation problem. 

Subsequent works (starting with~\cite{HS17}, see Conjecture 2.2.4 in~\cite{Hop18} and 1.16 in~\cite{KWB19}) have proposed the stronger conjecture that concludes a lower bound against all $n^{\wt{O}(d)}$ time distinguishing algorithms.

The following definition presents a formal, quantitative version of what it means to use low-degree polynomials to distinguish between a pair of distributions as above.

\begin{definition}[Distinguishing by Low-Degree Polynomials]  \label{def:distinguishing_polynomial}
Let $\nu_N$, $\nu_P$ be a pair of null and planted distributions on $\R^K$. 
We say that degree-$d$ polynomials succeed in $(1-\delta)$-distinguishing between $\nu_N$ and $\nu_P$ from a single sample if there is a degree $\leq d$ polynomial $f:\R^K \rightarrow \R$ such that:\
\begin{enumerate}
    \item $\E_{\nu_N} [f^2] = 1$. 
    \item $\E_{\nu_P}[f] \geq \frac{1}{\delta}$. 
\end{enumerate}
\end{definition}

It turns out that it is possible to precisely characterize the \emph{best} low-degree polynomial distinguisher $f$ in terms of the density functions of the associated pair of distributions.  

\begin{proposition}[Truncated Low-Degree Likelihood Ratio; see \cite{HKP+17} and Proposition 1.15 of {\cite{KWB19}}] \label{prop:likelihood_ratio_variance}
    Let $\nu_N,\nu_P$ be a pair of probability distributions on $\R^K$. The truncated low-degree likelihood ratio $L^{\leq d}$ at degree $d$ is defined as the unique solution to $\arg \min_{f} \E_{z\sim \nu_N}[(L(z)-f(z))^2]$ where the minimization is over all degree $\leq d$ polynomials $f$. The \emph{normalized} truncated likelihood ratio $L^{\leq d}/\expover{\nu_N}{(L^{\leq d})^2}^{1/2}$ is then the optimal solution to the following optimization problem:
    \begin{equation*}
        \max\ \E_{\nu_P} f \quad \text{s.t. } \E_{\nu_N} f^2 = 1 \text{ and } f \text{ is a degree-$d$ polynomial.}
    \end{equation*}
    Moreover, the value of the optimization problem is $\expover{\nu_N}{(L^{\leq d})^2}^{1/2}$. In particular, $\nu_N$ and $\nu_P$ are $(1-\delta)$-indistinguishable by degree $\leq d$ polynomials if $\expover{\nu_N}{(L^{\leq d})^2}^{1/2} \leq \frac{1}{\delta}$.
\end{proposition}

\section{Algorithmic Thresholds: Upper Bound}
\label{sec:nullsatz}

In this section, we describe and analyze our algorithm for refuting random polynomial systems. 
Our algorithmic results apply to all random polynomial systems where all coefficients are independent from some distribution on rational numbers that satisfies some niceness properties. Such properties are satisfied by the uniform distribution on a large enough subset of rational numbers, a polynomial bit truncation of the standard Gaussian distribution among others. 

\begin{definition}[Nice Rational Distributions] \label{def:nice-distributions}
For $B \in \N$, we say that a probability distribution $\nu$ on $\Q$ is $B$-nice if the following hold:
\begin{enumerate}
    \item \textbf{$\nu$ is supported on low-bit complexity rationals}: The support of $\nu$ are rational numbers with numerator and denominator in $[-2^B,2^B]$.
    \item \textbf{$\nu$ is spread-out}: for any $q \in \Q$, $\Pr_{x \sim \nu}[x=q] \leq \frac{1}{B^{100}}$.
\end{enumerate}
\end{definition}

The main result of this section is the following theorem:


\begin{theorem}[Refutation Algorithm for Random Polynomial Systems]\label{thm:nullsatz_upper_bound}
Fix $D \in \N$. There is an algorithm with the following properties: the algorithm takes input $m$ polynomial equations $\{g_i(x) = b_i\}_{i \in [m]}$ where each $g_i$ is a polynomial of degree $D$ with rational coefficients of bit-complexity $B$, and in $(Bn)^{O(d)}$ time, either correctly outputs ``infeasible'' or returns ``don't know''. 
Further, if $m \geq O_D\Paren{\frac{n^D}{d^{D-1}}}$ and $g_i,b_i$ are obtained by sampling each coefficient of each $g_i$ and each $b_i$ from (possibly different) independent $n^{2d}$-nice rational distributions, then, with probability $1-n^{-d}$ over the choice of the input equations, the algorithm outputs ``infeasible''.
\end{theorem}

Our algorithm is quite simple. It approximately solves the degree-$d$ SoS relaxation for the constraint system $\{p_i(x)=0\}_{i \leq m}$ where $p_i(x) = g_i(x) -b_i$, and returns ``infeasible'' if the SDP outputs infeasible and ``don't know'' otherwise. 
More precisely: 
 \begin{mdframed}
      \begin{algorithm}[Refute Random Polynomials]
        \label{algo:refute-degree-2-polynomials}\mbox{}
        \begin{description}
            \item[Given:]
                A rational accuracy parameter $\tau = \exp(-n^{O(d)} B)$ and degree-$D$ polynomials $p_1, \ldots, p_m$ with rational coefficients of bit complexity at most $B$ for $B \in \N$. 
            \item[Output:] ``Infeasible'' or ``Don't Know''. 
            \item[Operation:]\mbox{}
                \begin{enumerate}
                      \item Find a degree-$d$ pseudo-distribution $\mu$ such that $|\pE_{\mu} p_i(x) x^{\alpha}| \leq \tau$ for every $i \leq m$ and monomial index $\alpha$ of degree at most $d-\deg(p_i)$. 
                      \item If no such pseudo-distribution exists, return ``Infeasible''. 
                      \item Otherwise output ``don't know''.
                \end{enumerate}
        \end{description}
      \end{algorithm}
\end{mdframed} 

\paragraph{Analysis of algorithm.} 
The key to the proof of the theorem is the following lemma that guarantees the existence of a sum-of-squares refutation for the input random polynomial system.

\begin{lemma}[Sum-of-squares refutation for random polynomial systems]
\label{lem:sos-refutation-polynomials}
Let $D\in \N$ and $d\geq D$ be a multiple of $D$.
Let $\calF = \{g_1,\dots, g_m\}$ be a set of homogeneous degree-$D$ polynomials with each coefficient of each $g_i$ chosen from an independent $B$-nice rational distribution.
Let $\B{1}, \B{2}, \ldots, \B{m}$ be independent samples from a $B$-nice rational distribution. 
Then, whenever $m \geq O_D\Paren{\frac{n^D}{d^{D-1}}}$, with probability at least $1-n^{-d}$ over the choice of the $g_i$s and $b_i$s, there exist polynomials $a_1, a_2, \ldots, a_m$ of degree $d-D$ such that the following polynomial identity holds:
\begin{equation} \label{eq:sos-proof-eq}
-1 = \sum_{i\leq m }a_i (g_i -b_i)\mper
\end{equation}
Further, the coefficients of $a_i$s are rational numbers with bit complexity at most $O(n^{5d}d \log n + n^{5d} B)$.  
\end{lemma}
\begin{remark}[Nullstellensatz vs Sum-of-Squares] 
Observe that in the refutation identity, there is no additive sum-of-squares term. As a result, our refutation is in fact a \nullsatz refutation (\pref{def:nullsatz-refutation}). As we show, there's a strong indication (see the next section on lower bounds) that the trade-off achieved by \pref{lem:sos-refutation-polynomials} between $m$ and $d$ is tight up to absolute constant factors for the sum-of-squares proof system. Thus, in this case, we expect that the $m$ vs $d$ trade-off for \nullsatz and SoS proof systems to be essentially the same. Interestingly, the constant factor gap allowed by our upper and lower bounds might be ``real''. At degree $d = 2$, it is not hard to argue that $m \geq \frac{n^2}{2}$ is necessary for a \nullsatz refutation to exist. However, \pref{thm:deg_2_sos} shows that $m \gtrsim \frac{n^2}{4}$ is sufficient for degree-2 SoS. 
\end{remark}

It is easy to complete the analysis of the algorithm using this lemma. 

\begin{proof}[Proof of \pref{thm:nullsatz_upper_bound}]
The running time of the algorithm follows immediately by applying \pref{fact:finding-pseudo-distributions}. 
In order to prove correctness of the algorithm, let's assume that for the given set of $g_i$s, a sum-of-squares proof of the form promised by \pref{lem:sos-refutation-polynomials} holds. Let $B' = O(n^{5d}d \log n + n^{5d} B)$ be an upper-bound on the bit complexity of the coefficients of $a_i$. By \pref{lem:sos-refutation-polynomials}, such an  event happens with probability $1-n^{-d}$ over the choice of $g_i$s and $b_i$s. We will prove that conditioned on this event, the algorithm outputs ``infeasible'' with probability $1$.

By \pref{fact:finding-pseudo-distributions}, if there is a pseudo-distribution of degree $d$ consistent with $\{g_i(x)=b_i\}_{i \leq m}$ then the sum-of-squares algorithm finds a pseudo-distribution $\mu$ such that $|\pE_{\mu}[ x^{\alpha} (g_i -b_i)] |\leq \tau$ for each monomial index $\alpha$ of degree $\leq d-D$. 
We will show that there does not exist a pseudo-distribution satisfying the latter condition. Thus, the SDP solver must output ``infeasible'' as desired.

Assume for the sake of contradiction that for $\tau = 0.5 \cdot 2^{-B'} (n+1)^{-d} m^{-1}$, there is a pseudo-distribution $\mu$ satisfying $|\pE_{\mu}[(g_i-b_i)x^{\alpha}]| \leq \tau$ for every $i \leq m$ and every monomial index $\alpha$ of degree $\leq d-D$. Then, since all $\leq (n+1)^d$ coefficients of each of the $a_i$ are of bit complexity at most $B'$, the pseudo-expectation under $\mu$ of the RHS of \pref{eq:sos-proof-eq} can be upper-bounded by:

\[
\left|\sum_{i = 1}^m \pE_{\mu}[a_i (g_i-b_i)]\right| \leq \sum_{i\leq m, \alpha} 2^{B'} \tau \leq m (n+1)^d 2^{B'} \tau  \leq 0.5\mper
\]
On the other hand, the pseudo-expectation under $\mu$ of the LHS satisfies: $|\pE[-1]| = 1$. This is a contradiction. Thus, there is no such pseudo-distribution $\mu$. 
\end{proof}

\subsection{Proof of \pref{lem:sos-refutation-polynomials}}

\paragraph{Generated ideals.}  Our analysis relies on the key idea of \emph{generated ideals} and their completeness that we define and discuss below. Intuitively speaking, given a set of constraints $\calF = \{f_1=0,f_2=0,\ldots, f_m=0\}$, the generated ideal of $\calF$ at degree $d$ is the set of all degree-$d$ polynomials that the sum-of-squares proof system (and in fact, the \nullsatz proof system) can infer to be $0$ at any simultaneous solutions of $\calF$. The following definition captures this idea.



\begin{definition}[Generated Ideal at Degree $d$]
Let $D,d\in \N$ and $D\leq d$.
Let $\calF = \{f_1,\dots, f_m\}$ be a set of degree-$D$ polynomials.
The \emph{generated ideal} of $\calF$ at degree $d$ is defined as the following set of degree-$d$ polynomials:
\begin{equation*}
    \GenIdeal{d}{\calF} \seteq \{a_1 f_1 + \cdots + a_m f_m : \forall i \text{ } \deg(a_i)\leq d-D \} \mper
\end{equation*}
We say that the generated ideal is \emph{complete} at degree $d$ if $\calP_d \subseteq \GenIdeal{d}{\calF}$ where $\calP_d$ is the set of all homogeneous degree-$d$ polynomials.
\end{definition}


One important consequence of completeness of generated ideals at degree $d$ is the following important lemma that shows that every homogeneous degree-$d$ polynomial can be written as a polynomial combination of the $f_i$s such that the coefficients of all the polynomials appearing in the representation are of polynomial bit complexity. 

\begin{lemma}[Low-Bit Complexity Representations in Complete Generated Ideals] \label{lem:low-bit-complexity-rep-arbitrary-gen-ideals}
Let $D,d\in \N$ and $D\leq d$.
Let $\calF = \{g_1,\dots, g_m\}$ be a set of degree-$D$ polynomials with rational coefficients of bit complexity $B$ such that the generated ideal $\GenIdeal{d}{\calF}$ is complete.
Let $N_{d-D} \leq n^{d-D}$ be the number of all monomials in $x_1, x_2,\ldots, x_n$ of total degree exactly $d-D$.
Let $f$ be an arbitrary homogeneous polynomial of degree $d$ with rational coefficients of bit-complexity $B$. 

Then, there is a vector $v \in \Q^{m \cdot N_{d-D}}$ with entries of bit complexity at most $O(n^{5d}d \log n + n^{5d} B)$ such that $\sum_{i \leq m, \alpha} v_{i,\alpha} g_i(x) x^{\alpha} = f(x)$.
\end{lemma}

We will use the following fact that appears in a classical work of Kannan~\cite{MR807935}. 

\begin{fact}[Bit-Complexity of Solutions to Integer Systems, see Proposition 2.1 in ~\cite{MR807935}] \label{fact:kannan}
Let $Ax=u$ for $A \in \Z^{m \times n}$ and $u \in \Z^m$ be a system of $m$ linear equations in $n$ variables $x$ such that each entry of $A$ and $u$ is an integer of magnitude $\leq B$. 

Suppose that the system is soluble over $\Q$ -- i.e., there is an $x \in \Q^n$ such that $Ax = u$. 
Then, there is in fact an $x \in \Q^n$ such that $Ax=u$ where the entries of $x$ have bit complexity $O(n(B+ \log n))$.
\end{fact}

\begin{proof}[Proof of \pref{lem:low-bit-complexity-rep-arbitrary-gen-ideals}]
Since $\GenIdeal{d}{\calF}$ is complete at degree $d$ and $f$ is a homogeneous polynomial of degree $d$, $f$ must belong to $\GenIdeal{d}{\calF}$. Thus, there are polynomials $a_1, a_2, \ldots, a_m$ of degree $\leq d-k$ such that $\sum_i a_i g_i = f$. 

For each $i$, write $a_i (x) = \sum_{\alpha} a_i^{\alpha} x^{\alpha}$ where the sum ranges over monomial indices $\alpha$ of total degree $\leq d-k$. Then, we know that $f=\sum_{i,\alpha} a_i^{\alpha} x^{\alpha} g_i$. By matching the $\leq (n+1)^d$ coefficients of $f$ on both sides, we obtain a system of linear equations with rational coefficients. We are guaranteed that this system has a solution over the reals. In fact, since all the coefficients are rational numbers, we can infer that there must be a solution over the rationals. 

Each coefficient in this linear system is a sum of at most $m$ different coefficients of one from each $g_i$. Since each coefficient of each $g_i$ has bit complexity at most $B$, the coefficients of the resulting linear system have bit complexity at most $B+O(d \log n)$. 

The lowest common multiple of all the denominators appearing in the $\leq (n+1)^{2d}$ entries of the equation is at most their product that has bit complexity at most $O(n^{3d}d \log n + n^{3d} B)$. By multiplying all the equations by this integer, we obtain a system of linear equations over the integers. By \pref{fact:kannan}, such a system has a solution of bit complexity at most $O(n^{5d}d \log n + n^{5d} B)$. Thus the original system has a solution over the rationals with bit complexity at most $O(n^{5d}d \log n + n^{5d} B)$. This completes the proof.
\end{proof}


Our task thus reduces to showing that the generated ideal of the input polynomials is complete at degree $d$ when $m \geq O(n) \cdot \Paren{\frac{n}{d}}^{D-1}$.

\paragraph{Completeness of generated ideal at degree $d$.} The key to the proof of \pref{lem:sos-refutation-polynomials} is the following lemma that identifies a non-trivial $d$ such that the generated ideal at degree $d$ of a collection of $m$ random polynomials is complete. 

\begin{lemma}[Completeness of Generated Ideals]\label{lem:generating_set}
    Let $D\in \N$ be a constant, let $d,n \in \N$ such that $2 \leq D \leq d\leq n$, and let $m \geq O_D\Paren{\frac{n^D}{d^{D-1}}}$.
    Suppose $\calG = \{g_1(x) - \B{1}, \dots, g_m(x) - \B{m}\}$ is a set of $m$ degree-$D$ polynomials obtained by choosing each coefficient of each $g_i$ and each $\B{i}$ from independent $n^{2d}$-nice rational distributions. Then, the generated ideal of $\calG$ at degree $d$ is complete with probability $1-n^{-d}$.
\end{lemma}

\begin{proof}[Proof of \pref{lem:sos-refutation-polynomials} by \pref{lem:generating_set}]
    Consider the first polynomial equation $g_1(x) = \B{1}$.
    Then, $\B{1} \neq 0$ with probability $1-n^{-200d}$ since it is sampled from a $n^{2d}$-nice distribution.
    Let's condition on $\B{1} \neq 0$ in the following. 
    Let $p(x) \seteq \frac{1}{\B{1}}g_1(x)$, and let $q(x) \seteq \frac{1}{\B{1}}\sum_{k=0}^{d/D-1} p(x)^k$ (since $d$ is a multiple of $D$).
    Then, we have 
    \begin{equation} \label{eq:create-p}
    (g_1(x)-\B{1}) q(x) = p(x)^{d/D} - 1\mper
    \end{equation}
    Thus, the polynomial $p^{d/D}-1 \in \GenIdeal{d}{\calG}$ and moreover $p^{d/D}$ is a homogeneous polynomial of degree $d$. Thus, by \pref{lem:generating_set}, the following polynomial identity holds for some polynomials $a_1, a_2, \ldots, a_m$ of degree $\leq d-D$ such that each coefficient has bit-complexity $O(n^{5d}(B+\log n))$:
    \begin{equation}  \label{eq:p-is-generated}
        \sum_{i=1}^m (g_i(x)- \B{i}) a_i(x) = - p(x)^{d/D} \mcom
    \end{equation}
    Adding the identities from \pref{eq:create-p} and \pref{eq:p-is-generated}, we obtain:
    \begin{equation*}
        \sum_{i=1}^m (g_i(x)-\B{i}) a_i(x) + (g_1(x) - \B{1}) q(x) = -1 \mper
    \end{equation*}
    This completes the proof.
\end{proof}

We now focus on proving \pref{lem:generating_set}. 
\paragraph{Reduction to rank lower bounds.}
Let  $f$ be an arbitrary polynomial such that there exist polynomials $a_1, a_2, \ldots, a_m$ of degree $d-D$ such that $f = \sum_{i=1}^m (g_i-\B{i}) a_i \in \GenIdeal{d}{\calG}$. This polynomial identity holds if and only if the coefficients of $a_i$s satisfy a system of linear equations as we describe next.
To prove that $\GenIdeal{d}{\calG}$ is complete, it suffices to restrict the polynomials $a_i$ to be homogeneous degree $d-D$.

For every $i \in [m]$, let $g_i(x) = \sum_{\gamma: |\gamma|=D} \wh{g_i}(\gamma) x^{\gamma}$ where $\gamma \in \N^n$ ranges over indices of monomials in $x_1, x_2, \ldots, x_n$ of total degree $D$. Let $f(x) = \sum_{|\alpha|=d,d-D} \wh{f}(\alpha) x^{\alpha}$ and $a_i(x) = \sum_{|\beta|=d-D} \wh{a_i}(\beta) x^{\beta}$, where $\alpha,\beta \in \N^n$ are multisets indexing monomials in $x_1, x_2, \ldots, x_n$.
Then, we have
\begin{equation*}
    f(x) = \sum_{|\alpha|=d,d-D} \wh{f}(\alpha) x^{\alpha}
    = \sum_{i=1}^m \sum_{|\gamma| = D} \sum_{|\beta|=d-D} \wh{g_i}(\gamma) \cdot \wh{a_i}(\beta) x^{\beta + \gamma} 
    - \sum_{|\beta|=d-D} \sum_{i=1}^m \B{i} \cdot \wh{a_i}(\beta) x^{\beta} \mper
\end{equation*}
Comparing coefficients on both sides, we get $\wh{f} = \MnullsatzNonhom \cdot \wh{a}$, where $\wh{f}$ has dimension $\binom{n+d-1}{d} + \binom{n+d-D-1}{d-D}$ (the number of degree $d$ and $d-D$ monomials) and $\wh{a}$ has dimension $\binom{n+d-D-1}{d-D}$. 

Let's write such equations as $f$ varies over all monomials of total degree exactly $d$. If all the resulting equations admit a solution, then clearly, every homogeneous polynomial of degree $d$ is in $\GenIdeal{d}{\calG}$. The coefficient matrix $\MnullsatzNonhom$  of the resulting linear system has the following structure:

\begin{equation*}
    \MnullsatzNonhom = 
    \begin{bmatrix}
        \Mnullsatz \\
        M_{b}
    \end{bmatrix} \mper
\end{equation*}
Here, the rows of $\Mnullsatz$ and $M_{b}$ are indexed by multisets $\alpha$ with $|\alpha| = d$ and $d-D$, respectively.
The columns of $\MnullsatzNonhom$ are indexed by $(\beta, i)$ with $|\beta|=d-D$ and $i\in[m]$.
Writing out the entries of $\MnullsatzNonhom$ explicitly:
\begin{equation*} \label{eq:matrix_M}
    \Mnullsatz(\alpha, (\beta,i)) =
    \begin{cases}
        \wh{g_i}(\gamma) & \alpha = \beta + \gamma \textnormal{ where $|\gamma|=D$} \\
        0 & \textnormal{otherwise}
    \end{cases}
    ,\quad
    M_{b}(\alpha', (\beta,i)) =
    \begin{cases}
        -\B{i} & \alpha' = \beta \\
        0 & \textnormal{otherwise}
    \end{cases}.
    \numberthis
\end{equation*}

To prove \pref{lem:generating_set}, it suffices to show that $\MnullsatzNonhom$ has full row rank.
We will prove this by showing that that the rows of $\MnullsatzNonhom$ are linearly independent.

\begin{lemma} \label{lem:M_rows_lin_indep}
    Let $D\in \N$ be a constant, let $d,n \in \N$ such that $2 \leq D \leq d\leq n$, and let $B \geq n^{2d}$.
    Consider the matrix $\MnullsatzNonhom$ defined in \pref{eq:matrix_M}, where each nonzero entry is sampled from a $B$-nice rational distribution.
    If $m\geq O_D\Paren{\frac{n^D}{d^{D-1}}}$, then the rows of $\MnullsatzNonhom$ are linearly independent with probability $1- n^{-d}$.
\end{lemma}

\begin{remark}
    Observe that $m$ must be at least $\binom{n+d-1}{d} / \binom{n+d-D-1}{d-D}+1$ for $\MnullsatzNonhom$ to have more columns than rows.
    Thus, for small $d$ (e.g.\ $d=o(n)$), $m \geq \Omega(\frac{n^D}{d^D})$ is necessary for the generated ideal of $\calG$ at degree $d$ to be complete.
\end{remark}

\pref{lem:generating_set} is an immediate corollary of \pref{lem:M_rows_lin_indep}.
We proceed to prove \pref{lem:M_rows_lin_indep} in the next section.

\subsection{Rank lower bound by row-decomposition of \texorpdfstring{$\MnullsatzNonhom$}{MG,b}}

To prove that $\MnullsatzNonhom$ is full row rank, it's enough to work with an appropriate permutation of rows/columns and delete any column from $\MnullsatzNonhom$. If the modified matrix is full row rank, then the original matrix is full row rank as well.

The main insight in the proof is that although $\MnullsatzNonhom$ is difficult to analyze, we can extract square submatrices $M_1,\dots, M_N$ of $\MnullsatzNonhom$ that are full rank, and more importantly, can be ``stitched together'' to show that $\MnullsatzNonhom$ is full row rank.
To do so, we define the following,

\begin{definition}[\RankDecomposition] \label{def:row_rank_decomposition}
    We say that the collection $(M_1,\dots, M_N)$ of square submatrices of $\MnullsatzNonhom$ is a \RankDecomposition of $\MnullsatzNonhom$ if
    \begin{enumerate}
        \item they cover all the rows of $\MnullsatzNonhom$ (i.e.\ each row of $\MnullsatzNonhom$ appears in at least one $M_i$),
        \item they have disjoint columns of $\MnullsatzNonhom$ (i.e.\ no column of $\MnullsatzNonhom$ appears in more than one $M_i$),
        \item the entries in the diagonal of each $M_i$ are independent of the off-diagonal entries of $M_i$ and the entries of $M_j$ for every $j\neq i$.
    \end{enumerate}
\end{definition}

The following lemma illustrates why the existence of a \RankDecomposition suffices to prove that $\MnullsatzNonhom$ is full row rank.

\begin{lemma}
    \label{lem:block_matrix_independence}
    Let $A, B$ be submatrices of a matrix $M$ such that $A$ is full row rank and $A,B$ have disjoint columns.
    Let $M'$ be the submatrix of $M$ with rows (columns, respectively) equal to the union of rows (columns, respectively) of $A,B$.
    Suppose further that $B$ is $K\times K$ for some $K \leq n^{2d}$ and $B = B' + g \Id$, where $g$ is a scalar sampled from a $n^{2d}$-nice rational distribution independent of $B'$ and the other entries in $M'$.
    Then, $M'$ is full row rank with probability $1- n^{-100d}$.
\end{lemma}
\begin{proof}
    First, we write $M'$ (up to permutations of rows and columns) as
    \begin{equation*}
        M' = 
        \begin{bmatrix}
            A' & C_2 \\
            C_1 & B
        \end{bmatrix}
        =
        \begin{bmatrix}
            A' & C_2 \\
            C_1 & B' + g \Id
        \end{bmatrix},
    \end{equation*}
    where $A'$ is the matrix $A$ with the rows that overlap with $B$ removed (those rows are now in $C_1$).
    $A'$ may not be square, but since $A'$ is still full row rank (the rows are linearly independent), we may delete some columns from $A'$ (and $C_1$) such that $A'$ is square and full rank.
    Hence, we may assume that $A'$ and $M'$ are square matrices without loss of generality.

    $A'$ being full rank implies that $(A')^{-1}$ exists.
    Then, $M'$ is full rank if and only if the Schur complement
    \begin{equation*}
        B - C_{1} (A')^{-1} C_{2} = g \Id + B' - C_{1} (A')^{-1} C_{2}
    \end{equation*}
    is full rank.
    Suppose not, then the matrix $g \Id + B' - C_{1} (A')^{-1} C_{2}$ is rank-deficient, which implies that $g$ is an eigenvalue of $C_{1} (A')^{-1} C_{2} - B'$. 
    However, since $g$ is sampled from a $n^{2d}$-nice distribution and is independent of $C_1, C_2, A',B'$, the probability that $g$ is exactly one of the $K$ eigenvalues is $\leq K n^{-200d} \leq n^{-100d}$.
\end{proof}

As an immediate corollary,

\begin{corollary} \label{cor:decomposition_implies_full_rank}
    Let $d \in \N$ and $B\geq n^{2d}$.
    If there exists a \RankDecomposition $(M_1,\dots,M_N)$ of $\MnullsatzNonhom$ for $N \leq n^{2d}$, then $\MnullsatzNonhom$ is full row rank with probability $1- n^{-d}$.
\end{corollary}
\begin{proof}
    We apply \pref{lem:block_matrix_independence} inductively to $M_1,M_2, \ldots,M_N$.
    Each submatrix $M_i$ has dimension at most $n^{2d}$, thus by the union bound, $\MnullsatzNonhom$ is full row rank with probability $1- N n^{-100d} \geq 1- n^{-d}$.
\end{proof}

Thus, to prove \pref{lem:M_rows_lin_indep}, it suffices to construct a \RankDecomposition.
For clarity of exposition, we will first prove \pref{lem:M_rows_lin_indep} for the special case of $D=2$ in the subsequent sections, and then show how the ideas extend to the case of $D>2$ in \pref{sec:nullsatz_degree_D}.

\subsection{Proof of \pref{lem:M_rows_lin_indep}, \texorpdfstring{$D=2$}{D=2} case}

Recall that the rows and columns of $\MnullsatzNonhom$ are indexed by $\alpha$ and $(\beta,i)$ respectively, where $\alpha,\beta$ are multisets with $|\alpha| = d$ or $d-2$, $|\beta| = d-2$, and $i\in [m]$.
To ensure that the decomposition have disjoint columns, the submatrices will be constructed using different $i$, i.e.\ selected from disjoint subsets of $[m]$.
Thus, we need $m$ to be sufficiently large so that we have enough ``fresh random equations'' to select from.
Using different $i$ also ensures that each (random) submatrix is independent of each other, especially the diagonal entries.
All other columns not present in the decomposition are ignored since we can delete columns arbitrarily.

\paragraph{Covering rows of $\Mnullsatz$.}

To extract a submatrix for the decomposition, we first select a pair $\gamma = \{j_1, j_2\} \subseteq [n]$ ($j_1=j_2$ is allowed) and consider all multisets $\alpha$ such that $\alpha = \beta \cup \gamma$ where $|\beta| = d-2$, and pick one ``fresh'' $i\in[m]$.
This gives a square submatrix $A_{\gamma}$ where the columns are indexed by $\beta$ and the rows are indexed by $\alpha = \beta \cup \gamma$ and the entries are defined to be:
\begin{equation*}
    A_{\gamma}(\alpha,\beta) =
    \begin{cases}
       \wh{g_i}(\alpha \setminus \beta) & \textnormal{ if $\beta \subset \alpha$,} \\
        0 & \textnormal{ otherwise.}
    \end{cases}
\end{equation*}
For example, say $\gamma = \{1,2\}$ and $d=4$; the first three columns are indexed by $\{1,1\}, \{1,2\}, \{1,3\}$, and the first three rows are indexed by $\{1,1,1,2\}, \{1,1,2,2\}, \{1,1,2,3\}$.
\begin{equation*}
    A_{\{1,2\}} = 
    \begin{bmatrix}
        \wh{g_i}(\{1,2\}) & \wh{g_i}(\{1,1\}) & 0 & \cdots & 0\\
        \wh{g_i}(\{2,2\}) & \wh{g_i}(\{1,2\}) & 0 & \cdots & 0 \\
        \wh{g_i}(\{2,3\}) & \wh{g_i}(\{1,3\}) & \wh{g_i}(\{1,2\}) & \cdots & 0 \\
        \vdots & \vdots & \vdots & \ddots & \vdots \\
        0 & 0 & 0 & \cdots & \wh{g_i}(\{1,2\})
    \end{bmatrix}.
\end{equation*}

Note that there are non-zero off-diagonal entries, but they are all independent of the diagonal entries $\wh{g_i}({\gamma})$.
By \pref{lem:block_matrix_independence}, $A_{\gamma}$ is full rank with high probability and satisfies the conditions of the \RankDecomposition (\pref{def:row_rank_decomposition}).
Crucially, all multisets $\alpha$ containing $\gamma$ are covered by $A_{\gamma}$.

Now, to construct the \RankDecomposition, we will select pairs $\gamma_1,\dots,\gamma_N$ such that $|\gamma_k| = 2$ and that $A_{\gamma_1},\dots, A_{\gamma_N}$ cover all $\alpha$s with $|\alpha| = d$.

\begin{lemma} \label{lem:cover_multisets}
    Let $2 \leq d \leq n$. There exist pairs $\gamma_1,\dots,\gamma_N$ for $N \leq \frac{n^2}{2(d-1)} + O(n)$ such that the rows of $A_{\gamma_1},\dots, A_{\gamma_N}$ cover all multisets $\alpha$ of size $d$.
\end{lemma}
\begin{proof}
    First, we split $[n]$ into $d-1$ buckets, each bucket contains at most $\Ceil{\frac{n}{d-1}}$ items.
    Within each bucket, we choose all pairs in the bucket, giving us $\binom{\ceil{\frac{n}{d-1}}}{2} + \ceil{\frac{n}{d-1}}$ pairs.
    The total number of pairs
    \begin{equation*}
        (d-1) \binom{\ceil{\frac{n}{d-1}}}{2} \leq \frac{n^2}{2(d-1)} + O(n),
    \end{equation*}
    using the fact that $\Ceil{\frac{n}{d-1}} \leq \frac{n}{d-1}+1$.

    Now, it suffices to prove that all $\alpha$s are covered.
    Observe that any multiset $\alpha$ that intersects a bucket in more than 1 element must be covered: if the intersection contains $\{j_1,j_2\}$, then $A_{\{j_1,j_2\}}$ covers $\alpha$.
    Thus, any uncovered $\alpha$ can only have 1 element in each bucket.
    However, there are only $d-1$ buckets whereas $|\alpha| = d$, hence every $\alpha$ must be covered.
\end{proof}

\paragraph{Covering rows of $M_{b}$.}

We use a single submatrix to cover all rows of $M_{b}$ (indexed by $\alpha$ with $|\alpha| = d-2$).
Since $M_{b}(\alpha, (\beta,i)) = - \B{i}$ when $\alpha = \beta$ and $0$ otherwise, we simply take the submatrix of a single $i$:
\begin{equation*}
    B = -\B{i} \cdot \Id.
\end{equation*}
$B$ is full rank and cover all rows in $M_{b}$.

\paragraph{Putting things together.}

Recall that $\MnullsatzNonhom$ being full row rank (\pref{lem:M_rows_lin_indep}) implies that the generated ideal of $\calG = \{g_1(x) - \B{1}, \dots, g_m(x) - \B{m}\}$ at degree $d$ is complete (\pref{lem:generating_set}), which then implies our refutation result (\pref{lem:sos-refutation-polynomials}).

\begin{proof}[Proof of \pref{lem:M_rows_lin_indep}, $D=2$ case]
By \pref{lem:cover_multisets}, the submatrices $A_{\gamma_1},\dots, A_{\gamma_N}$ cover the rows of $\Mnullsatz$, and $B$ covers the rows in $M_{b}$.
Together they form a valid \RankDecomposition of $\MnullsatzNonhom$.

The total number of equations required is
\begin{equation*}
    N+1 \leq \frac{n^2}{2(d-1)} + O(n).
\end{equation*}
Thus, by \pref{cor:decomposition_implies_full_rank}, as long as $m \geq \frac{n^2}{2(d-1)} + O(n)$, $\MnullsatzNonhom$ is full row rank with probability $1 - n^{-d}$.
\end{proof}

\subsection{Proof of \pref{lem:M_rows_lin_indep}, \texorpdfstring{$D>2$}{D>2} case}
\label{sec:nullsatz_degree_D}

The proof strategy is very similar to the case of $D=2$: we construct a \RankDecomposition of $\MnullsatzNonhom$ by considering the rows of $\Mnullsatz$ and $M_{b}$ separately.
We first prove the following analog of \pref{lem:cover_multisets},

\begin{lemma} \label{lem:cover_multisets_degree_D}
    Let $D\in \N$ be a constant and $d\in \N$ such that $3 \leq D \leq d \leq n$.
    There exist multisets $\gamma_1,\dots,\gamma_N$ of size $D$ for $N \leq O\Paren{\frac{n^D}{d^{D-1}}}$ such that the rows of $A_{\gamma_1},\dots, A_{\gamma_N}$ cover all multisets $\alpha$ of size $d$.
\end{lemma}
\begin{proof}
    We split $[n]$ into $t \seteq \Floor{\frac{d-1}{D-1}}$ buckets of size at most $\Ceil{\frac{n}{t}}$.
    Within each bucket, we choose all size-$D$ multisets, which gives $\binom{\ceil{\frac{n}{t}} + D-1}{D}$.
    The total number is
    \begin{equation*}
        t \cdot \binom{\ceil{\frac{n}{t}} + D-1}{D}
        \leq \frac{t}{D!} \Paren{\frac{n}{t}+D} \Paren{\frac{n}{t}+D-1} \Paren{\frac{n}{t}+D-2} \cdots
        \leq  \frac{t}{D!} \Paren{\frac{n}{t}+D-1}^D,
    \end{equation*}
    using the fact that $\Ceil{\frac{n}{t}} \leq \frac{n}{t}+1$.
    Next, we have $\frac{d}{D-1} \leq 2\cdot \Floor{\frac{d-1}{D-1}} = 2t$ since $d \geq D\geq 3$.
    Furthermore, by $(D-1)t \leq d-1 \leq n$ and $(D-1)! \geq \Paren{\frac{D-1}{e}}^{D-1}$,
    \begin{equation*}
        \frac{t}{D!} \Paren{\frac{n}{t}+D-1}^D \leq \frac{(4e)^{D}}{D} \cdot \frac{n^D}{d^{D-1}} = O_D\Paren{\frac{n^D}{d^{D-1}}}.
    \end{equation*}

    By construction, any uncovered $\alpha$ can only intersect each bucket in $D-1$ elements, hence $|\alpha|$ is at most $t (D-1) < d$, contradicting that $|\alpha| = d$.
    Therefore, all $\alpha$s are covered.
\end{proof}

Finally, the same matrix $B$ covers all rows of $M_{b}$.
Thus, we have a \RankDecomposition and are ready to prove \pref{lem:M_rows_lin_indep}.

\begin{proof}[Proof of \pref{lem:M_rows_lin_indep}]
    The submatrices $A_{\gamma_1},\dots, A_{\gamma_{N}}$ and $B$ together form a valid \RankDecomposition of $\MnullsatzNonhom$.
    By \pref{lem:cover_multisets_degree_D}, the total number of equations required is
    \begin{equation*}
        N+1 \leq O_D\Paren{\frac{n^D}{d^{D-1}}}.
    \end{equation*}
    Thus, by \pref{cor:decomposition_implies_full_rank}, as long as $m \geq O_D\Paren{\frac{n^D}{d^{D-1}}}$, $\MnullsatzNonhom$ is full row rank with probability $1-n^{-d}$.
\end{proof}
\section{Algorithmic Thresholds: Lower Bounds}
\label{sec:low_degree_hardness}

In this section, we prove a lower bound for the problem of distinguishing random polynomial systems from a carefully constructed ``planted'' distribution on random polynomial systems that admit a solution with probability $1$. This algorithmic task is formally easier than refutation: observe that any refutation algorithm for random polynomial systems also serves as a distinguishing algorithm. Our lower bounds hold for algorithms in the restricted computation model called the \emph{low-degree polynomial method} and match (up to constant factors) the trade-offs achieved by our refutation algorithm from the previous section.

Specifically, we will prove the following theorem in this section. 

\begin{theorem}[Low-Degree Hardness of Distinguishing Planted vs Null Polynomial Systems] \label{thm:low_degree_hardness}
Let $D\geq 2$ be a constant and $d,n,m \in \N$.
Let $\nu_N$ be the probability distribution of the system of degree-$D$ $n$-variate polynomial equations $\{g_i(x) = \B{i}\}_{i \in [m]}$ such that $g_i(x) = \Iprod{\G{i}, x^{\otimes D}}$ for a $D$-th order coefficient tensor $\G{i}$ such that each entry of $\G{i}$ is chosen to be an independent standard Gaussian. 
Then, for every $d \leq \frac{2n}{D}$, whenever $m \leq O_D\Paren{\frac{n^D}{d^{D-1}}}$, there exists a probability distribution $\nu_P$ supported on solvable systems of $m$ polynomial equations such that degree-$d$ polynomials fail to $(1/2)$-distinguish between $\nu_N$ and $\nu_P$.
\end{theorem}

\begin{remark}
The random polynomials appearing in the theorem above are obtained by choosing a random-entry tensor instead of choosing the coefficients of the polynomial directly. This leads to the coefficients of different monomials to have variances that differ by constant factors. This choice is convenient for our analysis but not necessary for the result to hold though we do not formally prove this. 
\end{remark}

We will prove \pref{thm:low_degree_hardness} by exhibiting an explicit planted distribution defined below.

\begin{definition}[Planted distribution $\nu_P$] \label{def:planted_distribution}
    Fix a parameter $\scaling = o\Paren{\frac{1}{d\sqrt{m}}}$, the planted distribution $\Pl$ is sampled as follows,
    \begin{enumerate}
        \item Sample $z$ uniformly from $\NormCube{n}$.
        \item For each $i\in[m]$, sample $\B{i} \sim \calN(0,1)$ independently.
        \item For each $i\in[m]$, sample tensor $\G{i} \in (\R^n)^{\otimes D}$ with i.i.d.\ standard Gaussian entries conditioned on $\brac{\G{i}, z^{\otimes D}} = \scaling \B{i}$.
    \end{enumerate}
\end{definition}

From \pref{prop:likelihood_ratio_variance}, the task of proving indistinguishability of $\nu_N$ and $\nu_P$ by low-degree polynomials reduces to analyzing the truncated low-degree likelihood $L^{\leq d}$ of the pair $\nu_N$ and $\nu_P$. We will analyze $L^{\leq d}$ by computing a Hermite expansion for it: 
    \begin{equation*}
        L^{\leq d}(G,b) = \sum_{\alpha,\beta : |\alpha|+|\beta|\leq d} \wh{L}_{\alpha,\beta} \cdot h_{\alpha}(G) h_{\beta}(b)
    \end{equation*}

To analyze $L^{\leq d}$, we will show the following key technical claim:
\begin{lemma} \label{lem:low_degree_variance}
    Let $D \geq 2$ be a constant, let $d,n,m\in \N$ such that $0 < d \leq \frac{2n}{D}$, and let $\scaling = o\Paren{\frac{1}{d\sqrt{m}}}$. Let $\{h_{\beta}\}_{|\beta|\leq d}$ be the multivariate Hermite polynomials. 
    If $m \leq O_D\Paren{\frac{n^D}{d^{D-1}}}$, then
    \begin{equation*}
        \sum_{\stackrel{\alpha, \beta:}{1\leq |\alpha| + |\beta| \leq d}} \expover{\Pl}{h_{\alpha}(G) h_{\beta}(b)}^2 \leq 1 \mper
    \end{equation*}
\end{lemma}
We finish the proof of \pref{thm:low_degree_hardness} modulo this claim:

\begin{proof}[Proof of \pref{thm:low_degree_hardness} by \pref{lem:low_degree_variance}]
    From \pref{def:distinguishing_polynomial}, it's enough to prove that $\E_{\nu_N}[(L^{\leq d})^2] \leq 2$. 
    We first write the Hermite expansion of $L^{\leq d}$, as a function of $G$ and $b$, in the (unnormalized) Hermite basis,
    \begin{equation*}
        L^{\leq d}(G,b) = \sum_{\alpha,\beta : |\alpha|+|\beta|\leq d} \wh{L}_{\alpha,\beta} \cdot h_{\alpha}(G) h_{\beta}(b)
    \end{equation*}
    where $\alpha \in \N^{m\times n \times n}$ and $\beta\in \N^{m}$ are the Hermite indices.
    Since $\{h_{\alpha}(G)h_{\beta}(b)\}_{\alpha,\beta}$ are orthogonal with respect to $\nu_N$, the degree $\leq d$ Hermite coefficients of $L^{\leq d}$ equal that of $L$.

    Thus, the Hermite coefficients $\wh{L}_{\alpha,\beta}$ can be computed as:
    \begin{equation*}
    \begin{aligned}
        \wh{L}_{\alpha,\beta}
        &= \expover{(G,b)\sim \nu_N}{ L(G,b) \cdot h_{\alpha}(G) h_{\beta}(b)} \cdot \frac{1}{\alpha! \beta!}
        = \expover{(G,b)\sim \nu_N}{ (\nu_P/\nu_N) \cdot h_{\alpha}(G) h_{\beta}(b)} \cdot \frac{1}{\alpha! \beta!}\\
        &=\expover{(G,b)\sim \nu_P}{ h_{\alpha}(G) h_{\beta}(b)} \cdot \frac{1}{\alpha! \beta!} \mper
    \end{aligned}
    \end{equation*}
    
    Note that for $\alpha,\beta = \vec{0}$ (the first coefficient), $\wh{L}_{0,0} = 1$.
    Then, by \pref{lem:low_degree_variance},
    \begin{equation*}
        \expover{\nu_N}{(L^{\leq d})^2} = \sum_{|\alpha|+|\beta|\leq d} \wh{L}_{\alpha,\beta}^2 \cdot \alpha! \beta!
        = 1 + \sum_{1\leq |\alpha|+|\beta|\leq d} \expover{(G,b)\sim \Pl}{ h_{\alpha}(G) h_{\beta}(b)}^2 \cdot \frac{1}{\alpha!\beta!} \leq 2 \mper
    \end{equation*}

    Along with \pref{prop:likelihood_ratio_variance}, this shows that the value $\E_{\nu_P}[f]$ is at most $\expover{\nu_N}{(L^{\leq d})^2}^{1/2} \leq \sqrt{2}$ for any degree $\leq d$ polynomial $f$ such that $\E_{\nu_N}[f^2] = 1$.
\end{proof}

\begin{remark}[The Importance of Scaling $\scaling$] \label{rem:importance-of-scaling}
    The planted distribution outputs a feasible system of polynomial equations $\{\brac{\G{i}, x^{\otimes D}} = \B{i}\}_{i\in[m]}$, where the satisfying assignment is $x = \frac{z}{\scaling^{1/D}}$.
    Note that $x$ has large norm: $\|x\|_2 = \frac{1}{\scaling^{1/D}}$.
    We note that our proof of indistinguishability requires that the scaling $\scaling$ be appropriately small. This is necessary. In particular, there is an efficient distinguisher if $\scaling \gg \sqrt{n/m}$.
    Given input $(G,b)$, calculate the tensor
    \begin{equation*}
        Q \seteq \sum_{i=1}^m \G{i} \cdot \sgn(\B{i}) \mper
    \end{equation*}
    For the null distribution $\nu_N$, $Q$ is distributed as $\sqrt{m} H$ where $H\in (\R^n)^{\otimes D}$ is a tensor with i.i.d.\ standard Gaussian entries.
    On the other hand, for the planted distribution, $Q = \scaling \Paren{\sum_{i=1}^m |\B{i}|} z^{\otimes D} + \calL \cdot \sqrt{m} H$, where $\calL$ is a linear operator of norm 1 operating on the flattened vector of $H$ for $\Norm{z}_2 = 1$.

    In the case of $D=2$, $\|\sqrt{m}H\| = O(\sqrt{mn})$, whereas $\|\scaling \Paren{\sum_{i=1}^m |\B{i}|} z z^\top\| = \Omega(\scaling m)$ with high probability. Thus, if $\scaling \gg \sqrt{n/m}$, then the algorithm that computes the spectral norm of $Q$ is a distinguisher for $\nu_N$ and $\nu_P$. We note that there's an analogous distinguisher based on spectral relaxations of tensor norm for $D>2$.
\end{remark}

For clarity of exposition, we will first prove \pref{lem:low_degree_variance} for the special case of $D=2$ then show that the ideas generalize to the case of $D>2$.

\subsection{Computing Hermite coefficients of \texorpdfstring{$L^{\leq d}$}{L<d} for \texorpdfstring{$D=2$}{D=2}}
The Hermite coefficients of the truncated likelihood $L^{\leq d}$ are naturally characterized if we attach a certain combinatorial interpretation to each Hermite index. Towards this goal, let's associate every index $\alpha\in \N^{m\times n\times n}$ with a \emph{labeled directed multigraph} (with self-loops allowed) with $n$ vertices and $|\alpha|$ edges with labels from $[m]$.

\paragraph{Notations.}
From here on, we will use $\s$ to denote an index in $[m]$, and $i,j$ to denote indices in $[n]$.
For each $\s\in[m]$, $\Alphas\in \N^{n\times n}$ corresponds to the adjacency matrix of the subgraph whose edges have label $\s$, hence $|\Alphas|$ is the number of edges labeled $\s$.
Furthermore, define $\Delta \in \N^n$ such that $\Delta_i \seteq \sum_{\s=1}^m \sum_{j=1}^n \Alphasij + \AlphaEntries{\s}{ji}$ for $i\in [n]$, interpreted as the total degree of vertex $i$.
Note that $\alpha$ can have self-loops and each self-loop contributes an additive $2$ to the definition of $\Delta$.

\begin{lemma}[Hermite Coefficients of $L$]
\label{lem:exp_planted_h}
Let $\alpha\in \N^{m \times n\times n}$, $\beta \in \N^{m}$. 
Let $\Delta = \Delta(\alpha) \in \N^n$ such that $\Delta_i$ is the total degree of vertex $i$ in the labeled directed graph associated with $\alpha$. Then, if 1) $\Delta_i$ is even for all $i\in [n]$, 2) $\beta_{\s} \leq |\Alphas|$, and 3) $|\Alphas|+ \beta_{\s} \equiv 0 \pmod{2}$ for all $\s\in[m]$, then
    \begin{equation*}
    \begin{gathered}
        \expover{(G,b) \sim\Pl}{h_{\alpha}(G) h_{\beta}(b)} =
        n^{-|\alpha|} \prod_{\s=1}^m \scaledCoeff{|\Alphas|,\beta_{\s}}, \\
        \text{where } \scaledCoeff{k,\ell} \seteq \expover{g\sim \calN(0,1)}{h_k(\scaling g) h_\ell(g)} = \scaling^{\ell} \cdot \frac{k!}{(\frac{k-\ell}{2})!} \left(-\frac{1-\scaling^2}{2}\right)^{\frac{k-\ell}{2}}.
    \end{gathered}
    \end{equation*}
    Otherwise, $\expover{(G,b) \sim\Pl}{h_{\alpha}(G) h_{\beta}(b)} = 0$.
\end{lemma}

To prove \pref{lem:exp_planted_h}, we first look at the term $\scaledCoeff{k,\ell}$:

\begin{lemma} \label{lem:scaled_hermite_coefficient}
    For any $k,\ell \in \N$ and $\scaling \in [0,1]$,
    \begin{equation*}
        \scaledCoeff{k,\ell}:=\expover{g\sim \calN(0,1)}{h_k(\scaling g) h_\ell(g)} 
        = \scaling^{\ell} \cdot \frac{k!}{(\frac{k-\ell}{2})!} \left(-\frac{1-\scaling^2}{2}\right)^{\frac{k-\ell}{2}}
    \end{equation*}
    if $\ell \leq k$ and $k+\ell \equiv 0 \pmod{2}$.
    Otherwise, $\scaledCoeff{k,\ell} = 0$.
\end{lemma}
\begin{proof}
    First, let us write the function $h_k(\scaling x)$ in the Hermite basis,
    \begin{equation*}
        h_k(\scaling x) = \sum_{\ell=0}^\infty \scaledCoeff{k,\ell} \frac{h_{\ell}(x)}{\ell!} \mcom
    \end{equation*}
    such that the coefficients $\scaledCoeff{k,\ell}$ exactly equals $\expover{g\sim \calN(0,1)}{h_k(\scaling g) h_\ell(g)}$.

    Using the generating function of Hermite polynomials (\pref{fact:hermite_generating_function}), for any $x, t\in \R$,
    \begin{equation*}
        e^{\scaling xt - \frac{t^2}{2}} = \sum_{k=0}^\infty h_k(\scaling x) \frac{t^k}{k!}
        = \sum_{\ell=0}^\infty \sum_{k=0}^\infty \scaledCoeff{k,\ell} \frac{t^k}{k!} \cdot \frac{h_{\ell}(x)}{\ell!} \mper
    \end{equation*}
    On the other hand, we can rewrite the left-hand side:
    \begin{equation*}
    \begin{aligned}
        e^{\scaling xt - \frac{t^2}{2}} = e^{x\cdot \scaling t - \frac{\scaling^2 t^2}{2}} \cdot e^{- \frac{t^2}{2}(1-\scaling^2)}
        & = \sum_{\ell=0}^\infty h_{\ell}(x) \frac{(\scaling t)^\ell}{\ell!} \cdot \sum_{i=0}^{\infty} \frac{1}{i!}\left(-\frac{t^2(1-\scaling^2)}{2}\right)^i \\
        & = \sum_{\ell=0}^\infty \frac{h_{\ell}(x)}{\ell!} \sum_{i=0}^\infty \frac{\scaling^\ell}{i!} \left(-\frac{1-\scaling^2}{2}\right)^i t^{\ell+2i} \mper
    \end{aligned}
    \end{equation*}
    Matching coefficients, we see that $\scaledCoeff{k,\ell}$ is nonzero only if $k = \ell+2i$ for some $i\geq 0$, i.e.\ $\ell \leq k$ and $k+\ell \equiv 0 \pmod{2}$.
    In this case,
    \begin{equation*}
        \scaledCoeff{k,\ell} = \scaling^{\ell} \cdot \frac{k!}{i!} \left(-\frac{1-\scaling^2}{2}\right)^i \mcom
    \end{equation*}
    where $i = \frac{k-\ell}{2}$.
    This completes the proof.
\end{proof}

We will rely  on the following technical computation from \cite{GJJ20}:

\begin{lemma}[{\cite[Lemma 4.5]{GJJ20}}] \label{lem:pap_expected_h_g}
    Let $\alpha\in \N^N$, and fix $v\in \R^N$ and $b\in \R$ such that $\|v\|_2 = 1$.
    Suppose $g \in \R^N$ is sampled from $\calN(0, \Id)$ conditioned on $\brac{g,v} = b$, then
    \begin{equation*}
        \expover{g}{h_{\alpha}(g)} = v^{\alpha} \cdot h_{|\alpha|}(b).
    \end{equation*}
\end{lemma}

We are now ready to prove \pref{lem:exp_planted_h}.

\begin{proof}[Proof of \pref{lem:exp_planted_h}]
    In our planted distribution $\nu_P$, each $\Gs$ is sampled conditioned on $z^\top \Gs z = \brac{\Gs, z z^\top} = \scaling \Bs$.
    Thus, applying \pref{lem:pap_expected_h_g} with $v = z z^\top$ (a vector in $\R^{n^2}$),
    \begin{equation*}\label{eq:exp_planted_h}
    \begin{aligned}
        \expover{(G,b) \sim \Pl}{h_{\alpha}(G) h_{\beta_{\s}}(b)}
        &= \expover{z,b}{ \prod_{\s=1}^m (z z^\top)^{\Alphas} h_{|\Alphas|}(\scaling \Bs) h_{\beta_{\s}}(\Bs)} \\
        &= \expover{z\sim \NormCube{n}}{ \prod_{i=1}^n z_i^{\Delta_i} } \cdot \prod_{\s=1}^m \expover{\Bs\sim \calN(0,1)}{h_{|\Alphas|}(\scaling \Bs) h_{\beta_{\s}}(\Bs)},
    \end{aligned}
    \numberthis
    \end{equation*}
    since $\prod_{\s=1}^m (z z^\top)^{\Alphas} = \prod_{i=1}^n z_i^{\sum_{\s,j} \Alphasij + \AlphaEntries{\s}{ji}} = \prod_{i=1}^n z_i^{\Delta_i}$.

    Note that $\sum_i \Delta_i = 2|\alpha|$, thus $\expover{z}{\prod_{i} z_i^{\Delta_i}} = n^{-|\alpha|}$ if every $\Delta_i$ is even and 0 otherwise.
    Moreover, by \pref{lem:scaled_hermite_coefficient}, $\expover{\Bs\sim \calN(0,1)}{h_{|\Alphas|}(\scaling \Bs) h_{\beta}(\Bs)} = \scaledCoeff{|\Alphas|,\beta_{\s}}$ if $\beta_{\s} \leq |\Alphas|$ and have the same parity, and 0 otherwise.
\end{proof}

\subsection{Bounding Hermite coefficients of \texorpdfstring{$L^{\leq d}$}{L<d} for \texorpdfstring{$D=2$}{D=2}}
\label{sec:bounding_variance}

In this section, we prove \pref{lem:low_degree_variance} for the special case of $D=2$.
We first give a sketch.

\paragraph{Proof sketch.}
To begin, we divide the $\alpha$s based on $|\alpha|$ (number of edges $e$) and write the summation as
\begin{equation*}
    \sum_{\nedge=1}^d \sum_{\alpha: |\alpha|=\nedge} \sum_{\beta: |\beta|\leq d-\nedge} \expover{\Pl}{h_{\alpha}(G) h_{\beta}(b)}^2.
\end{equation*}

We upper bound the above in the following steps.
First, we show that for a fixed $\alpha$, the innermost sum is dominated by the $\beta$s where $\beta_{\s} = 0$ or $1$ if $|\Alphas|$ is even of odd, respectively (\pref{lem:sum_of_scaledCoeff}).
Moreover, any odd $|\Alphas|$ introduces an extra factor of $\scaling^2$ (\pref{cor:beta_contribution}).
Thus, in the end the terms where $|\Alphas|$ are all even dominate if $c$ is appropriately small.

Next, recall that $\Delta_i$ must be even by the condition in \pref{lem:exp_planted_h}.
We show that for all $|\alpha| = \nedge$, the dominating terms are the $\alpha$s with $\Delta_i = 2$ and $|\Alphas| = 2$ for all nonzero $\Delta_i$ and $|\Alphas|$ (\pref{lem:graph_upper_bound}, \pref{lem:edge_labels_contribution}).
Viewing $\alpha$ as a graph, the dominating terms are the graphs with $\nedge$ edges and $\nedge$ vertices such that each vertex has degree 2 and each edge label appears exactly twice.

For the sake of a clean sketch, let's ignore all other terms.
The number of 2-regular graphs with $\nedge$ edges is $\leq 2^{2\nedge}$, and there are $n^{\nedge}$ ways to label the vertices.
For edge labels, we choose $\nedge/2$ labels from $[m]$ and assign to the $\nedge$ edges, thus there are $m^{\nedge/2} \nedge^{\nedge/2}$ ways to do so.
Finally, we multiply by $n^{-2\nedge}$ (the coefficient in \pref{lem:exp_planted_h}) and summing from $\nedge = 2$ to $d$, we get
\begin{equation*}
    \sum_{\nedge \geq 2, \textnormal{ even}} O\left(\frac{m\nedge}{n^2}\right)^{\nedge/2}
    \leq \sum_{\nedge \geq 2, \textnormal{ even}} \left(\frac{\nedge}{2d}\right)^{\nedge/2} \leq 1,
\end{equation*}
when $m = O(\frac{n^2}{d})$. This completes the sketch.

\paragraph{Contributions from $\beta$ for fixed $\alpha$.}
Suppose we fix an $\alpha$ with $|\alpha|=\nedge$.
Note that we must have $\beta_\s \leq |\Alphas|$ due to the condition in \pref{lem:exp_planted_h}.
Thus,
\begin{equation*} \label{eq:beta_contribution}
    \sum_{\beta: |\beta|\leq d-\nedge} \expover{\Pl}{h_{\alpha}(\bG) h_{\beta}(\bb)}^2
    = n^{-2\nedge} \sum_{\beta: |\beta|\leq d-\nedge} \prod_{\s=1}^m \scaledCoeff{|\Alphas|,\beta_\s}^2
    \leq n^{-2\nedge} \prod_{\s=1}^m \sum_{\beta_\s \leq |\Alphas|} \scaledCoeff{|\Alphas|,\beta_\s}^2.
    \numberthis
\end{equation*}

Next, we show that the dominating term is when $\beta_{\s} = 0$ or $1$ for all $\s\in[m]$ (depending on the parity of $|\Alphas|$).

\begin{lemma} \label{lem:sum_of_scaledCoeff}
    For any $k > 0$ and $\scaling = o(\frac{1}{\sqrt{k}})$,
    \begin{equation*}
        \sum_{\ell \leq k} \scaledCoeff{k,\ell}^2 \leq (1+o(1)) \cdot
        \begin{cases}
            ((k-1)!!)^2 & \textnormal{if $k$ is even}, \\
            \scaling^2 (k!!)^2 & \textnormal{if $k$ is odd}.
        \end{cases}
    \end{equation*}
\end{lemma}

\begin{proof}
    Using \pref{lem:scaled_hermite_coefficient},
    \begin{equation*}
        \sum_{\ell \leq k} \scaledCoeff{k,\ell}^2
        \leq \sum_{\stackrel{\ell \leq k}{\ell+k \equiv 0 \mod{2}}} \scaling^{2\ell} \left( \frac{k!}{(\frac{k-\ell}{2})!}\right)^2 \left(\frac{1-\scaling^2}{2}\right)^{k-\ell}.
    \end{equation*}
    Let $a_{\ell}$ be the summand.
    We have that $\frac{a_{\ell+2}}{a_{\ell}} = \scaling^4 (\frac{k-\ell}{2})^2(\frac{2}{1-\scaling^2})^2$, which is $o(1)$ if $\scaling = o(\frac{1}{\sqrt{k}})$.
    Thus, the term with the smallest $\ell$ in the summation dominates.
    If $k$ is even, then $\ell=0$ dominates,
    \begin{equation*}
        \sum_{\ell \leq k} \scaledCoeff{k,\ell}^2
        \leq (1+o(1)) \cdot \left( \frac{k!}{(\frac{k}{2})!}\right)^2 2^{-k} = (1+o(1)) \cdot ((k-1)!!)^2.
    \end{equation*}
    If $k$ is odd, then $\ell=1$ dominates,
    \begin{equation*}
        \sum_{\ell \leq k} \scaledCoeff{k,\ell}^2
        \leq (1+o(1)) \cdot \scaling^2 \left( \frac{k!}{(\frac{k-1}{2})!}\right)^2 2^{-(k-1)} = (1+o(1)) \cdot \scaling^2 (k!!)^2.
        \qedhere
    \end{equation*}
\end{proof}

As an immediate corollary, we can upper bound \pref{eq:beta_contribution} based on the parity of $|\Alphas|$:

\begin{corollary} \label{cor:beta_contribution}
    Fix an $\alpha \in \N^{m\times n \times n}$ with $|\alpha| = \nedge \leq d$.
    Let $\odd{\alpha} = \{\s\in[m]: |\Alphas| \textnormal{ odd}\}$, and $\even{\alpha} = \{\s\in[m]: |\Alphas| > 0, \textnormal{ even}\}$.
    Then,
    \begin{equation*}
        \sum_{\beta: |\beta|\leq d-\nedge} \expover{\Pl}{h_{\alpha}(G) h_{\beta}(b)}^2
        \leq n^{-2\nedge} \cdot \prod_{\s\in\odd{\alpha}} \scaling^2 (|\Alphas|!!)^2 \prod_{\s\in \even{\alpha}} ((|\Alphas|-1)!!)^2\mper
    \end{equation*}
\end{corollary}

\paragraph{Contributions from $|\alpha| = \nedge$.}

Fix the number of edges $\nedge\leq d$, we upper bound the contribution of all $\alpha$ with $|\alpha| = \nedge$.
The nonzero condition of \pref{lem:exp_planted_h} means that the only nonzero terms are the $\alpha$s (viewed as graphs) where
each vertex has even degree (counting each self-loop twice).
To upper bound the total contribution of all such graphs, we
\begin{enumerate}
    \item upper bound the number of graphs with even degrees where the vertices have labels in $[n]$,
    \item upper bound the contributions from assigning labels in $[m]$ to the edges.
\end{enumerate}
Note that the contribution of each $\alpha$ can vary based on how we label the edges.


\begin{lemma}\label{lem:graph_upper_bound}
    Let $\nedge \in \N$ such that $0 < \nedge \leq d \leq n$.
    Consider directed graphs with $\nedge$ unlabeled edges (parallel edges and self-loops allowed) such that the vertices have even degrees and have distinct labels in $[n]$.
    The number of such graphs is upper bounded by $(8n)^{\nedge}$.
\end{lemma}
\begin{proof}
    We first count the number of unlabeled graphs.
    Let $\calG(\nedge,\nvertices)$ be the set of unlabeled undirected graphs with $\nedge$ edges, $\nvertices$ vertices, and has even degrees.
    We will prove an upper bound on undirected graphs. For directed graphs, we can simply multiply our upper bound by $2^{\nedge}$, since each edge can be in either direction.

    First, we look at the case when $\nvertices = \nedge$.
    In this case, all vertices must have degree 2, hence the graphs must consist of disjoint cycles and isolated vertices with self-loops.
    This is easily upper bounded by the number of ways to partition $\nedge$ identical elements.
    The number of ways to partition $\nedge$ identical elements into $j$ non-empty groups is $\binom{\nedge-1}{j-1}$. Thus,
    \begin{equation*}
        |\calG(\nedge,\nedge)| \leq \sum_{j=1}^{\nedge} \binom{\nedge-1}{j-1} = 2^{\nedge-1} \leq 2^{\nedge}.
    \end{equation*}

    For $\nvertices < \nedge$, observe that every graph in $\calG(\nedge,\nvertices)$ can be obtained by contracting $\nedge-\nvertices$ vertices from a graph in $\calG(\nedge,\nedge)$ without deleting any edge (possibly forming self-loops or parallel edges).
    The number of ways to do so can be upper bounded the number of ways to partition $\nedge$ distinct items into $\nvertices$ non-empty identical buckets, which we can bound by
    \begin{equation*}
        \binom{\nedge - 1}{\nedge - \nvertices} \cdot \frac{\nedge!}{\nvertices!} \leq \nedge^{\nedge-\nvertices} \binom{\nedge}{\nvertices}.
    \end{equation*}
    Thus,
    \begin{equation*}
        \Card{\calG(\nedge,\nvertices)} \leq \nedge^{\nedge-\nvertices} \binom{\nedge}{\nvertices} \cdot \Card{\calG(\nedge,\nedge)} \leq 2^{\nedge} \nedge^{\nedge-\nvertices} \binom{\nedge}{\nvertices}.
    \end{equation*}
    For directed graphs, we multiply the upper bound by an additional $2^{\nedge}$.

    Next, we assign labels to the vertices.
    For graphs with $\nvertices$ vertices, there are $n^{\nvertices}$ ways to assign labels.
    \begin{equation*}
        \sum_{\nvertices=1}^{\nedge} n^{\nvertices} \cdot 2^{2\nedge} \nedge^{\nedge-\nvertices} \binom{\nedge}{\nvertices}
        = (4\nedge)^{\nedge} \sum_{\nvertices=1}^{\nedge} \binom{\nedge}{\nvertices} \left(\frac{n}{\nedge}\right)^\nvertices
        \leq (4\nedge)^{\nedge} \Paren{ 1 + \frac{n}{\nedge}}^{\nedge}
        \leq (8n)^{\nedge}.
    \end{equation*}
    Here we use the fact $\nedge \leq d \leq n$.
\end{proof}

Next, fix a graph $H$, we assign labels to the edges.

\begin{lemma} \label{lem:edge_labels_contribution}
    Let $\nedge\in \N$ such that $0 < \nedge \leq d \leq n$ and let $\scaling = o\Paren{\frac{1}{d\sqrt{m}}}$.
    Let $H$ be any directed graph with $\nedge$ (unlabeled) edges and $n$ vertices, and let $A_H$ be the set of $\alpha$s that have $H$ as its graph, i.e.\ $\sum_{\s} \Alphas$ is the adjacency matrix of $H$.
    Then,
    \begin{equation}
        \sum_{\alpha \in A_H} \sum_{\beta} \expover{\Pl}{h_{\alpha}(G) h_{\beta}(b)}^2
        \leq
        \begin{cases}
            n^{-2\nedge} (2m\nedge)^{\nedge/2} & \textnormal{$\nedge$ is even} \\
            (\scaling^2 \nedge) n^{-2\nedge} (2m\nedge)^{\frac{\nedge+1}{2}} & \textnormal{$\nedge$ is odd}.
        \end{cases}
    \end{equation}
\end{lemma}

\begin{proof}
    We need to handle the odd and even $|\Alphas|$ differently.
    Recall that $\odd{\alpha} = \{\s\in[m]: |\Alphas| \textnormal{ odd}\}$ and $\even{\alpha} = \{\s\in[m]: |\Alphas| > 0, \textnormal{ even}\}$.
    Suppose we assign labels from $[m]$ to edges such that $|\even{\alpha}| = i$ and $|\odd{\alpha}| = j$ where $2i + j\leq \nedge$.
    Note that $\nedge$ and $j$ must have the same parity since $\nedge-j$ must be even.

    We first fix $i \leq \floor{\nedge/2}$.
    We will see that $j = 0$ or $1$ is the dominating term, depending on the parity of $\nedge$.
    We upper bound the contribution as follows:
    \begin{enumerate}
        \item Choose $i$ different labels for $\even{\alpha}$ and $j$ labels for $\odd{\alpha}$. The number of ways to choose is $\binom{m}{i} \binom{m-i}{j}$.

        \item Choose the $|\Alphas|$. First, set the default values: $|\Alphas| = 2$ for $\s\in \even{\alpha}$ and $|\Alphas| = 1$ for $\s\in \odd{\alpha}$.
        Next, for the other $\nedge - 2i - j$, we can add any even number to any $|\Alphas|$.
        This is the same as the number of ways $i+j$ nonnegative integers add up to $\frac{\nedge-2i-j}{2}$, which is
        \begin{equation*}
            \binom{\frac{\nedge-2i-j}{2} + (i+j)-1}{i+j-1} = \binom{\frac{\nedge+j}{2}-1}{i+j-1}.
        \end{equation*}
        
        \item Assign all $\nedge$ edges: $\nedge!$. Note that each $\s$ is double counted $|\Alphas|!$ times.

        \item For each $\s\in[m]$, the contribution is scaled by a factor given by \pref{cor:beta_contribution}.
        In this step, we also adjust the contribution due to the double counting in the previous step.
        \begin{itemize}
            \item $|\Alphas|$ even: $\frac{((|\Alphas|-1)!!)^2}{|\Alphas|!} \leq 1$.
            \item $|\Alphas|$ odd: $\scaling^2 \frac{(|\Alphas|!!)^2}{|\Alphas|!} \leq \scaling^2 |\Alphas| \leq \scaling^2 \nedge$.
        \end{itemize}
        Thus, the contribution is scaled by $n^{-2\nedge} (\scaling^2 \nedge)^j$.
    \end{enumerate}

    For a fixed $i$, the total contribution is
    \begin{equation*} \label{eq:sum_of_j}
        \sum_{\stackrel{j\leq \nedge - 2i}{\nedge - j \textnormal{ even}}}
        n^{-2\nedge} \scaling^{2j} \nedge^j \cdot \binom{m}{i} \binom{m-i}{j} \cdot \binom{\frac{\nedge+j}{2}-1}{i+j-1} \cdot \nedge! \mper
        \numberthis
    \end{equation*}
    Let $a_j$ be the summand.
    $\frac{a_{j+2}}{a_j} \leq m^2 \cdot \scaling^4 \nedge^2 \cdot (\frac{\nedge+j}{2}) (\frac{\nedge-2i-j}{2}) \leq (\scaling^2 \nedge^2 m)^2 = o(1)$ since $\scaling = o\Paren{\frac{1}{d\sqrt{m}}}$.
    Thus, the summation is dominated by $j =0$ and $1$ for even and odd $\nedge$ respectively.

    If $\nedge$ is even, then $j=0$ dominates: \pref{eq:sum_of_j} equals $(1+o(1)) n^{-2\nedge} \nedge! \cdot \binom{m}{i} \binom{\frac{\nedge}{2}-1}{\frac{\nedge}{2}-i}$.
    Summing $i$ from $1$ to $\nedge/2$, we get
    \begin{equation*}
        (1+o(1)) n^{-2\nedge} \nedge! \cdot \sum_{i=1}^{\nedge/2} \binom{m}{i} \binom{\frac{\nedge}{2}-1}{\frac{\nedge}{2}-i}
        = (1+o(1)) n^{-2\nedge} \nedge! \cdot \binom{m+\frac{\nedge}{2}-1}{\frac{\nedge}{2}}.
    \end{equation*}
    Since $m \geq n$ and $\frac{\nedge}{2} \leq \frac{n}{2} \leq \frac{m}{2}$, we can upper bound the above by $n^{-2\nedge} \nedge! \frac{(2m)^{\nedge/2}}{(\nedge/2)!}$.
    Thus,
    \begin{equation*}
        \sum_{\alpha \in A_H} \sum_{\beta} \expover{\Pl}{h_{\alpha}(G) h_{\beta}(b)}^2
        \leq n^{-2\nedge} (2m)^{\nedge/2} \frac{\nedge!}{(\nedge/2)!}
        \leq n^{-2\nedge} (2m\nedge)^{\nedge/2}.
    \end{equation*}

    If $\nedge$ is odd, then $j=1$ dominates: \pref{eq:sum_of_j} equals $(1+o(1)) n^{-2\nedge} \nedge! (\scaling^2 \nedge) (m-i) \cdot \binom{m}{i} \binom{\frac{\nedge-1}{2}}{i}$.
    In this case, we sum $i$ from $0$ to $\frac{\nedge-1}{2}$,
    \begin{equation*}
        (1+o(1)) n^{-2\nedge} \nedge! (\scaling^2 \nedge) \cdot \sum_{i=0}^{\frac{\nedge-1}{2}} (m-i) \binom{m}{i} \binom{\frac{\nedge-1}{2}}{\frac{\nedge-1}{2}-i}
        \leq (1+o(1)) n^{-2\nedge} \nedge! (\scaling^2 \nedge m) \cdot \binom{m+\frac{\nedge-1}{2}}{\frac{\nedge-1}{2}}.
    \end{equation*}
    Similar analysis shows that
    \begin{equation*}
        \sum_{\alpha \in A_H} \sum_{\beta} \expover{\Pl}{h_{\alpha}(G) h_{\beta}(b)}^2
        \leq (\scaling^2 \nedge) n^{-2\nedge} (2m\nedge)^{\frac{\nedge+1}{2}}.
    \end{equation*}
    This completes the proof.
\end{proof}

\paragraph{Proof of \pref{lem:low_degree_variance}.}
Now, it suffices to sum up the contributions of $\nedge$ from $1$ to $d$.

\begin{proof}[Proof of \pref{lem:low_degree_variance}]
    Combining \pref{lem:graph_upper_bound} and \pref{lem:edge_labels_contribution},
    in total we have,
    \begin{equation*}
    \begin{gathered}
        \sum_{\nedge> 0, \textnormal{ even}} (8n)^{\nedge} \cdot n^{-2\nedge} (2m\nedge)^{\nedge/2}
        + \sum_{\nedge \textnormal{ odd}} (8n)^{\nedge} \cdot (\scaling^2 \nedge) n^{-2\nedge} (2m\nedge)^{\frac{\nedge+1}{2}} \\
        = \sum_{\nedge > 0, \textnormal{ even}} \left(\frac{128 m \nedge}{n^2}\right)^{\nedge/2}
        + O(\scaling^2 \nedge n) \sum_{\nedge \textnormal{ odd}} \left(\frac{128 m \nedge}{n^2}\right)^{\frac{\nedge+1}{2}}.
    \end{gathered}
    \end{equation*}
    Then, setting $\scaling = o\Paren{\frac{1}{d\sqrt{m}}}$, we can ignore the odd terms.
    Moreover, take $m = \frac{n^2}{256d}$, we have
    \begin{equation*}
        \sum_{1 \leq |\alpha| + |\beta| \leq d} \expover{\Pl}{h_{\alpha}(G) h_{\beta}(b)}^2
        \leq \sum_{\nedge \geq 2, \text{ even}}^d \left(\frac{\nedge}{2d}\right)^{\nedge/2}
        \leq 1.
        \qedhere
    \end{equation*}
\end{proof}

\subsection{Generalizing to  \texorpdfstring{$D > 2$}{D>2}}

In this section, we prove \pref{lem:low_degree_variance} for arbitrary $D$.
In this case, we have polynomial equations $\gs(x) = \brac{\Gs, x^{\otimes D}} = \Bs$ for $\s\in[m]$, where $\Gs\in (\R^n)^{\otimes D}$.

For Hermite indices $\alpha \in \N^{m\times n \times \cdots \times n}$ and $\beta \in \N^{m}$, we calculate $\expover{(G,b) \sim\Pl}{h_{\alpha}(G) h_{\beta}(b)}$.
Here, we view $\alpha$ as a labeled directed $D$-uniform hypergraph with edges labeled $1,\dots, m$, and define $\Delta \in \N^n$ as the total degree of vertex $i\in[n]$.
Note that $|\Delta| = \sum_{i=1}^n \Delta_i = D|\alpha|$.
The following lemma is almost identical to \pref{lem:exp_planted_h}.

\begin{lemma}
    For $D \geq 2$, indices $\alpha\in \N^{m \times n\times \cdots \times n}$, $\beta \in \N^{m}$, and $\scaling > 0$,
    define $\Delta\in \N^n$ such that $\Delta_i$ is the total degree of vertex $i$ when viewing $\alpha$ as a labeled $D$-uniform hypergraph.
    Then, if $\Delta_i$ is even for all $i\in [n]$ and $\beta_{\s} \leq |\Alphas|$, $|\Alphas|+ \beta_{\s} \equiv 0 \pmod{2}$ for all $\s\in[m]$, then
    \begin{equation*}
    \begin{gathered}
        \expover{(G,b) \sim\Pl}{h_{\alpha}(G) h_{\beta}(b)} =
        n^{-D |\alpha|/2} \prod_{\s=1}^m \scaledCoeff{|\Alphas|,\beta_{\s}}.
    \end{gathered}
    \end{equation*}
    Otherwise, $\expover{(G,b) \sim\Pl}{h_{\alpha}(G) h_{\beta}(b)} = 0$.
\end{lemma}

\begin{proof}
    Similar to the proof of \pref{lem:exp_planted_h}, we apply \pref{lem:pap_expected_h_g} with $v = z^{\otimes D}$ (a vector in $\R^{n^D}$),
    \begin{equation*}
    \begin{aligned}
        \expover{(G,b)\sim\Pl}{h_{\alpha}(G) b^{\beta}}
        & = \expover{z,b}{ \prod_{\s=1}^m (z^{\otimes D})^{\Alphas} h_{|\Alphas|}(\scaling \Bs) h_{\beta_{\s}}(\Bs)} \\
        & = n^{-D|\alpha|/2} \prod_{\s=1}^m \scaledCoeff{|\Alphas|,\beta_{\s}},
    \end{aligned}
    \end{equation*}
    since $|\Delta| = \sum_{i=1}^n \Delta_i = D|\alpha|$.
    This completes the proof.
\end{proof}

The proof of \pref{lem:low_degree_variance} for arbitrary $D$ is almost identical to the case of $D=2$, except for counting the number of graphs with even degrees.
The following is the generalization of \pref{lem:graph_upper_bound}.

\begin{lemma}\label{lem:graph_upper_bound_degree_D}
    Let $D,d,\nedge,n \in \N$ such that $D>2$ and $0 \leq \nedge \leq d \leq \frac{2n}{D}$.
    Consider directed $D$-uniform hypergraphs with $\nedge$ unlabeled edges (parallel edges and self-loops allowed) such that the vertices have even degrees and have distinct labels in $[n]$.
    The number of such graphs is upper bounded by
    \begin{equation*}
        O(Dn)^{\frac{D \nedge}{2}} \nedge^{(\frac{D}{2}-1)\nedge}\mcom
    \end{equation*}
    if $D\nedge$ is even.
    Otherwise, there are no such graphs.
\end{lemma}
\begin{proof}
    Note that $D\nedge$ is the total degree, which must be even.
    The number of vertices $\nvertices$ can range from $1$ to $\frac{D\nedge}{2}$.
    In order to perform the counting, we view a hypergraph $H$ as a bipartite factor graph $(V,F,E)$ with left-hand side vertex set $V$ that contains the vertices of $H$, and right-hand side vertices $F$ that contains a vertex for each hyperedge. Note, in particular, that the right-degree of the bipartite graph is $D$. And all left-degrees have to be even in the hypergraphs we intend to count.

    We directly analyze the number of $\nvertices$-vertex graphs.
    \begin{enumerate}
        \item We choose $\nvertices$ labels from $[n]$, giving us $\binom{n}{\nvertices}$.
        \item We choose the left-degrees: the degrees must be even and sum to $D\nedge$.
        This is the same as the number of ways $\nvertices$ positive integers add up to $\frac{D\nedge}{2}$, which is $\binom{\frac{D\nedge}{2}-1}{\frac{D\nedge}{2}-\nvertices}$.

        \item We add edges between $V$ and $F$ while ensuring that the degrees are consistent.
        To do so, we construct a vertex set $V'$ by including $\deg(i)$ copies of vertex $i\in V$, hence $|V'| = D\nedge$.
        Then, we add edges between $F$ and $V'$ such that each $f\in F$ has degree $D$ and each $i\in V'$ has degree 1.
        Since the factor vertices are unlabeled, there are at most $\frac{(D\nedge)!}{\nedge!}$ ways to do so.
    \end{enumerate}

    \noindent Then, combining the above and summing $\nvertices$ from $1$ to $\frac{D\nedge}{2}$,
    \begin{equation*}
        \frac{(D\nedge)!}{\nedge!} \sum_{\nvertices=1}^{\frac{D\nedge}{2}} \binom{n}{\nvertices} \binom{\frac{D\nedge}{2}-1}{\frac{D\nedge}{2}-\nvertices}
        \leq \frac{(D\nedge)!}{\nedge!} \binom{n + \frac{D\nedge}{2}-1}{\frac{D\nedge}{2}}\mper
    \end{equation*}
    Using the fact that $\frac{D\nedge}{2} \leq n$ and Stirling's approximation, we can upper bound the above by
    \begin{equation*}
        O(Dn)^{\frac{D\nedge}{2}} \nedge^{(\frac{D}{2}-1)\nedge}\mper
        \qedhere
    \end{equation*}
\end{proof}

The contributions from assigning labels to edges is exactly the same as \pref{lem:edge_labels_contribution}, except with coefficient $n^{-D\nedge}$.
Thus, we are in position to prove \pref{lem:low_degree_variance}.

\begin{proof}[Proof of \pref{lem:low_degree_variance}]
    We sum over all contributions of $|\alpha|=\nedge$ from $1$ to $d$.
    Setting $\scaling = o\Paren{\frac{1}{d\sqrt{m}}}$, we can ignore the odd terms; in fact, if $D$ is odd, the odd terms are exactly zero since $D\nedge$ must be even.
    Thus, the total contribution is
    \begin{equation*}
        \sum_{\stackrel{\alpha, \beta:}{1 \leq |\alpha| + |\beta| \leq d}} \expover{\Pl}{h_{\alpha}(\bG) h_{\beta}(\bb)}^2
        \leq \sum_{\nedge \textnormal{ even}} O\Paren{\frac{D\nedge}{n}}^{\frac{D\nedge}{2}} \Paren{\frac{m}{\nedge}}^{\frac{\nedge}{2}}
        \leq 1,
    \end{equation*}
    when $m \leq O_D\Paren{\frac{n^D}{d^{D-1}}}$.
    This completes the proof.
\end{proof}
\section{Algorithmic Thresholds at Degree 2}
\label{sec:degree_2}

In this section, we give a short proof of the following theorem that gives a sharp threshold on the number of quadratic equations $m$ required for the existence of degree-$2$ SoS refutations. 
\begin{theorem}\label{thm:deg_2_sos}
For any homogeneous quadratic polynomials $g_1, g_2, \ldots, g_m$ in $x_1, x_2, \ldots, x_n$ and real numbers $b_1, b_2,\ldots, b_m$, let $\DegreeTwoNonhom$ be the degree-$2$ SoS relaxation of the system of constraints $\{g_i(x)=\B{i}\}_{i \leq m}$. Specifically, let $\G{i} \in \R^{n \times n}$ be matrices such that $g_i(x) = x^{\top} \G{i} x$ for each $i\in[m]$. Then, the degree-$2$ SoS relaxation is the following SDP:
\begin{equation} \label{eq:sdp-deg2}
X \succeq 0, \text{ } \tr(\G{i} X) = b_i\ \text{ for all } 1 \leq i \leq m \mper
\end{equation}
Suppose each coefficient of $g_i$ is chosen to be an independent draw from the standard Gaussian distribution $\cN(0,1)$. 
Then, there is an absolute constant $C$ such that if $m \geq \frac{n^2}{4}+Cn\log n$, the semidefinite program above is infeasible with probability at least $0.49$. 
On the other hand if $m \leq \frac{n^2}{4} -Cn\log n$ then the semidefinite program above is feasible with probability at least $1-\frac{1}{n}$. 
\end{theorem}

Our proof is an immediate application of a classical work~\cite{ALMT14} on understanding phase transitions for convex programs with random data that relies on deep results from conic integral geometry~\cite{SW08}. In particular, our proof relies on the following \textit{approximate kinematic formula}.



\begin{fact}[{\cite[Theorem I]{ALMT14}}]
\label{fact:kinematic_formula}
    Fix a tolerance $\eta \in (0,1)$. Let $C$ and $K$ be convex cones in $\R^N$, and let $Q \in \R^{N \times N}$ be a uniformly random (i.e.\ Haar distributed) orthogonal matrix. Then,
\begin{equation*}
\begin{aligned}
    \delta(C) + \delta(K) \leq N - O(\sqrt{N} \log(1/\eta)) &\imp \probover{Q}{C \cap Q K  \neq \{0\}} \leq \eta; \\
    \delta(C) + \delta(K) \geq N + O(\sqrt{N} \log(1/\eta)) &\imp \probover{Q}{C \cap Q K  \neq \{0\}} \geq 1-\eta.
\end{aligned}
\end{equation*}
Here, $Q K = \{Q z\mid z \in K\}$ is the rotation of the cone $K$ by $Q$ and $\delta(C), \delta(K)$ are \emph{statistical dimensions} of the cones $C,K$ respectively.
\end{fact}

We will not define statistical dimension formally in this work but note that the statistical dimension of a subspace of dimension $r$ is $r$ and that of the cone of positive semidefinite $n \times n$ matrices is $\frac{1}{4}n(n+1)$ (see Table 3.1 of \cite{ALMT14}). For background and proofs, we refer the reader to~\cite{ALMT14}.


\begin{proof}[Proof of \pref{thm:deg_2_sos}]

    Let $S_+$ be the open convex cone of positive definite matrices. 
    Let $K$ be the linear span of the symmetric matrices $\G{1},\dots,\G{m}$ viewed as $\frac{n(n+1)}{2}$ dimensional vectors. 
    Let $K^{\perp}$ be the orthogonal complement of $K$ in $\R^{n^2}$.

    Since $0 < m < \frac{n(n+1)}{2}$, $K$ and $K^{\perp}$ have dimension $m$ and $\frac{n(n+1)}{2}-m$ with probability $1$ over the draw of the $\G{i}$s. Thus, the statistical dimension of $K$, $K^{\perp}$ is $m$ and $\frac{n(n+1)}{2}-m$ respectively. Observe that because the coefficients of $g_i$s are independent standard Gaussians, $\G{i}$s are standard Gaussian vectors, and $K$, $K^{\perp}$ are random (rotations of) subspaces of their dimension. 
    The statistical dimension of $S_+$ is $\frac{1}{4} n (n+1)$. 
    
    Applying \pref{fact:kinematic_formula} to $K$ and $K^{\perp}$ with $\eta = \frac{1}{n}$ yields that there is a constant $C>0$ such that:
    \begin{enumerate}
        \item \textbf{Case 1:} If $m \geq \frac{n^2}{4} + Cn \log n$, then, there is a positive definite matrix $M_1$ in $K$. 
        \item \textbf{Case 2:} If $m \leq \frac{n^2}{4} - Cn \log n$, with probability at least $1-1/n$, there is a positive definite matrix $M_2$ in $K^{\perp}$.
    \end{enumerate}

    Let's now condition on the existence of $M_1$/$M_2$ in the two cases and analyze the SDP \pref{eq:sdp-deg2}.

    \textbf{Case 1:} Suppose for the sake of contradiction that there is a PSD $Y$ such that $\langle \G{i},Y \rangle = \B{i}$ for every $i\in[m]$. Let $M_1 = \sum_i c_i\G{i}\in K$ for $c_i\in\R$. Then, $\langle M_1, Y \rangle =\sum_i c_i \B{i}$. Now, the LHS is non-negative since $M_1, Y$ are both positive semidefinite. The RHS $\sum_i c_i \B{i}$, on the other hand, is distributed as a standard scalar Gaussian and is thus $<0$ with probability $1/2$. Thus, there can be no such $Y$ with probability at least $1/2$.

    \textbf{Case 2:} Let $M_2$ be the positive definite matrix such that $\langle M_2, \G{i} \rangle = 0$ for every $i\in[m]$. Let $Y \in \R^{n \times n}$ be any solution to $\langle \G{i}, Y \rangle =\B{i}$ for every $i$. Such a $Y$ exists since $\G{i}$s are linearly independent with probability $1$. Then, observe that for some large enough scaling $R$, $RM_2+Y$ is positive semidefinite and is feasible for the SDP \pref{eq:sdp-deg2}. 

    This completes the proof. 
\end{proof}

\section{Sum-of-Squares Lower Bounds at Degree 4}
\label{sec:degree_4_lower_bound}

In this section, we show that there is an $m = n^2/\poly(\log n)$ such that for random homogeneous quadratic polynomials $g_1, g_2, \ldots, g_m$ of degree $2$, the constraint system $\{g_i(x)=0\}_{i \leq m}$ does not admit a degree-$4$ sum-of-squares refutation. Specifically, we will establish the following dual version of such a claim:

\begin{theorem} \label{thm:main-sos-lower-bound}
Fix $m = m(n) \leq n^2/\poly(\log n)$. Let $g_1, g_2, \ldots, g_m$ be homogeneous degree-$2$ polynomials in $x_1, x_2, \ldots, x_n$ such that each coefficient of each $g_i$ is an independent draw of the standard Gaussian distribution $\cN(0,1)$. Then, with probability $1-o(1)$, there exists a degree-$4$ pseudo-distribution $\mu$ on $x_1, x_2, \ldots, x_n$ consistent with the constraint system $\{g_i(x)=0\}_{i \leq m}$. 
\end{theorem}

We will prove \pref{thm:main-sos-lower-bound} by giving an explicit construction of a pseudo-distribution $\mu$ satisfying the requirements of the theorem. Our construction of $\mu$ will rely on the standard technique of \pc and will use the planted distribution constructed in the previous section. Our analysis adapts the high-level analysis strategy invented in~\cite{GJJ20} who proved a sum-of-squares lower bound for optimizing the Sherrington-Kirkpatrick Hamiltonian. The details of this strategy in our setting are somewhat different.

\paragraph{Candidate \pd.}

We construct a candidate \pd $\mu$ based on the \pc method, using the planted distribution $\nu_P$ in \pref{def:planted_distribution} (with $D=2$ and $\scaling = 0$).
In a nutshell, the \pc method is a mechanical way to construct each entry of the candidate pseudo-moment matrix based on $\nu_P$.

\begin{definition}[Candidate \pd] \label{def:candidate-pd}
    Fix $m = m(n) \leq n^2 / \poly(\log n)$ and truncation threshold $\tau = \poly(\log n)$.
    Given $\bG$ sampled from the null distribution $\nu_N$,
    we define the \pd $\mu$ over $\NormCube{n}$ (as a function of $G$) by describing the \pe of all degree $\leq 4$ monomials:
    for $I \subseteq[n]$ and $|I| \leq 4$,
    \begin{equation*}
        \pE_{\mu}[x^I] \seteq \sum_{\stackrel{\alpha\in \N^{m\times n\times n}}{|\alpha| \leq \tau}} \expover{(\bG',\bz)\sim \nu_P}{\bz^{I} h_{\alpha}(\bG')} \cdot \frac{h_{\alpha}(\bG)}{\alpha!} \mper
    \end{equation*}
\end{definition}
Note that we have the ``normalized'' booleanity constraint $x_i^2 = \frac{1}{n}$.
Our final construction that yields \pref{thm:main-sos-lower-bound} will be obtained by a small perturbation of the construction in \pref{def:candidate-pd}.


To analyze this construction, it is helpful to study a matrix -- the moment matrix -- associated with the pseudo-distribution.

\paragraph{The Moment Matrix.} The moment matrix $\MomentMatrix$ of $\mu$ is a matrix indexed by subsets $I,J\subseteq [n]$ of size $\leq 2$ and entries defined by:
\begin{equation*}
    \MomentMatrix(I,J) \seteq \pE_{\mu}[x^{I+J}] = \sum_{\stackrel{\alpha\in \N^{m\times n\times n}}{|\alpha| \leq \tau}} \expover{(\bG',\bz)\sim \nu_P}{\bz^{I+J} h_{\alpha}(\bG')} \cdot \frac{h_{\alpha}(\bG)}{\alpha!} \mper
\end{equation*}

We can explicitly compute the coefficient of the Hermite polynomial $h_{\alpha}(\bG)$ in the above expression for $ \MomentMatrix(I,J)$ as follows.
Again, we will use $\s$ to denote an index in $[m]$ and $i,j$ to denote indices in $[n]$.
By the computation we did in the context of our low-degree lower bounds, specifically \pref{lem:exp_planted_h} (setting $\beta=0$ and $\scaling=0$), we obtain that for any $I, J\subseteq [n]$ and any $\alpha\in \N^{m\times n \times n}$,
\begin{equation*} \label{eq:lambda_explicit}
    \lambda_{\alpha,I,J} \seteq \frac{1}{\alpha!}\expover{(\bG', \bz) \sim \Pl}{\bz^{I+J} h_{\alpha}(\bG')}
    = (-1)^{|\alpha|/2} n^{-|\alpha| - \frac{|I|+|J|}{2}} \prod_{\s=1}^m (|\Alphas|-1)!! \cdot \frac{1}{\alpha!}
    \numberthis
\end{equation*}
if $|\Alphas|$ is even for all $\s\in[m]$ and $\Delta_i + I_i + J_i$ is even for all $i\in[n]$ (here we denote $I_i \seteq \1\{i\in I\}$), and 0 otherwise
(recall that $\Delta \in \N^n$ where $\Delta_i \seteq \sum_{\s=1}^m \sum_{j=1}^n \Alphasij + \AlphaEntries{\s}{ji}$, interpreted as the total degree of vertex $i$).
Thus, we have
\begin{equation*}
    \begin{gathered}
    \MomentMatrix(I,J) \seteq \sum_{\stackrel{\alpha: |\alpha|\leq \tau}{|\Alphas| \text{ even, } \Delta_i +I_i + J_i \text{ even}}}
    \lambda_{\alpha,I,J} h_{\alpha}(\bG).
    \end{gathered}
\end{equation*}
Note the $1/\alpha!$ factor in \pref{eq:lambda_explicit} is there because we use the unnormalized Hermite polynomials.
By an upper bound on the double factorial (\pref{fact:double_factorial_bound}),
\begin{equation*} \label{eq:lambda_bound}
    |\lambda_{\alpha,I,J}| \leq n^{-|\alpha|-\frac{|I|+|J|}{2}} \left(\frac{|\alpha|}{2}\right)^{|\alpha|/2}.
    \numberthis
\end{equation*}

Keep in mind that $\MomentMatrix$ will only approximately satisfy the conditions of a pseudo-moment, e.g.\ $\MomentMatrix(\varnothing, \varnothing) \approx 1$ and $\MomentMatrix(\{i\},\{i\}) \approx \frac{1}{n}$.
However, we will show that we can ``fix'' the moment matrix such that it represents a valid \pd and satisfies all constraints. Note that the positivity property, i.e., $\pE_{\mu}[q^2] \geq 0$ for every degree-$2$ polynomial $q$ is equivalent to the positive semidefiniteness of the moment matrix $\MomentMatrix$ of $\mu$.

\begin{lemma} \label{lem:M_E_PSD}
    There exist constants $C_1, C_2 > 0$ such that if $m = n^2 / \log^{C_1} n$ and $\tau = \log^{C_2} n$, then there exists a correction matrix $\calE$ such that $\MomentMatrix - \calE$ satisfies all constraints $\{g_\s(x) = 0\}_{\s\leq m}$ and that $\MomentMatrix - \calE \succeq 0$.
\end{lemma}

This lemma is the bulk of the proof of \pref{thm:main-sos-lower-bound} and requires a relatively technically involved argument. In order to prove PSDness of $\MomentMatrix$ we need to analyze its spectrum. This is somewhat challenging as the matrix has dependent random entries. Our proof relies on a strategy invented in previous works (starting with~\cite{DBLP:conf/focs/BarakHKKMP16} and built further in~\cite{HKP+17,GJJ20}) that decomposes moment matrices built via pseudo-calibration into a sum of structured random matrices (called \emph{graph matrices}) that are helpful in spectral analysis. We start with a brief background of graph matrices specialized to our setting before giving an outline of our proof. 


\subsection{Background on graph matrices}
\label{sec:graph_matrices}

Our notations and definitions follow that of \cite{AMP20,GJJ20} who also studied with graphical matrices when the input data is random Gaussian.

We represent each Hermite index $\alpha\in \N^{m\times n\times n}$ as a \textit{3-uniform hypergraph} with two types of vertices: circles $\circle{}$ and squares $\square{}$.
Each square $\square{i}$ has a label $i\in[n]$, and each circle $\circle{\s}$ has a label $\s\in [m]$.
A nonzero entry $\Alphasij$ is represented by a hyperedge $(\square{i},\square{j}, \circle{\s})$.
See \pref{fig:alpha_hypergraph} for an example.
Note that the order of $\square{i}$ and $\square{j}$ matters since we allow $\AlphaEntries{\s}{ij} \neq \AlphaEntries{\s}{ji}$, but for simplicity we don't draw it out explicitly.

\begin{figure}[h!]
    \centering
    \begin{subfigure}[b]{0.4\textwidth}
        \centering
        \begin{tikzpicture}
        \draw
            node at (0,0) [squarenode] (i1) {$1$}
            node at (0,\VerticalSep) [squarenode] (i2) {$2$}
            node at (0,2*\VerticalSep) [squarenode] (i3) {$3$}
            node at (\HorizontalSep,\VerticalSep/2) [circlenode] (s1) {$1$}
            node at (\HorizontalSep,\VerticalSep*1.5) [circlenode] (s2) {$4$}
        ;
        \Hyperedge{i1}{i2}{s1}{\HorizontalSep/2}{\VerticalSep/2}{2}
        \Hyperedge{i2}{i3}{s2}{\HorizontalSep/2}{\VerticalSep*1.5}{3}
        \end{tikzpicture}
        \caption{$\AlphaEntries{1}{12} = 2$ and $\AlphaEntries{4}{23} = 3$.}
    \end{subfigure}
    \quad
    \begin{subfigure}[b]{0.4\textwidth}
        \centering
        \begin{tikzpicture}
        \draw
            node at (0,0) [squarenode] (i1) {$1$}
            node at (0,\VerticalSep) [squarenode] (i2) {$2$}
            node at (0,2*\VerticalSep) [squarenode] (i3) {$3$}
            node at (\HorizontalSep,\VerticalSep) [circlenode] (s1) {$1$}
        ;
        \ParallelEdge{i1}{s1}{\HorizontalSep/2}{\VerticalSep/2}{2}
        \Hyperedge{i2}{i3}{s1}{\HorizontalSep/2}{\VerticalSep*1.5}{}
        \end{tikzpicture}
        \caption{$\AlphaEntries{1}{11} = 2$ and $\AlphaEntries{1}{23} = 1$.}
    \end{subfigure}
    \caption{Examples of $\alpha\in\N^{m\times n \times n}$ represented as hypergraphs.}
    \label{fig:alpha_hypergraph}
\end{figure}
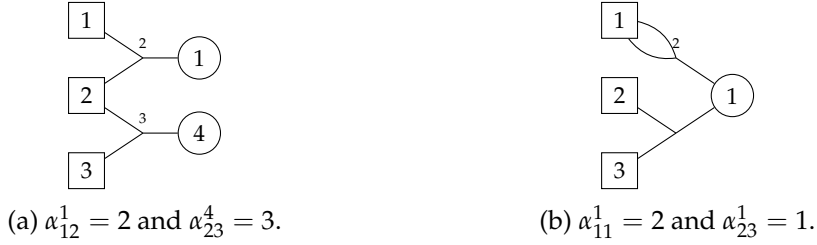

Next, we define \textit{ribbons} and \textit{shapes} (see Definitions 2.9--2.12 in \cite{GJJ20}).
Denote $\calS\seteq \{\square{i}: i\in[n]\}$ and $\calC \seteq \{\circle{\s}: \s\in[m]\}$.
A ribbon $R$ is simply a hypergraph $(V(R), E(R))$ of some $\alpha$ (as in \pref{fig:alpha_hypergraph}) with a set of ``left'' and ``right'' vertices $A_R, B_R\subseteq V(R)$.
Each ribbon defines a matrix with a single entry.
\begin{definition}[Ribbons] \label{def:ribbon}
    A ribbon is a 3-uniform hypergraph $R = (V(R), E(R), A_R, B_R)$ such that $V(R) \subseteq \calS \cup \calC$ contains labeled square and circle vertices, and $A_R, B_R\subseteq V(R)$ (not necessarily disjoint).
    The edges in $E(R)$ are labeled and must be connected to two square vertices and one circle vertex.
\end{definition}

\begin{definition}[Matrix of a ribbon]
    Let a ribbon $R = (V(R), E(R), A_R, B_R)$, and let $\alpha\in \N^{m\times n \times n}$ be the multiset represented by $(V(R), E(R))$.
    The matrix of a ribbon $M_R$, indexed by subsets of $\calS \cup \calC$, is defined as
    \begin{equation*}
        M_R(I,J) = \begin{cases}
            h_\alpha(\bG) & I = A_R, J = B_R, \\
            0 &  \text{otherwise}.
        \end{cases}
    \end{equation*}
\end{definition}

The \textit{shape} is a ribbon with the labels of each vertex removed, i.e.\ ribbons with the same hypergraph structure but different labels have the same shape.
\begin{definition}[Shape] \label{def:shape}
    A shape is a 3-uniform hypergraph $a = (V(a), E(a), U_a, V_a)$ where $V(a)$ contains unlabeled circle and square vertices and $U_a,V_a\subseteq V(a)$ (not necessarily disjoint).
    The edges in $E(a)$ are labeled and must be connected to two square vertices and one circle vertex.

    We call $U_a,V_a$ the ``left'' and ``right'' vertices.
    Moreover, define $W_a \seteq V(a) \setminus (U_a \cap V_a)$ to be the ``middle'' vertices of the shape and $\Wiso$ to be the isolated vertices in $W_a$.
\end{definition}

\begin{definition}[Graph matrix] \label{def:graph_matrix}
    The matrix of a shape $M_a$ is defined as
    \begin{equation*}
        M_a \seteq \sum_{R: \textnormal{ ribbon of shape $a$}} M_R.
    \end{equation*}
\end{definition}

Ribbons and shapes are best explained by examples.
Consider the ribbon $R$ and shape $a$ in \pref{fig:ribbons_shapes}.
The matrix $M_R$ has entries $M_R(I, J) = h_2(\GEntries{1}{12}) h_3(\GEntries{4}{23})$ if $I=\{1\}$, $J = \{3\}$, and $0$ otherwise.
The graph matrix $M_a$ is a sum of all ribbons of shape $a$, including $R$.
Thus, $M_a(\{i\}, \{j\}) = \sum_{k\in[n], k\neq i,j}\sum_{\s_1 \neq \s_2\in [m]} h_2(\GEntries{\s_1}{ik}) h_3(\GEntries{\s_2}{kj})$ for $i\neq j$.

\begin{figure}[ht!]
    \centering
    \begin{subfigure}[b]{0.4\textwidth}
        \centering
        \begin{tikzpicture}
        \draw
            node at (0,0) [squarenode] (i1) {$1$}
            node at (0,\VerticalSep) [squarenode] (i2) {$2$}
            node at (0,2*\VerticalSep) [squarenode] (i3) {$3$}
            node at (\HorizontalSep,\VerticalSep/2) [circlenode] (s1) {$1$}
            node at (\HorizontalSep,\VerticalSep*1.5) [circlenode] (s2) {$4$}
        ;
        \Hyperedge{i1}{i2}{s1}{\HorizontalSep/2}{\VerticalSep/2}{2}
        \Hyperedge{i2}{i3}{s2}{\HorizontalSep/2}{\VerticalSep*1.5}{3}
        \LeftVertices{0}{0}{1}{$A_R$}
        \RightVertices{0}{2*\VerticalSep}{1}{$B_R$}
        \end{tikzpicture}
        \caption{Ribbon $R$.}
    \end{subfigure}
    \quad
    \begin{subfigure}[b]{0.4\textwidth}
        \centering
        \begin{tikzpicture}
        \draw
            node at (0,0) [squarenode] (i1) {$i$}
            node at (0,\VerticalSep) [squarenode,separator] (i2) {$k$}
            node at (0,2*\VerticalSep) [squarenode] (i3) {$j$}
            node at (\HorizontalSep,\VerticalSep/2) [circlenode] (s1) {$\s_1$}
            node at (\HorizontalSep,\VerticalSep*1.5) [circlenode] (s2) {$\s_2$}
        ;
        \Hyperedge{i1}{i2}{s1}{\HorizontalSep/2}{\VerticalSep/2}{2}
        \Hyperedge{i2}{i3}{s2}{\HorizontalSep/2}{\VerticalSep*1.5}{3}
        \LeftVertices{0}{0}{1}{$U_a$}
        \RightVertices{0}{2*\VerticalSep}{1}{$V_a$}
        \end{tikzpicture}
        \caption{Shape $a$. $\MVS = \{\square{k}\}$.
        \label{fig:first_shape_example}
        }
    \end{subfigure}
    \caption{Example of a ribbon and shape. The minimum vertex separator of a shape is colored green.}
    \label{fig:ribbons_shapes}
\end{figure}
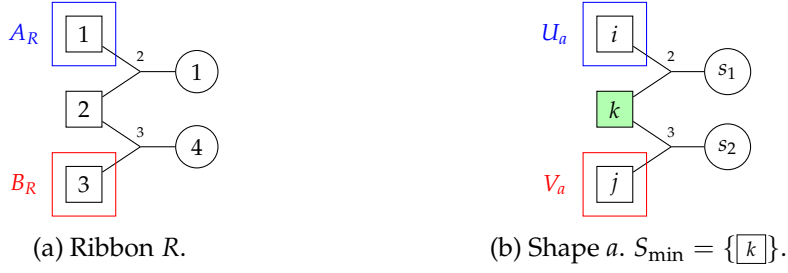

\begin{definition}[Transpose of a shape]
    The transpose of a shape $a = (V(a),E(a),U_a,V_a)$ is defined as $a^\top \seteq (V(a), E(a), V_a, U_a)$.
    This implies that $M_{a} = (M_{a^\top})^\top$.
\end{definition}

\paragraph{Graph matrix norm bounds.}
We will require spectral norm bounds of graph matrices.
We can directly use the norm bounds from \cite{AMP20}, which are obtained using the trace power method.
First, define the weights of square and circle vertices: $w(\square{}) = 1$ and $w(\circle{}) = \log_n(m)$.
This is defined such that for any shape $a$ and any subsets $S, C$ of square and circle vertices,
$n^{w(S) + w(C)} = n^{|S|} m^{|C|}$, which is roughly the number of ways you can label $S,C$;
such quantities naturally arise in trace moment calculations.

Next, we define the \textit{minimum vertex separator}:
\begin{definition}[Minimum vertex separator]
    For a shape $a$, a set $S \subseteq V(a)$ is a vertex separator if all paths from $U_a$ to $V_a$ pass through $S$.
    A minimum vertex separator $\MVS$ is the smallest weight vertex separator.
\end{definition}

See \pref{fig:first_shape_example} for example; in our figures the minimum vertex separator is colored green.
Note that by definition, $U_a \cap V_a$ must be in the minimum vertex separator.
Using the norm bounds from \cite[Corollary 8.16]{AMP20} and the same calculations from \cite[Appendix A]{GJJ20}), we have
\begin{proposition} \label{prop:graph_matrix_norm_bound}
    With probability over $1-o(1)$, for all shapes $a$ the graph matrix satisfies
    \begin{equation*}
        \|M_a\| \leq (|V(a)| \cdot |E(a)| \cdot \log n)^{O(|V(a)| + |E(a)|)} \cdot n^{\frac{w(V(a)) - w(\MVS) + w(\Wiso)}{2}} = \wt{O}\left( n^{\frac{w(V(a)) - w(\MVS) + w(\Wiso)}{2}} \right).
    \end{equation*}
\end{proposition}

\subsection{Proof overview of \pref{lem:M_E_PSD}}
\label{sec:SoS_lower_bound_proof_overview}

Since the proof is rather technical, we first provide an overview of the proof and defer the technical details to the Appendix. At a high-level, our strategy resembles that of~\cite{GJJ20} who proved a sum-of-squares lower bound for the problem of certifying the optimum value of the Sherrington-Kirkpatrick Hamiltonian. However, there are important differences to adapt this strategy to our setting as we describe in \pref{rem:difference_from_GJJ}.

At a high-level, our strategy works in two steps which are rather common in the analyses of moment matrices arising in several prior works on SoS lower bounds using pseudo-calibration. In the first step, we will prove that the moment matrix $\MomentMatrix$ is positive semidefinite and approximately (but not exactly) satisfies the polynomial constraints. In the second step, we will modify the pseudo-distribution $\mu$ so as to satisfy the constraints exactly and further show that this correction is small and does not affect the analysis of PSDness.

\paragraph{Decomposition of $\MomentMatrix$.}

Observe that the coefficients $\lambda_{\alpha,I,J}$ in \pref{eq:lambda_explicit} only depend on the shapes.
Thus, we can write $\MomentMatrix$ as
\begin{equation*}
    \MomentMatrix = \sum_{a: \text{ shape}} \lambda_a M_a.
\end{equation*}

We will first identify combinatorial conditions on the shapes defining the graphical matrices that appear with nonzero coefficients in the above expansion. 
The shapes with $\lambda_a \neq 0$ need to satisfy the following conditions,

\begin{definition}\label{def:valid_shapes}
    Let $\ValidShapes$ be the set of shapes $a$ such that
    \begin{enumerate}
        \item $U_a, V_a$ contain only square vertices and $|U_a|, |V_a| \leq 2$,
        \label{cond:Ua_Va}

        \item $\deg(\square{i}) + \1\{\square{i} \in U_a\} + \1\{\square{i} \in V_a\}$ is even for all $\square{i}\in V(a)$,
        \label{cond:square_degree}

        \item $\deg(\circle{\s})$ is even for all $\circle{\s}\in V(a)$,
        \label{cond:circle_degree}

        \item $|E(a)| \leq \tau$,
        \label{cond:num_edges}

        \item There are no isolated vertices in $W_a$.
        \label{cond:no_isolated_vertices}
    \end{enumerate}
\end{definition}
In words, \pref{cond:Ua_Va} is because $\MomentMatrix$ only contains moments of degree $\leq 4$; \pref{cond:square_degree} ensures that $\Delta_i + I_i + I_j$ is even; \pref{cond:circle_degree} ensures that $|\Alphas|$ is even; \pref{cond:num_edges} ensures that $|\alpha| \leq \tau$; and finally \pref{cond:no_isolated_vertices} is simply because shapes with isolated vertices don't appear in the decomposition (there can only be isolated vertices in $U_a\cap V_a$).

\begin{remark}
    For any $a\in \ValidShapes$, the conditions in \pref{def:valid_shapes} also imply that $|E(a)|$ is even and the total degree of square vertices is a multiple of 4.
\end{remark}

Thus, we can decompose $\MomentMatrix$ into shapes in $\ValidShapes$:
\begin{equation*}
    \MomentMatrix = \sum_{a\in \ValidShapes} \lambda_a M_a.
\end{equation*}

Next, observe that we can break $\MomentMatrix$ into blocks indexed by $(k,\ell)\in \{0,1,2\}^2$.
The $(k,\ell)$ block $\MomentMatrix_{k\ell}$ is $\binom{n}{k}\times \binom{n}{\ell}$ whose rows are indexed by subsets $\binom{[n]}{k}$ and columns are indexed by subsets $\binom{[n]}{\ell}$.
Clearly, shapes $a$ with $|U_a| = k$, $|V_a| = \ell$ contribute to $\MomentMatrix_{k\ell}$ only.
Moreover, $|U_a| + |V_a|$ must be even because the total degree of the square vertices must be even (each hyperedge contributes two).
Thus, the blocks $\MomentMatrix_{01}$, $\MomentMatrix_{10}$, $\MomentMatrix_{12}$, and $\MomentMatrix_{21}$ are zero, i.e.\ all odd moments are zero.
Thus, $\MomentMatrix$ has the following structure
\begin{equation*}
    \MomentMatrix =
    \begin{bmatrix}
        \MomentMatrix_{00} & 0 & \MomentMatrix_{02} \\
        0 & \MomentMatrix_{11} & 0 \\
        \MomentMatrix_{20} & 0 & \MomentMatrix_{22}
    \end{bmatrix}
    ,
\end{equation*}
where $\MomentMatrix_{00}$ is a scalar, and $\MomentMatrix_{02} = \MomentMatrix_{20}^\top$ is a vector and has the same entries as $\MomentMatrix_{11}$.

We need to show that $\MomentMatrix$, with some small modifications, is positive semidefinite and satisfies all constraints in the quadratic system.

\paragraph{Proving PSDness.} 
We parameterize $m = n^{2-\eps}$ for $\eps = \frac{C \log\log n}{\log n}$ for a sufficiently large constant $C$ in the analysis that follows. We know that $\MomentMatrix$ can be expanded as a sum of graphical matrices indexed by shapes $a$ in $\ValidShapes$ with coefficient $\lambda_a$. 
We first identify the shapes that contribute scaled identity matrices in the diagonal blocks.
We call these shapes the \textit{trivial shapes}; see \pref{fig:trivial_shapes} for examples.
\begin{definition}[Trivial shape] \label{def:trivial_shapes}
    A shape $a$ is trivial if $U_a = V_a$, $W_a = \varnothing$, and $E(a) = \varnothing$.
    Its associated matrix $M_a = \Id$.
\end{definition}

\FigTrivialShapes

In other words, the trivial shapes correspond to the Hermite indices $\alpha = \vec{0}$ and $|I| = |J|$.
For a trivial shape $\TrivialShape{k}$ with $|U_{\TrivialShape{k}}| = |V_{\TrivialShape{k}}| = k$, its matrix $\lambda_{\TrivialShape{k}} M_{\TrivialShape{k}} = n^{-k}\cdot \Id$ is a component in $\MomentMatrix_{kk}$.
Crucially, it is full rank and has minimum singular value $n^{-k}$, hence we can \textit{charge} other shapes that have small norm to the trivial shapes.
We call this procedure a \emph{charging scheme}.

\paragraph{Negligible shapes.}
We can charge several shapes to the trivial shapes if the contribution from those shapes are dominated by the scaled identity matrices from the trivial shapes;
we call all shapes that can be charged this way \emph{negligible}.

\begin{definition} \label{def:negligible_shapes}
    We say a shape is negligible if $|E(a)| \neq 0$ and
    \begin{equation*}
        \|\lambda_a M_a\| \leq n^{-\frac{|U_a|+|V_a|}{2}} \cdot n^{-\Omega(\eps|E(a)|)} \mper
    \end{equation*}
\end{definition}

Intuitively, for a negligible shape $a$ in block $\MomentMatrix_{kk}$ (meaning $|U_a| = |V_a| = k$), its contribution $\|\lambda_a M_a\| \ll n^{-k}$, which is the minimum singular value of $\lambda_{\TrivialShape{k}} M_{\TrivialShape{k}}$.

In \pref{sec:negligible_shapes}, we will identify a simple criterion to determine whether a shape is negligible or not (\pref{lem:shape_phi}), then we will prove that $M_{\TrivialShape{k}}$ dominates all negligible shapes, hence forming a PSD component in $\MomentMatrix$:

\begin{lemma} \label{lem:trivial_negl_psd}
    For $k = 0,1,2$, let $\NeglShapes{k}$ be the set of negligible shapes in block $\MomentMatrix_{kk}$, and let $\NeglError{k} \seteq \sum_{a\in \NeglShapes{k}} \lambda_a M_a$.
    There exist constants $c_1,c_2 > 0$ such that if the threshold $\tau \leq n^{c_1 \eps}$, then
    \begin{equation*}
        \|\NeglError{k}\| \leq n^{-k - c_2\eps} \mper
    \end{equation*}
    This implies that
    \begin{equation*}
        \lambda_{\TrivialShape{k}} M_{\TrivialShape{k}} + \NeglError{k} \succ 0 \mper
    \end{equation*}
\end{lemma}

Note that in the case $k=0$, we have $\MomentMatrix_{00} = 1 + o(1)$.
This is consistent with the calculations of low-degree hardness in \pref{sec:low_degree_hardness}.
Note also that $\MomentMatrix$ must have a non-trivial null space due to the constraints, hence there must be non-negligible shapes in $\ValidShapes$ which we deal with next.

The same analysis also shows the following norm bounds,
\begin{lemma} \label{lem:M_norm}
    There exists a constant $c_1 > 0$ such that if the threshold $\tau \leq n^{c_1\eps}$, then for any $k,\ell$, $\|\MomentMatrix_{k\ell}\| \leq n^{-\frac{k+\ell}{4}}$.
\end{lemma}

\paragraph{Connected shapes and spider.}
We look at the shapes in $\MomentMatrix_{11}$ and $\MomentMatrix_{22}$ that are connected, meaning there is path from $U_a$ to $V_a$ and $\MVS\neq \varnothing$.
We show in \pref{sec:connected_shapes} that there is only \textit{one} connected non-trivial shape that is not negligible, namely the \textit{spider}; see \pref{fig:spider} for illustration.

\FigSpider

\begin{lemma} \label{lem:connected_shapes_negligible}
    If $a \in \ValidShapes$ is a connected shape and not a trivial shape nor a spider, then $a$ is negligible.
\end{lemma}


Next, we handle the spider $a_{\spider}$.
The main insight is that $M_{\spider}$ is ``almost'' in the null space of $\MomentMatrix$,
i.e.\ $\MomentMatrix M_{\spider} \approx 0$.
Then, we use the following result (see also \cite[Fact 3.1]{GJJ20}); we give a short proof for completeness.
\begin{lemma}  \label{lem:kill_spider}
    Suppose a matrix $A$ satisfies $\MomentMatrix A = 0$, then $\MomentMatrix - A \succeq 0$ implies $\MomentMatrix \succeq 0$.
\end{lemma}
\begin{proof}
    For any vector $x$, let $y$ be its projection onto the column space of $\MomentMatrix$.
    We have $y \perp \Null(\MomentMatrix)$ and $y^\top A y = 0$.
    Then, $x^\top \MomentMatrix x = y^\top \MomentMatrix y = y^\top (\MomentMatrix - A) y$.
    Thus, $\MomentMatrix - A\succeq 0$ implies $x^\top \MomentMatrix x \geq 0$ for all $x$, which means $\MomentMatrix \succeq 0$.
\end{proof}

Intuitively, \pref{lem:kill_spider} allows us to add/remove any component of $\MomentMatrix$ which is in the null space of $\MomentMatrix$.
In our case, we can thus remove the component $\lambda_{\spider} M_{\spider}$ from $\MomentMatrix$ modulo some small error $\SpiderError$.
More specifically, in \pref{sec:spider}, we will show the following,
\begin{lemma} \label{lem:spider_expansion}
    Suppose $\MomentMatrix$ exactly satisfies all constraints $\{\gs(x) = 0\}_{\s\leq m}$. Then there exists a matrix $A$ such that $\MomentMatrix A = 0$ and
    \begin{equation*}
        \lambda_{\spider} M_{\spider}
        = A + \SpiderErrorBlock{00} + \SpiderErrorBlock{20} + \SpiderErrorBlock{20}^\top + \SpiderErrorBlock{22},
    \end{equation*}
    where $\SpiderErrorBlock{00}$, $\SpiderErrorBlock{20}$, $\SpiderErrorBlock{20}^\top$, $\SpiderErrorBlock{22}$ are errors in blocks $\MomentMatrix_{00}$, $\MomentMatrix_{20}$, $\MomentMatrix_{02}$, $\MomentMatrix_{22}$ respectively,
    and $|\SpiderErrorBlock{00}| = \wt{O}(n^{-3})$, $\|\SpiderErrorBlock{20}\| = \wt{O}(n^{-5/2})$, and $\|\SpiderErrorBlock{22}\| = \wt{O}(n^{-2-\eps})$.
\end{lemma}

Thus, we have
\begin{equation*} \label{eq:M_minus_A}
\begin{aligned}
    \MomentMatrix' \seteq \MomentMatrix - A &= \MomentMatrix - \lambda_{\spider} M_{\spider} + \SpiderError \\
    &= 
    \begin{bmatrix}
        \MomentMatrix_{00} + \SpiderErrorBlock{00} & 0 & \MomentMatrix_{20}^\top + \SpiderErrorBlock{20}^\top \\
        0 & \MomentMatrix_{11} & 0 \\
        \MomentMatrix_{20} + \SpiderErrorBlock{20} & 0 & \MomentMatrix_{22}' + \SpiderErrorBlock{22}
    \end{bmatrix},
\end{aligned}
\numberthis
\end{equation*}
where $\MomentMatrix_{22}'$ is the block $\MomentMatrix_{22}$ with the spider removed.

Then, by \pref{lem:kill_spider}, it suffices to prove that $\MomentMatrix- A \succeq 0$.
Next, we turn to the shapes in $\MomentMatrix_{20}, \MomentMatrix_{02}$ and the disconnected shapes in $\MomentMatrix_{22}$.

\paragraph{Disconnected shapes.}
Several disconnected shapes in $\ValidShapes$ (with $\MVS = \varnothing$) are not negligible.
We note that all shapes in $\MomentMatrix_{11}$ must be connected due to the conditions in \pref{def:valid_shapes}.
We will show that all disconnected shapes can be captured in a positive semidefinite component while introducing negligible errors.

We first introduce the following definition,

\begin{definition}[One-sided shape] \label{def:one_sided_shape}
    We say a shape is one-sided if either $U_a$ or $V_a$ is empty and there is no isolated component disconnected from $U_a$ or $V_a$.
    If $V_a = \varnothing$, we call it a left one-sided shape;
    if $U_a = \varnothing$, we call it a right one-sided shape.
\end{definition}

Note that the transpose of a left one-sided shape is a right one-sided shape, and further any disconnected shape in $\MomentMatrix_{22}$ contains a left and right one-sided shape.
The main observation is that for any disconnected shape $a = (a_1,a_2^\top)$, $M_a \approx M_{a_1} M_{a_2}^\top$.

\begin{lemma} \label{lem:collapsed_shapes_negligible}
    For a disconnected shape $a = (a_1,a_2^\top)$ where $a_1,a_2$ are left one-sided shapes,
    \begin{equation*}
        M_{a} = M_{a_1} M_{a_2}^\top + \CollapsedErrorShapes{a_1,a_2^\top}
    \end{equation*}
    where $\CollapsedErrorShapes{a_1,a_2^\top}$ consists of shapes obtained from \emph{collapsing} $a_1$ and $a_2^\top$.
    Moreover, all such collapsed shapes are negligible.
\end{lemma}

The collapsed shapes are a result of graph matrix multiplication; see details in \pref{sec:disconnected_shapes}.
We will show that all disconnected shapes can be captured in a PSD component.
The intuition is that since $\lambda_a = \lambda_{a_1} \lambda_{a_2}$, the term $\lambda_a M_a \approx (\lambda_{a_1} M_{a_1}) (\lambda_{a_2} M_{a_2})^\top$.

\begin{lemma} \label{lem:disconnected_shapes_captured}
    Consider the first column of $\MomentMatrix'$: $(\MomentMatrix_{00}, 0, \MomentMatrix_{20})$, and let $v \seteq (1, 0, \frac{\MomentMatrix_{20}}{\MomentMatrix_{00}})$.
    The matrix $\MomentMatrix_{00} \cdot vv^\top$ captures all disconnected shapes in $\MomentMatrix_{22}$ modulo some error consisting of negligible shapes.
\end{lemma}



At this point, we can conclude that all shapes in $\MomentMatrix$ are accounted for and thus $\MomentMatrix$ is positive semidefinite. However, it does not exactly satisfy the constraints.
We now proceed to prove that we can correct $\MomentMatrix$ with a small perturbation.

\paragraph{Fixing the \pd.}
We show that we can ``fix'' $\MomentMatrix$ such that $\FinalMomentMatrix = \MomentMatrix + \calE$ satisfies all constraints exactly and that $\calE$ is negligible. Suppose we view the \pe $\pE$ as a flattened vector, then there exists a matrix $Q$ such that $Q\pE = 0$ if and only if $\pE$ satisfies all constraints.
Here we assume that $\pE$ only includes even-degree monomials, since odd-degree monomials are zero already and don't need to be fixed.

To begin, in \pref{sec:truncation_error} we show that if the truncation threshold $\tau$ in \pref{def:candidate-pd} is chosen appropriately, then $\MomentMatrix$ already approximately satisfies the constraints, i.e.\ the norm of the ``error vector'' $\|Q\pE\|_2 \approx 0$.
\begin{lemma} \label{lem:error_upper_bound}
    There exist constants $C,C_1, c_2, c_3 > 0$ such that if $\eps \geq \frac{C\log\log n}{\log n}$ and $\frac{C_1}{\eps} \leq \tau \leq n^{c_2\eps}$, then $\|Q\pE\|_2 \leq n^{-c_3 \eps \tau}$.
\end{lemma}

Next, we fix $\pE$ by projecting it to the null space of $Q$,
\begin{equation*}
    \FinalpE \seteq \pE - Q^\top(QQ^\top)^\dag Q\pE,
\end{equation*}
where $(QQ^\top)^\dag$ is the pseudo-inverse of $QQ^\top$ since it is not invertible.
Clearly, $Q \FinalpE = 0$.

Finally, to bound the norm of the correction $\|Q^\top(QQ^\top)^\dag Q\pE\|_2$, it suffices to upper bound $\|Q\|$ and $\|(QQ^\top)^\dag\|$.
For $\|(QQ^\top)^\dag\|$, we need to \textit{lower bound} the smallest \textit{nonzero} singular value of $Q$.
We prove the following in \pref{sec:fixing_pd},
\begin{lemma} \label{lem:Q_norm_singular_value}
    There exists a constant $C$ such that for $\eps \geq \frac{C\log\log n}{\log n}$,
    $\|Q\| \leq \wt{O}(n)$ and the smallest nonzero eigenvalue of $QQ^\top$ is $\Omega(n^2)$.
\end{lemma}

\pref{lem:error_upper_bound} and \pref{lem:Q_norm_singular_value} immediately imply the following
\begin{lemma} \label{lem:correction_norm}
    There exist constants $C,C_1, c_2, c_3 > 0$ such that if $\eps \geq \frac{C\log\log n}{\log n}$ and $\frac{C_1}{\eps} \leq \tau \leq n^{c_2\eps}$, then there exists a matrix $\CorrectionMatrix$ that corrects the nonzero blocks of $\MomentMatrix$ such that $\MomentMatrix + \CorrectionMatrix$ satisfies all constraints $\{\gs(x) = 0\}_{\s\leq m}$ and that $\|\CorrectionMatrix\| \leq n^{-\Omega(\eps\tau)}$.
\end{lemma}
\begin{proof}
    $\|Q\| \cdot \|(QQ^\top)^\dag\| \leq \wt{O}(1/n)$ due to \pref{lem:Q_norm_singular_value}.
    Thus, the correction $\|Q^\top(QQ^\top)^\dag Q\pE\|_2 \leq \|Q\| \cdot \|(QQ^\top)^\dag\| \cdot \|Q\pE\|_2 \leq n^{-\Omega(\eps\tau)}$.
\end{proof}

\paragraph{Putting things together.}
We are ready to prove \pref{lem:M_E_PSD}. The proof is essentially a summary of the results in this overview.

\begin{proof}[Proof of \pref{lem:M_E_PSD}]
    The candidate moment matrix given by \pref{def:candidate-pd} can be written as a sum of graph matrices of shapes in $\ValidShapes$: $\MomentMatrix = \sum_{a\in \ValidShapes} \lambda_a M_a$, and has the following structure,
    \begin{equation*}
        \MomentMatrix =
        \begin{bmatrix}
            \MomentMatrix_{00} & 0 & \MomentMatrix_{02} \\
            0 & \MomentMatrix_{11} & 0 \\
            \MomentMatrix_{20} & 0 & \MomentMatrix_{22}
        \end{bmatrix} \mper
    \end{equation*}

    By \pref{lem:correction_norm}, we can correct the moment matrix so that $\FinalMomentMatrix \seteq \MomentMatrix + \CorrectionMatrix$ satisfies all constraints with error $\|\CorrectionMatrix\| \leq n^{-\Omega(\eps \tau)}$.

    Next, by \pref{lem:spider_expansion}, there exists a matrix $A$ such that $\FinalMomentMatrix A = 0$ and that the spider term $\lambda_{\spider} M_{\spider}$ equals $A$ plus some errors:
    \begin{equation*}
        \begin{aligned}
            \MomentMatrix' \seteq \FinalMomentMatrix - A = 
            \begin{bmatrix}
                \MomentMatrix_{00} + \SpiderErrorBlock{00} & 0 & \MomentMatrix_{20}^\top + \SpiderErrorBlock{20}^\top \\
                0 & \MomentMatrix_{11} + \SpiderErrorBlock{11} & 0 \\
                \MomentMatrix_{20} + \SpiderErrorBlock{20} & 0 & \MomentMatrix_{22}' + \SpiderErrorBlock{22}
            \end{bmatrix} \mcom
        \end{aligned}
    \end{equation*}
    where $\MomentMatrix_{22}'$ is the block $\MomentMatrix_{22}$ with the spider removed, and the errors $|\SpiderErrorBlock{00}| = \wt{O}(n^{-3})$, $\|\SpiderErrorBlock{20}\| = \wt{O}(n^{-5/2})$, $\|\SpiderErrorBlock{22}\| = \wt{O}(n^{-2-\eps})$, and $\|\SpiderErrorBlock{11}\| = n^{-\Omega(\eps\tau)}$.
    Now, due to \pref{lem:kill_spider} it suffices to prove that $\MomentMatrix'$ is positive semidefinite.

    Next, let $u \seteq (1, 0, \frac{\MomentMatrix_{20} + \SpiderErrorBlock{20}}{\MomentMatrix_{00} + \SpiderErrorBlock{00}})$, the first column of $\MomentMatrix'$ divided by the first entry, and consider the matrix $(\MomentMatrix_{00} + \SpiderErrorBlock{00}) \cdot uu^\top$ (note that $\MomentMatrix_{00} + \SpiderErrorBlock{00} = 1+o(1)$).
    This matrix is approximately the matrix $\MomentMatrix_{00} \cdot vv^\top$ in \pref{lem:disconnected_shapes_captured} that captures the disconnected shapes in $\MomentMatrix_{22}$:
    \begin{equation*}
        \MomentMatrix' = (\MomentMatrix_{00} + \SpiderErrorBlock{00}) \cdot uu^\top + 
        \begin{bmatrix}
            0 & 0 & 0 \\
            0 & \MomentMatrix_{11} + \SpiderErrorBlock{11} & 0 \\
            0 & 0 & \MomentMatrix_{22}'' + \SpiderErrorBlock{22}''
        \end{bmatrix} \mcom
    \end{equation*}
    where $\MomentMatrix_{22}''$ is $\MomentMatrix_{22}'$ with the disconnected shapes removed, and $\SpiderErrorBlock{22}''$ contains error terms including negligible shapes, $\MomentMatrix_{20} \SpiderErrorBlock{20}^\top$, and $\SpiderErrorBlock{00} \MomentMatrix_{20} \MomentMatrix_{20}^\top$.
    By \pref{lem:M_norm}, the latter two have norms $\wt{O}(n^{-3}) \ll n^2$.

    Finally, both $\MomentMatrix_{11} + \SpiderErrorBlock{11}$ and $\MomentMatrix_{22}'' + \SpiderErrorBlock{22}''$ now only contain the trivial shapes and negligible shapes, hence by \pref{lem:trivial_negl_psd} they are PSD.
    This proves that $\MomentMatrix' \succeq  0$, which completes the proof.
\end{proof}

\begin{remark}[Comparison to the proof strategy of~\cite{GJJ20}] \label{rem:difference_from_GJJ}
    Our proof is conceptually similar and builds heavily on the analysis in~\cite{GJJ20} with some key differences.
    In~\cite{GJJ20}, the goal is to work with a special form of ``rank 1'' polynomial constraints $\{\brac{x,g_i}^2 = 1\}_{i\leq m}$ where the $g_i$s are random vectors (the ``affine planes'' problem). As a result, the construction of pseudo-distribution leads to a moment matrix with a different set of shapes playing a prominent role -- 2-uniform graphs as opposed to 3-uniform hypergraphs in our case. As a result, several components in the proof (including the spectral norm bounds, the characterization of negligible shapes and spiders) are different. 

    Our analysis also requires dealing with certain disconnected shapes a bit differently by ``charging'' them to an appropriate extra PSD component. This actually leads to an important quantitative difference: in the result of~\cite{GJJ20}, the sum-of-squares lower bound (that works for $n^{O(1)}$-degree as against just degree $4$ in our work) succeeds only for $m \leq n^{3/2-\eps}$. This is despite the fact that low-degree hardness even for the rank-1 random polynomial above suggests a threshold of $m \leq n^{2-\eps}$. In contrast, our analysis provides a nearly optimal lower bound at degree $4$ that matches the prediction of low-degree hardness.
\end{remark}

\section*{Acknowledgments}
We thank anonymous reviewers for their comments and suggestions. We also thank Alperen Ergür, Amit Sahai and Aayush Jain for illuminating discussions and pointing us to relevant related work.
Finally, we would like to thank Sidhanth Mohanty and Jeff Xu for discussions on low-degree hardness and SoS lower bounds in general.

\bibliographystyle{alpha}
\bibliography{main}
\newpage

\appendix
\section{Omitted Proofs}

\subsection{Contributions from negligible shapes}
\label{sec:negligible_shapes}

In this section, we look at what shapes are negligible, and we show that the total contribution of the negligible shapes is dominated by the trivial shapes.

By the graph matrix norm bounds (\pref{prop:graph_matrix_norm_bound}) and \pref{eq:lambda_bound},
\begin{equation*} \label{eq:lambda_M_contribution}
    |\lambda_a| \cdot \|M_a\| \leq n^{-|E(a)| - \frac{|U_a|+|V_a|}{2}}\cdot n^{\frac{w(V(a)) - w(\MVS) + w(\Wiso)}{2}} \cdot \left(|E(a)|\cdot \log n\right)^{O(|E(a)|)} \mcom
    \numberthis
\end{equation*}
since $|V(a)| \leq O(|E(a)|)$ for $a\in \ValidShapes$.

In the following lemma, we identify a quantity that indicates if a shape is negligible or not.

\begin{lemma} \label{lem:negligible_shape_large_phi}
    For a shape $a\in \ValidShapes$, define
    \begin{equation*}  \label{eq:shape_norm_exponent}
        \varphi(a) \seteq |E(a)| - \frac{w(V(a)) -w(\MVS) + w(\Wiso)}{2}.
        \numberthis
    \end{equation*}
    Then, there exist constants $C, c_1>0$ such that for $\eps = \frac{C\log\log n}{\log n}$ and $\tau = n^{c_1\eps}$,
    if a shape $a\in \ValidShapes$ satisfies $|E(a)| \neq 0$ and
    \begin{equation*}
        \varphi(a) \geq \frac{\eps}{8} |E(a)| \mcom
    \end{equation*}
    then it is negligible.
\end{lemma}

\begin{proof}
    First, $\Wiso = \varnothing$ due to \pref{cond:no_isolated_vertices} of $\ValidShapes$ (in \pref{def:valid_shapes}).
    If we choose $C$ to be sufficiently large and $c_1$ sufficiently small, then since $|E(a)| \leq \tau = n^{c_1\eps}$, by \pref{eq:lambda_M_contribution} we have
    \begin{equation*}
        \|\lambda_a M_a\| \leq n^{- \frac{|U_a|+|V_a|}{2} - \varphi(a)} n^{O(\eps|E(a)|)}
        \leq n^{-\frac{|U_a|+|V_a|}{2}} \cdot n^{-c_2 \eps |E(a)|}
    \end{equation*}
    for some constant $c_2$.
    This satisfies the definition of negligible shapes (\pref{def:negligible_shapes}), thus $a$ is negligible.
\end{proof}

\pref{fig:negligible_shapes} includes some negligible shapes for illustration.

\FigNegligibleShapes

Next, we analyze the quantity $\varphi(a)$.
To do so, we first introduce some more notations.
For a shape $a$, let $S_a, C_a$ be the set of square and circle vertices respectively.
Let $\wt{U}_a \seteq U_a \setminus (U_a\cap V_a)$ and $\wt{V}_a \seteq V_a \setminus (U_a\cap V_a)$, and let $\MVSnointer \seteq \MVS \setminus (U_a \cap V_a)$.
The conditions of $\ValidShapes$ in \pref{def:valid_shapes} ensure that
\begin{enumerate}
    \item The degree of any square vertex in $\wt{U}_a \cup \wt{V}_a$ must be odd; minimum degree $1$,
    \item The degree of any square vertex in $U_a \cap V_a$ must be even; minimum degree $0$,
    \item The degree of any other vertex must be even; minimum degree $2$.
\end{enumerate}

\begin{definition}[Large-degree vertex] \label{def:large_degree_vertex}
    We say that a vertex has \textit{large degree} if its degree is larger than the minimum degree (and must be larger by at least 2 due to the conditions in \pref{def:valid_shapes}).
\end{definition}

The following lemma lets us determine whether a shape in $\ValidShapes$ is negligible or not,

\begin{lemma} \label{lem:shape_phi}
    For a shape $a\in \ValidShapes$, let $\delta_s,\delta_c$ be the number of large-degree square and circle vertices,
    \begin{equation*}
        \varphi(a) \geq -\frac{1}{4}(|\wt{U}_a| + |\wt{V}_a|) + \frac{\eps}{4}|E(a)| + \frac{1}{2}w(\MVSnointer) + \frac{1}{2}\delta_s + \left(1-\frac{\eps}{2}\right) \delta_c \mper
    \end{equation*}
\end{lemma}
\begin{proof}
    First, by \pref{cond:no_isolated_vertices} of $\ValidShapes$, $\Wiso = \varnothing$.
    Moreover, $U_a\cap V_a$ must be in the minimum vertex separator, so their contributions in \pref{eq:shape_norm_exponent} cancel out.
    Then, using $w(\square{}) = 1$ and $w(\circle{})= 2-\eps$, we can rewrite $\varphi(a)$ as
    \begin{equation*}
        \varphi(a) = |E(a)| - \frac{1}{2}|\wt{S}_a| - \left(1-\frac{\eps}{2}\right) |C_a| + \frac{1}{2} w(\MVSnointer) \mcom
    \end{equation*}
    where $\wt{S}_a$ is the set of square vertices excluding $U_a \cap V_a$.

    By \pref{cond:square_degree} of $\ValidShapes$, each square vertex in $\wt{U}_a,\wt{V}_a$ must have degree at least $1$, and each square vertex in $\wt{S}_a \setminus (\wt{U}_a \cup \wt{V}_a)$ must have degree at least $2$ (they are not isolated).
    Each large degree vertex introduces at least two more.
    Thus, the total number of hyperedges
    \begin{equation*}
        |E(a)| \geq \frac{1}{2} \left( |\wt{U}_a| + |\wt{V}_a| + 2|\wt{S}_a \setminus (\wt{U}_a \cup \wt{V}_a)| + 2\delta_s \right) = |\wt{S}_a| - \frac{1}{2}\left(|\wt{U}_a|+|\wt{V}_a|\right) + \delta_s \mper
    \end{equation*}
    Moreover, each circle vertex must have degree at least 2. Thus,
    \begin{equation*}
        |E(a)| \geq 2|C_a| + 2\delta_c.
    \end{equation*}
    Combining the two, we get
    \begin{equation*}
    \begin{aligned}
        \varphi(a) &= |E(a)| - \left(\frac{1}{2}|E(a)| + \frac{1}{4}(|\wt{U}_a| + |\wt{V}_a|) -\frac{1}{2}\delta_s \right) - \left( 1-\frac{\eps}{2} \right)\left(\frac{1}{2}|E(a)| - \delta_c\right) + \frac{1}{2} w(\MVSnointer) \\
        &\geq -\frac{1}{4}(|\wt{U}_a| + |\wt{V}_a|) + \frac{\eps}{4}|E(a)| + \frac{1}{2}w(\MVSnointer) + \frac{1}{2}\delta_s + \left(1-\frac{\eps}{2}\right) \delta_c \mper
    \end{aligned}
    \end{equation*}
    This completes the proof.
\end{proof}

We now prove \pref{lem:trivial_negl_psd}.
We first derive a bound on the number of negligible shapes.

\begin{lemma}\label{lem:number_of_shapes}
    Let $\ell \geq 2$.
    The number of shapes with exactly $\ell$ edges and no isolated vertices is at most $\ell^{O(\ell)}$.
\end{lemma}
\begin{proof}
    Since each hyperedge connects to 3 vertices, there can be at most $3\ell$ vertices.
    If there are $v$ vertices, then there are $v^{3\ell}$ ways to assign edges.
    Thus, in total, the number of shapes is at most $\sum_{v} v^{3\ell} \leq (3\ell) \cdot (3\ell)^{3\ell} \leq \ell^{O(\ell)}$.
\end{proof}

\begin{lemma}[Restatement of \pref{lem:trivial_negl_psd}]
    For $k \in \{0,1,2\}$, let $\NeglShapes{k}$ be the set of negligible shapes in block $\MomentMatrix_{kk}$, and let $\NeglError{k} \seteq \sum_{a\in \NeglShapes{k}} \lambda_a M_a$.
    There exist constants $c_1,c_2 > 0$ such that if the threshold $\tau \leq n^{c_1 \eps}$, then
    \begin{equation*}
        \|\NeglError{k}\| \leq n^{-k - c_2\eps} \mper
    \end{equation*}
    This implies that
    \begin{equation*}
        \lambda_{\TrivialShape{k}} M_{\TrivialShape{k}} + \NeglError{k} \succ 0 \mper
    \end{equation*}
\end{lemma}

\begin{proof}
    The shapes in $\NeglShapes{k}$ can have number of hyperedges ranging from 2 to $\tau$.
    By \pref{lem:number_of_shapes}, \pref{def:negligible_shapes} of the negligible shapes, and the triangle inequality,
    \begin{equation*}
        \|\NeglError{k}\| \leq \sum_{a \in \NeglShapes{k}} |\lambda_a| \cdot \|M_a\|
        \leq \sum_{\ell=2}^\tau \ell^{O(\ell)} \cdot n^{-k - \Omega(\eps\ell)}
        \leq n^{-k} \sum_{\ell=2}^\tau n^{-c_2' \eps \ell} \mcom
    \end{equation*}
    for some constant $c_2' > 0$ provided that $c_1$ is small enough.
    The summation is upper bounded by $n^{-c_2 \eps}$ for some $c_2$.
    Thus, $\|\NeglError{k}\| \leq n^{-k-c_2\eps} \leq o(n^{-k})$, much smaller than the minimum singular value of $\lambda_{\TrivialShape{k}} M_{\TrivialShape{k}}$.
    This completes the proof.
\end{proof}

The same analysis shows that the norm of block $\MomentMatrix_{k\ell}$ is dominated by the shape with the largest norm.

\begin{lemma}[Restatement of \pref{lem:M_norm}]
    There exists a constant $c_1 > 0$ such that if the threshold $\tau \leq n^{c_1\eps}$, then for any $k,\ell$, $\|\MomentMatrix_{k\ell}\| \leq n^{-\frac{k+\ell}{4}}$.
\end{lemma}
\begin{proof}
    By \pref{lem:shape_phi}, $\varphi(a) \geq -\frac{|U_a|+|V_a|}{4} - \frac{\eps}{4}|E(a)|$.
    The same calculations show that
    \begin{equation*}
        \Norm{\MomentMatrix_{k\ell}} \leq  \sum_{\stackrel{a\in \ValidShapes:}{|U_a|=k, |V_a|=\ell}} \|\lambda_a M_a\|
        \leq n^{-\frac{k+\ell}{4}} \mper
    \end{equation*}
\end{proof}

\subsection{Non-trivial non-spider connected shapes are negligible}
\label{sec:connected_shapes}

We say that a shape is connected if there is path from $U_a$ to $V_a$.
For connected shapes, we will show that except for \textit{one} shape, all connected shapes can be charged to the trivial shapes.
We call that shape a \textit{spider}; see \pref{fig:spider}.

We first show the following result about the structure of connected shapes.

\begin{lemma} \label{lem:connected_shape_mvs_degree}
    For connected shapes $\alpha\in \ValidShapes$, suppose $|U_a| = |V_a| = 2$, $U_a\cap V_a = \varnothing$, and the minimum vertex separator contains only one square vertex, then that vertex must have degree at least $4$.
\end{lemma}
\begin{proof}
    Suppose for contradiction $\MVS$ contains one square vertex $\square{i}$ and it has degree 2 (it cannot be isolated).
    Then, consider the left and right sides of the graph separated by this vertex.
    Since $\square{i}$ must contribute exactly one degree to each side, the total degree of square vertices on each side must be odd (each hyperedge contributes two degrees).
    However, by the conditions in \pref{def:valid_shapes}, the degree of any $\square{j}\notin U_a \cup V_a$ must be even, whereas the degree of $\square{j} \in U_a \cup V_a$ must be odd, thus due to $|U_a| = |V_a| = 2$ the total degree must be even.
    This is a contradiction.
    Thus, $\deg(\square{i})$ must be larger than 2, which means it must be at least 4.
\end{proof}

Next, we show that all non-spider non-trivial connected shapes are negligible.
Recall that a shape is negligible if $\varphi(a) \geq \frac{\eps}{8} |E(a)|$, and this can be determined by \pref{lem:shape_phi}.

\begin{lemma}[Restatement of \pref{lem:connected_shapes_negligible}]
    If $a \in \ValidShapes$ is a connected shape and not a trivial shape nor a spider, then $a$ is negligible.
\end{lemma}

\begin{proof}
    First, observe that for connected shapes, we must have $|U_a| = |V_a| \leq 2$.
    We split into several cases and apply \pref{lem:shape_phi},
    \begin{enumerate}
        \item $U_a = V_a$. In this case, $\wt{U}_a, \wt{V}_a = \varnothing$, and since $a$ is not a trivial shape, $|E(a)| > 0$. Thus, $\varphi(a) \geq \frac{\eps}{4}|E(a)|$.

        \item $|U_a| = |V_a| = 1$ and $U_a \cap V_a = \varnothing$.
        Such shapes must be connected, so $w(\MVS)$ is at least $1$, which cancels out with $-\frac{1}{4}(|\wt{U}_a| + |\wt{V}_a|) = -\frac{1}{2}$. Thus, $\varphi(a) \geq \frac{\eps}{4}|E(a)|$.

        \item $|U_a| = |V_a| = 2$ and $|U_a \cap V_a| = 1$.
        Since $|\wt{U}_a| = |\wt{V}_a| = 1$ and $\wt{U}_a,\wt{V}_a$ must be connected, this is the same as the previous case.

        \item $|U_a| = |V_a| = 2$, $U_a \cap V_a = \varnothing$, and $w(\MVS) = 1$ or $w(\MVS) \geq 2$.
        First, if $\MVS$ contains just one square vertex, then by \pref{lem:connected_shape_mvs_degree} it must have large degree, hence $\frac{1}{2} w(\MVS) + \frac{1}{2}\delta_s \geq 1$, canceling out the term $-\frac{1}{4}(|U_a|+|V_a|) = -1$.
        On the other hand, if $w(\MVS) \geq 2$, then clearly it already cancels out the $-1$.

        \item $|U_a| = |V_a| = 2$, $U_a \cap V_a = \varnothing$, and $w(\MVS) = 2-\eps$.
        In this case, $\MVS$ contains exactly one circle vertex, and $\varphi(a) \geq -1 + \frac{\eps}{4}|E(a)| + (1-\eps/2) = \eps (\frac{1}{4}|E(a)|-\frac{1}{2})$.
        Now, if $a$ is a spider, then $|E(a)| = 2$ and $\varphi(a) = 0$.
        Fortunately, for other non-spider shapes, $|E(a)| \geq 4$, which means $\varphi(a) \geq \frac{\eps}{8}|E(a)|$.
    \end{enumerate}
    Thus, except the spider, all non-trivial connected shapes have $\varphi(a) \geq \frac{\eps}{8} |E(a)|$.
\end{proof}

\subsection{The spider is close to the null space of \texorpdfstring{$\MomentMatrix$}{M}}
\label{sec:spider}

As described in the proof overview, we first identify the null space of the moment matrix $\MomentMatrix$ that satisfies all constraints exactly: for any $I\subseteq [n]$ ($|I|\leq 2$) and $\s\in[m]$,
\begin{equation*}
    0 = \pE[x^I (x^\top \Gs x)]
    = \sum_{i\neq j} \Gsij \pE[x^I x_i x_j] + \frac{1}{n}\sum_{i} \GEntries{\s}{ii} \pE[x^I] \mcom
\end{equation*}
where we use $x_i^2 = \frac{1}{n}$.
We can represent the above using a matrix $L_2 = M_{a_1} + \frac{1}{n}M_{a_2}$ (drawn as graph matrices in \pref{fig:L2_shapes}), with rows indexed by $\s\in[m]$ and columns indexed by $I \subseteq [n]$.
It is easy to see that the $(\s,I)$ entry of $L_2 \MomentMatrix$ is exactly $\pE[x^I (x^\top \Gs x)]$.
One can view $L_2$ as a ``check matrix'', i.e.\ if $\MomentMatrix$ exactly satisfies the constraints, then $L_2 \MomentMatrix = 0$.

\FigDegTwoL

Next, we prove that $M_{\spider}$ is close to the null space of $\MomentMatrix$.

\begin{lemma}[Restatement of \pref{lem:spider_expansion}]
    Suppose $\MomentMatrix$ exactly satisfies all constraints $\{\gs(x) = 0\}_{\s\leq m}$. Then there exists a matrix $A$ such that $\MomentMatrix A = 0$ and
    \begin{equation*}
        \lambda_{\spider} M_{\spider}
        = A + \SpiderErrorBlock{00} + \SpiderErrorBlock{20} + \SpiderErrorBlock{20}^\top + \SpiderErrorBlock{22},
    \end{equation*}
    where $\SpiderErrorBlock{00}$, $\SpiderErrorBlock{20}$, $\SpiderErrorBlock{20}^\top$, $\SpiderErrorBlock{22}$ are errors in blocks $\MomentMatrix_{00}$, $\MomentMatrix_{20}$, $\MomentMatrix_{02}$, $\MomentMatrix_{22}$ respectively,
    and $|\SpiderErrorBlock{00}| = \wt{O}(n^{-3})$, $\|\SpiderErrorBlock{20}\| = \wt{O}(n^{-5/2})$, and $\|\SpiderErrorBlock{22}\| = \wt{O}(n^{-2-\eps})$.
\end{lemma}

\begin{proof}
    Consider the matrix $L_2$ and shapes $a_1,a_2$ in \pref{fig:L2_shapes}.
    Using the graph matrix norm bounds (\pref{prop:graph_matrix_norm_bound}), $\|M_{a_1}\| = \wt{O}(n)$ and $\|\frac{1}{n}M_{a_2}\| = \wt{O}(n^{\frac{1}{2} - \frac{\eps}{2}})$.
    Now, we consider $L_2^\top L_2$:
    \begin{equation*}
        L_2^\top  L_2 = M_{a_1}^\top M_{a_1} + \frac{1}{n}(M_{a_2}^\top M_{a_1}+ M_{a_1}^\top M_{a_2}) + \frac{1}{n^2} M_{a_2}^\top M_{a_2} = M_{a_1}^\top M_{a_1} + E_{20}' + E_{20}'^\top + E_{00}'.
    \end{equation*}
    where $\|E_{20}'\| \leq \wt{O}(n^{\frac{3}{2}-\frac{\eps}{2}})$ and $\|E_{00}'\| \leq \wt{O}(n^{1-\eps})$.
    For the first term $M_{a_1}^\top M_{a_1}$, for $i_1 \neq i_2$ and $j_1 \neq j_2$,
    \begin{equation*}
        M_{a_1}^\top M_{a_1}( \{i_1,i_2\}, \{j_1,j_2\} ) = \sum_{\s\in[m]} \GEntries{\s}{i_1 i_2} \GEntries{\s}{j_1 j_2}.
    \end{equation*}
    Represented using graph matrix multiplication, $M_{a_1}^\top M_{a_1}$ is a sum of shapes in \pref{fig:LTL}.
    Note that the last two graphs come from the term $M_{a_1}^\top M_{a_1}( \{i_1,i_2\}, \{i_1,i_2\} ) = \sum_{\s\in[m]} (\GEntries{\s}{i_1 i_2})^2$ and using the fact $h_1(z)^2 = (z^2-1)+ 1 = h_2(z) + h_0(z)$.

    \FigLTL

    Observe that the first shape in \pref{fig:LTL} is the spider, which is the dominating term in the expansion:
    $\|M_{\spider}\| = \wt{O}(n^2)$, whereas the rest of the shapes have norms $\wt{O}(n^{2- \eps})$.
    Since $\lambda_{\spider} = n^{-4}$, we can rewrite the spider term
    \begin{equation*}
        \lambda_{\spider} M_{\spider} = n^{-4} (M_{a_1}^\top M_{a_1} + E_{22}')
        = n^{-4} L_2^\top L_2 + \SpiderErrorBlock{00} + \SpiderErrorBlock{20} + \SpiderErrorBlock{20}^\top + \SpiderErrorBlock{22},
    \end{equation*}
    where $\|\SpiderErrorBlock{00}\| = \wt{O}(n^{-3})$, $\|\SpiderErrorBlock{20}\| = \wt{O}(n^{-5/2})$, and $\|\SpiderErrorBlock{22}\| = \wt{O}(n^{-2-\eps})$.
\end{proof}

Thus, since $\MomentMatrix \cdot L_2^\top L_2 = 0$, by \pref{lem:kill_spider} it suffices to show that $\MomentMatrix - n^{-4} L_2^\top L_2 \succeq 0$.
This essentially kills the spider term $\lambda_{\spider} M_{\spider}$ in $\MomentMatrix$.

\subsection{Disconnected shapes are captured in a PSD component}
\label{sec:disconnected_shapes}

Disconnected shapes are the shapes where $U_a,V_a$ are disconnected, i.e.\ there is no minimum vertex separator.
Many of these shapes are not negligible compared to the trivial shapes.
However, we will prove that such shapes can be captured in a PSD component of $\MomentMatrix$.



Given a disconnected shape $a = (a_1,a_2^\top)$ where $a_1, a_2$ are left one-sided shapes, we first analyze the multiplication $M_{a_1}M_{a_2}^\top$ (recall that $M_{a_1}, M_{a_2}$ are both vectors).
Consider the example in \pref{fig:disconnected_shapes}.
For $i_1\neq i_2 \neq j_1\neq j_2$, the entry of $M_{a_1}M_{a_2}^\top$ is
\begin{equation*}
\begin{aligned}
    (M_{a_1}M_{a_2}^\top)(\{i_1,i_2\}, \{j_1,j_2\}) &=
    \left(\sum_{\stackrel{k\neq i_1,i_2}{\s_1\in [m]}} \GEntries{\s_1}{i_1 k} \GEntries{\s_1}{k i_2} \right)
    \cdot
    \left(\sum_{\stackrel{\ell\neq j_1,j_2}{\s_2\in [m]}} \GEntries{\s_2}{j_1 j_2} \GEntries{\s_2}{\ell\ell} \right) \\
    &= \sum_{\stackrel{k\neq \ell \notin \{i_1, i_2,j_1,j_2\}}{\s_1\neq \s_2}} \GEntries{\s_1}{i_1 k} \GEntries{\s_1}{k i_2} \GEntries{\s_2}{j_1 j_2} \GEntries{\s_2}{\ell\ell}
    + \sum_{\stackrel{k\neq \ell \notin \{i_1, i_2,j_1,j_2\}}{\s_1 = \s_2}} \GEntries{\s_1}{i_1 k} \GEntries{\s_1}{k i_2} \GEntries{\s_1}{j_1 j_2} \GEntries{\s_1}{\ell\ell} \\
    &\quad + \sum_{\stackrel{k= \ell \notin \{i_1, i_2,j_1,j_2\}}{\s_1\neq \s_2}} \GEntries{\s_1}{i_1 k} \GEntries{\s_1}{k i_2} \GEntries{\s_2}{j_1 j_2} \GEntries{\s_2}{kk}
    + \sum_{\stackrel{k= \ell \notin \{i_1, i_2,j_1,j_2\}}{\s_1 = \s_2}} \GEntries{\s_1}{i_1 k} \GEntries{\s_1}{k i_2} \GEntries{\s_1}{j_1 j_2} \GEntries{\s_1}{kk} + \cdots \mper
\end{aligned}
\end{equation*}
This expansion introduces several graph matrices, the first is the disconnected shape $a$.
The first two terms are drawn out in \pref{fig:disconnected_shapes}.
For other entries such as $i\seteq i_1 = j_1$ and $i \neq i_2 \neq j_2$,
\begin{equation*}
\begin{aligned}
    (M_{a_1}M_{a_2}^\top)(\{i,i_2\}, \{i,j_2\})
    &= \sum_{\stackrel{k\neq \ell \notin \{i, i_2,i,j_2\}}{\s_1\neq \s_2}} \GEntries{\s_1}{i k} \GEntries{\s_1}{k i_2} \GEntries{\s_2}{i j_2} \GEntries{\s_2}{\ell\ell}
    + \sum_{\stackrel{k\neq \ell \notin \{i, i_2,i,j_2\}}{\s_1 = \s_2}} \GEntries{\s_1}{i k} \GEntries{\s_1}{k i_2} \GEntries{\s_1}{i j_2} \GEntries{\s_1}{\ell\ell} \\
    &\quad + \sum_{\stackrel{k= \ell \notin \{i, i_2,i,j_2\}}{\s_1\neq \s_2}} \GEntries{\s_1}{i k} \GEntries{\s_1}{k i_2} \GEntries{\s_2}{i j_2} \GEntries{\s_2}{kk}
    + \sum_{\stackrel{k= \ell \notin \{i, i_2,i,j_2\}}{\s_1 = \s_2}} \GEntries{\s_1}{i k} \GEntries{\s_1}{k i_2} \GEntries{\s_1}{i j_2} \GEntries{\s_1}{kk} + \cdots \mper
\end{aligned}
\end{equation*}
These terms correspond to the collapsed shapes: each collapsed shape is obtained by iteratively merging one vertex from the left one-sided shape with one vertex (of the same type) from the right one-sided shape.
The first and second terms above are drawn out in \pref{fig:disconnected_shapes}.

\FigDisconnectedShapes

Now, we proceed to prove \pref{lem:collapsed_shapes_negligible}.

\begin{lemma}[Restatement of \pref{lem:collapsed_shapes_negligible}]
    For a disconnected shape $a = (a_1,a_2^\top)$ where $a_1,a_2$ are left one-sided shapes,
    \begin{equation*}
        M_{a} = M_{a_1} M_{a_2}^\top + \CollapsedErrorShapes{a_1,a_2^\top}
    \end{equation*}
    where $\CollapsedErrorShapes{a_1,a_2^\top}$ consists of shapes obtained from \emph{collapsing} $a_1$ and $a_2^\top$.
    Moreover, all such collapsed shapes are negligible.
\end{lemma}

\begin{proof}
    By the discussions above, expanding the matrix $M_{a_1} M_{a_2}^\top$ results in a summation of several shapes, one of which is the disconnected shape $a = (a_1, a_2^\top)$.
    Now, it suffices to show that all collapsed shapes are negligible.

    For a collapsed shape $a$, since it is connected, $\MVS$ must contain at least one vertex.
    Moreover, the merged vertices must have large degree (recall \pref{def:large_degree_vertex}), hence $\delta_s$ or $\delta_c$ must be at least 1.
    Since $|U_a|+|V_a| = 4$, by \pref{lem:shape_phi} we must have $\varphi(a) \geq \frac{\eps}{4}|E(a)|$, meaning that $a$ is negligible.

    Some care is required if the collapsed shape has parallel edges, which may happen if the endpoints of two different hyperedges collapse.
    We handle this by breaking the parallel edge into a sum of graph matrices.
    For example, suppose a shape collapsed from $a_1$ and $a_2^\top$ has two parallel edges $e$ labeled 1.
    Then by $h_1(z)^2 = h_2(z) + 1$, we get a sum of two shapes $b_1,b_2$ (with the same coefficient $\lambda_{a_1} \lambda_{a_2}$), where $b_1$ has the same edge $e$ labeled 2, and $b_2$ has no edge (and may have an isolated vertex).
    Clearly, $b_1$ is negligible, but to show that $b_2$ is also negligible requires some work.

    Let $a = (a_1, a_2^\top)$ be the disconnected shape, and consider a shape $b$ with isolated vertices collapsed from $a_1$ and $a_2^\top$.
    This collapsed shape introduces an error $\lambda_a M_b$, and we must show that $\|\lambda_a M_b\| \leq O(n^{-2 - \Omega(\eps |E(b)|)})$.
    Let $\DeletedEdges$ be the set of deleted edges, and note that $|\lambda_a| = |\lambda_b| \cdot O(n^{-|\DeletedEdges|})$ and $\|\lambda_b M_b\| = \wt{O}(n^{-2 - \varphi(b)})$ (recall the definition of $\varphi$ in \pref{eq:shape_norm_exponent}).
    Thus, it suffices to show that
    \begin{equation*}
        \varphi(b) + |\DeletedEdges| \geq \Omega(\eps |E(b)|).
    \end{equation*}
    Let $\Siso, \Ciso$ be the set of isolated square and circle vertices respectively, and let $b' = b \setminus (\Siso \cup \Ciso)$, the shape without the isolated vertices.
    Clearly, $\varphi(b) = \varphi(b') - w(\Wiso) = \varphi(b') - |\Siso| - (2-\eps) |\Ciso|$, and further by \pref{lem:shape_phi}, $\varphi(b') \geq -1 + \frac{\eps}{4}|E(b)| + \frac{1}{2}w(\MVS)$.
    Next, observe that the vertices in $\Wiso$ must have degree $\geq 4$ before the edges were removed.
    We consider two cases:
    \begin{enumerate}
        \item The isolated vertices were originally connected to circle vertices only:
        In this case, we have $w(\MVS) \geq 2-\eps$.
        Further, $|\DeletedEdges| \geq \frac{1}{2}\cdot 4 |\Siso|$ and $|\DeletedEdges| \geq 4|\Ciso|+2$.
        Thus,
        \begin{equation*}
            \varphi(b) + |\DeletedEdges| \geq -1 + \frac{\eps}{4}|E(b)| + \frac{1}{2}(2-\eps) - \frac{1}{2} |\DeletedEdges| - (2-\eps) \frac{|\DeletedEdges| - 2}{4}+ |\DeletedEdges|
            \geq \frac{\eps}{4}|E(b)| \mper
        \end{equation*}

        \item The isolated vertices were originally connected to some square vertices:
        In this case, we have $w(\MVS) \geq 1$.
        Observe that the originally connecting square vertices must contribute at least 2 deleted edges.
        Thus, $|\DeletedEdges| \geq 2 |\Siso| + 2$ and $|\DeletedEdges| \geq 4|\Ciso|$, and we have
        \begin{equation*}
            \varphi(b) + |\DeletedEdges| \geq -1 + \frac{\eps}{4}|E(b)| + \frac{1}{2} - \frac{|\DeletedEdges|-2}{2} - (2-\eps) \frac{|\DeletedEdges|}{4}+ |\DeletedEdges|
            \geq \frac{\eps}{4}|E(b)| \mper
        \end{equation*}
    \end{enumerate}
    In both cases, the collapsed shape is negligible.
    This completes the proof.
\end{proof}

Finally, we handle the disconnected shapes with an additional disconnected component.
Consider $a = (a_1, a_2)$ where $a_1$ is a left one-sided shape and $a_2$ is the disconnected component.
Note that $a_2$ must be a shape in $\ValidShapes$ and $M_{a_2}$ is a scalar which is negligible due to \pref{lem:shape_phi}.
Moreover, $\lambda_{(a_1,a_2)} = \lambda_{a_1}\lambda_{a_2}$.
The matrix $M_a$ can be written as $M_a = M_{a_1} M_{a_2} - \CollapsedErrorShapes{a_1,a_2}$, where $\CollapsedErrorShapes{a_1,a_2}$ consists of shapes obtained by collapsing $a_1, a_2$.
Now summing up all possible $a_2$'s for a fixed $a_1$, we get
\begin{equation*}
    \lambda_{a_1} M_{a_1} + \sum_{a_2} \lambda_{(a_1,a_2)} M_{(a_1,a_2)} = \lambda_{a_1} M_{a_1} \Paren{1 + \sum_{a_2} \lambda_{a_2} M_{a_2}} - \lambda_{a_1} \sum_{a_2} \lambda_{a_2} \CollapsedErrorShapes{a_1,a_2}
\end{equation*}
Observe that $1+ \sum_{a_2} \lambda_{a_2} M_{a_2}$ is simply $\MomentMatrix_{00}$!
The same procedure can be done for shapes $(a_1,a_2, a_3^\top)$ where $a_1,a_3$ are left one-sided shapes.
This essentially shows that all disconnected components can be absorbed into the shape without that component.

\begin{lemma}[Restatement of \pref{lem:disconnected_shapes_captured}]
    Consider the first column of $\MomentMatrix$: $(\MomentMatrix_{00}, 0, \MomentMatrix_{20})$, and let $v \seteq (1, 0, \frac{\MomentMatrix_{20}}{\MomentMatrix_{00}})$.
    The matrix $\MomentMatrix_{00} \cdot vv^\top$ captures all disconnected shapes in $\MomentMatrix_{22}$ modulo some error consisting of negligible shapes.
\end{lemma}

\begin{proof}
    Let $\LeftOneSidedShapes, \RightOneSidedShapes$ be the set of left and right one-sided shapes in $\ValidShapes$.
    Also, note that $\MomentMatrix_{00} = 1+o(1)$.
    From the discussion above, we can write the vector
    \begin{equation*}
        \MomentMatrix_{20} = \MomentMatrix_{00} \sum_{a\in \LeftOneSidedShapes} \lambda_a M_a + \sum_{a\in \LeftOneSidedShapes} \sum_{b\in \ValidShapes} \lambda_a \lambda_b \CollapsedErrorShapes{a,b} \mper
    \end{equation*}
    Similarly, the sum of disconnected shapes in $\MomentMatrix_{22}$ can be written this way:
    \begin{equation*}
        \MomentMatrix_{00} \sum_{a_1, a_2\in \LeftOneSidedShapes} \lambda_{a_1} \lambda_{a_2} M_{a_1,a_2} + \sum_{a_1,a_2\in \LeftOneSidedShapes} \sum_{b\in \ValidShapes} \lambda_{a_1} \lambda_{a_2} \lambda_b \CollapsedErrorShapes{(a_1,a_2),b} \mcom
    \end{equation*}
    where the first term
    \begin{equation*}
        \sum_{a_1, a_2\in \LeftOneSidedShapes} \lambda_{a_1} \lambda_{a_2} M_{a_1,a_2}
        = \Paren{ \sum_{a\in \LeftOneSidedShapes} \lambda_a M_a } \Paren{ \sum_{a\in \LeftOneSidedShapes} \lambda_a M_a }^\top + \CollapsedError
    \end{equation*}
    with $\CollapsedError$ consists of negligible collapsed shapes due to \pref{lem:collapsed_shapes_negligible}.

    Then, consider the first column of $\MomentMatrix$: $(\MomentMatrix_{00}, 0, \MomentMatrix_{20})$, and let $v \seteq (1, 0, \frac{\MomentMatrix_{20}}{\MomentMatrix_{00}})$.
    Clearly, the matrix $\MomentMatrix_{00} \cdot vv^\top$ captures all disconnected shapes in $\MomentMatrix_{22}$ modulo some negligible collapsed shapes.
\end{proof}


\subsection{Truncation error is small}
\label{sec:truncation_error}

We first prove that the candidate moment matrix $\MomentMatrix$ given by the \pc method already approximately satisfies the constraints $x^\top \Gs x = 0$ with very small error.
Specifically, we show that $\pE[x^I (x^\top \Gs x)]$ is close to 0 for any $I\subseteq [n]$, $|I| \leq 2$.

Adopting the notation of \cite{GJJ20}, we use $Q\pE$ to denote the error, where $\pE$ is treated as a dimension $\binom{n}{\leq 4}$ vector, $Q$ is a matrix with rows indexed by $(I,\s)$ for $I\subseteq [n], \s\in[m]$, and
\begin{equation*}
    Q\pE(I,\s) \seteq \pEover{x^I (x^\top \Gs x)}.
\end{equation*}
Note that we only work with degree-4 SoS, so $|I| \leq 2$.

\begin{lemma}[Restatement of \pref{lem:error_upper_bound}]
    There exist constants $C,C_1, c_2, c_3 > 0$ such that if $\eps \geq \frac{C\log\log n}{\log n}$ and $\frac{C_1}{\eps} \leq \tau \leq n^{c_2\eps}$, then $\|Q\pE\|_2 \leq n^{-c_3 \eps \tau}$.
\end{lemma}

\begin{proof}
    By \pref{def:pseudoexpectation},
    \begin{equation*} \label{eq:constraint_error}
    \begin{aligned}
        \pEover{x^I (x^\top \Gs x)}
        &= \sum_{i,j\in[n]} \Gsij \pEover{x^I x_i x_j} \\
        &= \sum_{i,j\in[n]} \sum_{\alpha: |\alpha|\leq \tau} \Gsij \cdot h_{\alpha}(\bG) \cdot \frac{1}{\alpha!} \expover{(\bG',\bz)\sim \Pl}{\bz^{I}\bz_i \bz_j h_{\alpha}(\bG')} \mper
    \end{aligned}
    \numberthis
    \end{equation*}
    Next, using the recurrence of Hermite polynomials, we have $x h_k(x) = h_{k+1}(x) + k h_{k-1}(x)$ for all $k\in \N$ (assuming $h_{-1}(x) = 0$).
    For simplicity of presentation, let us single out an entry $(\s,i,j)$ and let $k \seteq \Alphasij$. We look at the terms $h_k(x)$ and $h_{k+2}(x)$,
    \begin{equation*}
    \begin{aligned}
        & \frac{1}{k!} x \cdot h_k(x) h_k(x') + \frac{1}{(k+2)!} x \cdot h_{k+2}(x) h_{k+2}(x') \\
        =\ & \frac{1}{k!} \left(h_{k+1}(x) + kh_{k-1}(x)\right) h_k(x')
        + \frac{1}{(k+2)!} \left(h_{k+3}(x) + (k+2)h_{k+1}(x)\right) h_{k+2}(x') \mper
    \end{aligned}
    \end{equation*}
    Now, the coefficient of $h_{k+1}(x)$ is
    \begin{equation*}
        \frac{1}{(k+1)!}h_{k+1}(x) \cdot ( (k+1)h_k(x') + h_{k+2}(x') )
        = \frac{1}{(k+1)!} h_{k+1}(x) \cdot x' h_{k+1}(x') \mcom
    \end{equation*}
    again using the recurrence of Hermite polynomials.
    Intuitively, this allows us to rewrite a sum of $x h_k(x) h_k(x')$ as a sum of $x' h_k(x')h_k(x)$.
    Thus, \pref{eq:constraint_error} can be rewritten as
    \begin{equation*}
    \begin{aligned}
        Q\pE (I,\s) \seteq \pEover{x^I (x^\top \Gs x)}
        &= \sum_{\beta: |\beta|\leq \tau-1} \frac{h_{\beta}(\bG)}{\beta!} \sum_{i,j\in[n]} \expover{(\bG',\bz)\sim\Pl}{\bz^{I} \bz_i\bz_j \CoeffEntries{\bG'}{\s}{ij} h_{\beta}(\bG') } + \eps_{\tau}(I,\s) \\
        &=   \sum_{\beta: |\beta|\leq \tau-1} \frac{h_{\beta}(\bG)}{\beta!}  \expover{(\bG',\bz)\sim\Pl}{\bz^{I} h_{\beta}(\bG') (\bz \Coeff{\bG'}{\s}\bz^\top) } + \eps_{\tau}(I,\s) \\
        &= \eps_{\tau}(I,\s) \mcom
    \end{aligned}
    \end{equation*}
    here we see the importance of the planted distribution: any $(\bG',\bz)\sim \Pl$ satisfies $\bz \Coeff{\bG'}{\s} \bz^\top = 0$.

    Finally, we analyze the remaining error term $\eps_\tau(I,\s)$.
    Denote $\IndexPlussij{\alpha}\in \N^{m\times n\times n}$ as the index $\alpha$ with the entry $\Alphasij$ incremented by 1.
    Since $|\alpha|$ must be even, $|\beta|$ must be odd, and the error only consists of terms $h_{\beta}(\bG)$ where $\beta = \IndexPlussij{\alpha}$ and $|\alpha| = \tau$.
    \begin{equation*} \label{eq:error}
        \eps_\tau(I,\s) = \sum_{\alpha: |\alpha|=\tau} \sum_{i,j\in[n]} \frac{h_{\IndexPlussij{\alpha}}(\bG)}{\alpha!}  \cdot \expover{\Pl}{\bz^{I+\{i,j\}} h_{\alpha}(\bG')}
        = \sum_{\alpha: |\alpha| = \tau} \sum_{i,j\in[n]} \lambda_{\alpha,I,\{i,j\}} h_{\IndexPlussij{\alpha}}(\bG) \mper
        \numberthis
    \end{equation*}

    First, the magnitude of $\lambda_{\alpha,I,\{i,j\}}$ (recall equation \pref{eq:lambda_explicit}) can be upper bounded by
    \begin{equation*}
        |\lambda_{\alpha,I,\{i,j\}}| \leq n^{-|\alpha| + |I|/2+1} (|\alpha|-1)!!
        \leq n^{-\tau+2} \tau^{\tau/2}
    \end{equation*}
    here we use the fact that $|I| \leq 2$, $|\alpha| =\tau$, and $(2k-1)!! \leq (2k)^k$.

    Next, fix $I,\s,i,j$.
    The quantity $\sum_{|\alpha|=\tau} \lambda_{\alpha,I+\{i,j\}} h_{\IndexPlussij{\alpha}}(\bG)$ is a sum of graph matrices over shapes with $\tau$ edges.
    By \pref{lem:number_of_shapes}, there are at most $\tau^{O(\tau)}$ such shapes.
    Since $(\alpha,I,\{i,j\})$ must satisfy the conditions in \pref{def:valid_shapes} so that $\lambda_{\alpha,I,\{i,j\}}$ is nonzero, we can use \pref{lem:shape_phi} to upper bound
    \begin{equation*}
        \Abs{ \sum_{|\alpha|=\tau} \lambda_{\alpha,I+\{i,j\}} h_{\IndexPlussij{\alpha}}(\bG) } \leq n^{-\frac{\eps \tau}{4} + O(1)} \tau^{O(\tau)} \mper
    \end{equation*}

    Summing over all $i,j$, we get $|\eps_{\tau}(I,\s)| \leq n^{-\Omega(\eps \tau) + O(1)}$ if $\tau \leq n^{c_2\eps}$ for a small enough constant $c_2$.
    Moreover, if $\tau \geq \frac{C_1}{\eps}$ for a large enough constant $C_1$, then $\|\eps_{\tau}\|_2 \leq n^{-\Omega(\eps\tau)}$.
    This requires $\eps \geq \frac{C\log\log n}{\log n}$ for some constant $C$.
    This completes the proof.
\end{proof}

We remark that a result similar to \pref{lem:error_upper_bound} can be also obtained using \cite[Lemma 7.7]{GJJ20}.
In general, due to the \pc method, if the truncation threshold $\tau$ is not too small, then the candidate moment matrix already approximately satisfies all constraints with tiny error.

\subsection{Bounds on the norm and nonzero singular values of \texorpdfstring{$Q$}{Q}}
\label{sec:fixing_pd}

Since $Q\pE = 0$ if and only if $\pE$ exactly satisfies all constraints, the natural ``fix'' is
\begin{equation*}
    \pE_{\text{fix}} \seteq \pE - Q^\top (QQ^\top)^{\dagger} Q \pE \mper
\end{equation*}
$(QQ^\top)^\dagger$ is the \textit{pseudo inverse}.
Clearly, $Q\pE_{\text{fix}} = 0$.

We assume that $\pE$ only contains the even degree monomials since the odd monomials are zero and don't need to be fixed.
Moreover, we assume that the $\Gs$'s are symmetrized so that $\GEntries{\s}{ij} = \GEntries{\s}{ji}$; this has no effect on the results and will greatly simplify the presentation.
The entries $\GEntries{\s}{ij}$ and $\GEntries{\s}{ii}$ will thus have different scaling, but this is only a constant factor difference.

Recall that the rows of $Q$ are indexed by $(\s,I)$ where $|I| = 0$ or $2$.

\paragraph{$|I| = 0$ case.}
We first look at the entries of $Q\pE$ corresponding to $I= \varnothing$:
\begin{equation*}
    Q\pE(\varnothing,\s) = \pE[x^\top \Gs x] = 2\sum_{i < j} \Gsij \pE[x_ix_j] + \frac{1}{n}\sum_{i\in[n]} \GEntries{\s}{ii} \mper
\end{equation*}
Here we use $x_i^2 = \frac{1}{n}$.
We can see that this is same as the analysis in \pref{sec:spider}, and the above can be represented by the matrix $L_2$ (see \pref{fig:L2_shapes}).

\begin{lemma} \label{lem:L_two_min_eigenvalue}
    $\|L_2\| = \wt{O}(n)$ and $L_2 L_2^\top$ has minimum eigenvalue $\Omega(n^2)$.
\end{lemma}

\begin{proof}
    $\|L_2\| = \wt{O}(n)$ is immediate from graph matrix norm bounds.
    For the minimum singular value, observe that since $m \ll n^2$, $L_2$ is a dense rectangular matrix and every entry is independent: $L_2((\s,\varnothing),\{i,j\}) = \Gsij$.
    Standard techniques in random matrix theory (such as an $\eps$-net argument) show that $L_2 L_2^\top$ is full rank and has minimum eigenvalue $\Omega(n^2)$ with high probability.
\end{proof}

\paragraph{$|I|=2$ case.}
Suppose $I = \{k,\ell\}$ with $k \neq \ell$, we have
\begin{equation*}
\begin{aligned}
    Q\pE(I,\s) = \pE[x^I (x^\top \Gs x)] 
    =\ & 2 \sum_{i<j: i\neq j\neq k\neq \ell} \Gsij \pE[x_ix_jx_k x_\ell]
    + \frac{1}{n} \sum_{\stackrel{j: i\neq j \neq \ell}{i=k}} \GEntries{\s}{kj} \pE[x_j x_\ell] \\
    &+ \frac{1}{n} \sum_{i = j \neq k \neq \ell} \GEntries{\s}{ii} \pE[x_k x_\ell]
    + \cdots
\end{aligned}
\end{equation*}
The expansion corresponds to the shapes in \pref{fig:L4_shapes} (first 3 terms are drawn out).
We denote the sum as $L_4$.

\FigDegFourL

Before we dive into the analysis, we first define a shape $a^*$ drawn in \pref{fig:special_shape}.
This shape appears in the expansion of $L_4 L_4^\top$ and will play a crucial role in our analysis.

\begin{definition}[Shape $a^*$]
    We define $a^*$ as the shape drawn in \pref{fig:special_shape}.
    The matrix $M_{a^*}$ has entries
    \begin{equation*}
        M_{a^*}(\{\s_1,i_1,j_1\}, \{\s_2,i_2,j_2\}) = \GEntries{\s_1}{i_2j_2} \GEntries{\s_2}{i_1 j_1} \mcom
    \end{equation*}
    if $\s_1 \neq \s_2$ and $i_1 \neq j_1 \neq i_2 \neq j_2$, and 0 otherwise.
\end{definition}

\FigSpecialShape

\paragraph{Analysis of $L_4 L_4^\top$.}
Let $a_1, a_2, a_3$ be the first three shapes in $L_4$ drawn in \pref{fig:L4_shapes}.
$\|M_{a_1}\| = \wt{O}(n)$, $\|\frac{1}{n}M_{a_2}\| \leq \wt{O}(n^{\frac{1}{2}-\frac{\eps}{2}})$, $\|\frac{1}{n}M_{a_3}\| = \wt{O}(n^{\frac{1}{2}-\frac{\eps}{2}})$, and the rest of the terms have norm $o(1)$.
Thus, our analysis will focus on $M_{a_1}$.
We show the following lemma,

\begin{lemma} \label{lem:LLT}
    There exists a matrix $A_1$ such that
    \begin{equation*}
        L_4 L_4^\top = \Theta(n^2) \cdot \Id + A_1 A_1^\top + M_{a^*} + \calE_2
    \end{equation*}
    where $\|\calE_2\| = \wt{O}(n^{2-\frac{\eps}{2}})$.
\end{lemma}

\begin{proof}
    First, $L_4 L_4^\top$ can be written as
    \begin{equation*}
        L_4 L_4^\top = M_{a_1} M_{a_1}^\top + \calE_1
    \end{equation*}
    where $\|\calE_1\| = \wt{O}(n^{\frac{3}{2}- \frac{\eps}{2}})$.
    Thus, it suffices to analyze the matrix $M_{a_1} M_{a_1}^\top$.

    Let us write out the entries of $M_{a_1}$ explicitly:
    $M_{a_1}(\{\s,k',\ell'\}, \{i,j,k,\ell\}) = \Gsij$ if $\{k',\ell'\} = \{k,\ell\}$ and 0 otherwise.
    In other words, it is nonzero only when $\{k',\ell'\} \subset \{i,j,k,\ell\}$.
    Then, we can write the entries of $M_{a_1} M_{a_1}^\top$ explicitly,

    \begin{equation*}
        (M_{a_1} M_{a_1}^\top) (\{\s_1,i_1,j_1\}, \{\s_2, i_2, j_2\}) =
        \begin{cases}
            \sum_{k\neq \ell \neq i_1 \neq j_1} \GEntries{\s_1}{k\ell} \GEntries{\s_2}{k\ell} & \textnormal{ if $\{i_1,j_1\} = \{i_2,j_2\}$}, \\
            \sum_{k\notin \{i_1,i_2,j_1\}} \GEntries{\s_1}{i_2 k} \GEntries{\s_2}{i_1 k}  & \textnormal{ if $j_1 = j_2 \neq i_1 \neq i_2$}, \\
            \GEntries{\s_1}{i_2 j_2} \GEntries{\s_2}{i_1 j_1} & \textnormal{ if $i_1 \neq i_2 \neq j_1 \neq j_2$}.
        \end{cases}
    \end{equation*}

    The above can be represented as a sum of several graph matrices.
    We split into different cases; each case corresponds to a shape:
    \begin{itemize}
        \item $\s_1= \s_2$ (diagonal blocks of $M_{a_1} M_{a_1}^\top$):
        \begin{itemize}
            \item Case $\{i_1,j_1\} = \{i_2,j_2\}$: for this shape there is an identity component, $M_a = \Theta(n^2) \cdot \Id + \calE$ where the error $\|\calE\| = \wt{O}(n)$.

            \item Case $j_1 = j_2 \neq i_1 \neq i_2$: for this shape $\|M_a\| = \wt{O}(n^{2-\frac{\eps}{2}})$.

            \item Case $i_1 \neq i_2 \neq j_1 \neq j_2$: this shape can be decomposed into a PSD component plus some errors: $M_a = A_1 A_1^\top + \calE$ where $\|\calE\| = \wt{O}(n)$.
        \end{itemize}

        \item $\s_1 \neq \s_2$ (off-diagonal blocks of $M_{a_1} M_{a_1}^\top$):
        \begin{itemize}
            \item Case $\{i_1,j_1\} = \{i_2,j_2\}$: for this shape $\|M_a\| = \wt{O}(n^{2-\frac{\eps}{2}})$.

            \item Case $j_1 = j_2 \neq i_1 \neq i_2$: for this shape $\|M_a\| = \wt{O}(n^{2-\frac{\eps}{2}})$.

            \item Case $i_1 \neq i_2 \neq j_1 \neq j_2$: this shape is exactly $a^*$ in \pref{fig:special_shape}.
        \end{itemize}
    \end{itemize}

    Therefore, we can write
    \begin{equation*}
        M_{a_1} M_{a_1}^\top = \Theta(n^2) \cdot \Id + A_1 A_1^\top + M_{a^*} + \calE_2
    \end{equation*}
    where $\|\calE_2\| = \wt{O}(n^{2-\frac{\eps}{2}})$.
    This completes the proof.
\end{proof}

\begin{remark}
    We will later show that $L_4 L_4^\top$ has a non-trivial null space.
    Thus, the shape $a^*$ must exist in the expansion of $M_{a_1} M_{a_1}^\top$; without it, $L_4 L_4^\top$ would be full rank, which is a contradiction.
\end{remark}

\paragraph{Null space of $L_4 L_4^\top$.}
Consider any $\pE$ and fix $\s_1 < \s_2 \in[m]$.
Observe that
\begin{equation*}
    \pE \sum_{k,\ell} (x^\top \G{\s_1} x) \GEntries{\s_2}{k\ell} x_k x_\ell - \pE \sum_{k,\ell} (x^\top \G{\s_2} x) \GEntries{\s_1}{k\ell} x_k x_\ell = 0 \mper
\end{equation*}
Treating $\pE$ as a vector, this can be written as $\pE^\top L_4^\top N_{\s_1,\s_2} = 0$.
Since this holds for all vectors $\pE$, $N_{\s_1,\s_2}$ is in the null space of $L_4^\top$.
Collecting the vectors for all pairs $\s_1< \s_2$, we get a matrix $N$ such that $L_4^\top N = 0$.

Similar to the analysis of $L_4$, we look at the dominating component $M_{b_1}$ of $N$; $M_{b_1}$ has norm $\wt{O}(n)$ whereas the other term has norm $\wt{O}(n^{\frac{1}{2}-\frac{\eps}{2}})$.
The rows of $M_{b_1}$ are indexed by $\{\s,i,j\}$ and the columns are indexed by $\{\s_1,\s_2\}$:
\begin{equation*}
    M_{b_1}(\{\s,i,j\}, \{\s_1,\s_2\}) = 
    \begin{cases}
        \GEntries{\s_2}{ij} & \textnormal{if $\s=\s_1$,}\\
        -\GEntries{\s_1}{ij} & \textnormal{if $\s=\s_2$,}\\
        0  & \textnormal{otherwise.}
    \end{cases}
\end{equation*}

Next, we prove the following result for $NN^\top$,
\begin{lemma} \label{lem:null_space_NNT}
    There exists a matrix $A_2$ such that
    \begin{equation*}
        NN^\top = A_2 A_2^\top - M_{a^*} + \calE_3
    \end{equation*}
    where $\|\calE_3\| \leq \wt{O}(n^{2-\frac{\eps}{2}})$.
\end{lemma}

\begin{proof}
    It suffices to consider $M_{b_1} M_{b_1}^\top$.
    \begin{equation*}
        (M_{b_1} M_{b_1}^\top) \Paren{ \{\s_1,i_1,j_1\}, \{\s_2,i_2,j_2\}} = 
        \begin{cases}
            - \GEntries{\s_2}{i_1 j_1} \GEntries{\s_1}{i_2 j_2} & \textnormal{if $\s_1 \neq \s_2$} \\
            \sum_{\s_3\neq \s_1}\GEntries{\s_3}{i_1 j_1} \GEntries{\s_3}{i_2 j_2} & \textnormal{if $\s_1 = \s_2$} \\
        \end{cases}
    \end{equation*}
    We can also write $M_{b_1} M_{b_1}^\top$ as a sum of graph matrices:

    \begin{itemize}
        \item $s = s_1 = s_2$ (diagonal blocks of $M_{b_1} M_{b_1}^\top$): it is clear that $\sum_{\s_3\neq \s} \GEntries{\s_3}{i_1 j_1} \GEntries{\s_3}{i_2 j_2}$ is a PSD component.
        Thus, we can write this component as $A_2 A_2^\top$ for some matrix $A_2$.

        \item $\s_1 \neq \s_2$ (off-diagonal blocks of $M_{b_1} M_{b_1}^\top$):
        \begin{itemize}
            \item Case $i_1 \neq i_2 \neq j_1 \neq j_2$: this shape is exactly $a^*$ but with a crucial negative sign.
            \item Other cases: these shapes have norms bounded by $\wt{O}(n^{2-\frac{\eps}{2}})$.
        \end{itemize}
    \end{itemize}
    Thus, we have
    \begin{equation*}
        NN^\top = A_2 A_2^\top - M_{a^*} + \calE_3
    \end{equation*}
    where $\|\calE_3\| \leq \wt{O}(n^{2-\frac{\eps}{2}})$.
\end{proof}

\paragraph{Proof of \pref{lem:Q_norm_singular_value}.}
Combining \pref{lem:LLT} and \pref{lem:null_space_NNT}, we see that the term $M_{a^*}$ cancels out.
This implies that $L_4 L_4^\top + N N^\top$ is full rank and has minimum eigenvalue $\Omega(n^2)$.
Now, we are ready to prove \pref{lem:Q_norm_singular_value}.

\begin{lemma}[Restatement of \pref{lem:Q_norm_singular_value}]
    There exists a constant $C$ such that for $\eps \geq \frac{C\log\log n}{\log n}$,
    $\|Q\| \leq \wt{O}(n)$ and the smallest nonzero eigenvalue of $QQ^\top$ is $\Omega(n^2)$.
\end{lemma}

\begin{proof}
    $Q = L_2 + L_4$.
    By the graph matrix norm bounds, we have $\|Q\| \leq \wt{O}(n)$.

    Next, we lower bound the minimum eigenvalue of $QQ^\top$.
    Observe that
    \begin{equation*}
        QQ^\top =
        \begin{bmatrix}
            L_2 L_2^\top & L_2 L_4^\top \\
            L_4 L_2^\top & L_4 L_4^\top \\
        \end{bmatrix} \mper
    \end{equation*}
    For $L_2 L_2^\top$, \pref{lem:L_two_min_eigenvalue} shows that it has minimum eigenvalue $\Omega(n^2)$.

    For $L_4L_4^\top$, by \pref{lem:LLT} and \pref{lem:null_space_NNT} we have
    \begin{equation*}
        L_4 L_4^\top + NN^\top = \Theta(n^2) \cdot \Id + A_1 A_1^\top + A_2 A_2^\top + \calE_4 \mcom
    \end{equation*}
    where $\|\calE_4\| \leq \wt{O}(n^{2-\frac{\eps}{2}})$.
    This means that $L_4 L_4^\top + NN^\top$ is full rank and has minimum eigenvalue $\Omega(n^2)$.

    For the off-diagonal block $L_2 L_4^\top$, although both $\|L_2\|$ and $\|L_4\| = \wt{O}(n)$, note that $L_2$ and $M_{a_1}$ (the dominating component of $L_4$) have disjoint rows and columns in $Q$, meaning that $L_2 M_{a_1}^\top$ does not contribute to $L_2 L_4^\top$.
    Then, since $\|L_4 - M_{a_1}\| \leq \wt{O}(n^{\frac{1}{2}- \frac{\eps}{2}})$, we have $\|L_2 L_4^\top\| \leq \wt{O}(n^{\frac{3}{2}-\frac{\eps}{2}})$.

    We have shown that $QQ^\top$ plus an orthogonal matrix is full rank and has minimum eigenvalue $\Omega(n^2)$.
    This implies that the minimum nonzero eigenvalue of $QQ^\top$ is $\Omega(n^2)$.
    This completes the proof.
\end{proof}

\end{document}